\documentclass[12pt]{article}
\usepackage[english]{babel}

\usepackage{fullpage}
\usepackage{cite}

\usepackage{booktabs}
\usepackage{multirow}
\usepackage{makecell}
\usepackage{xcolor}
\usepackage{pifont}
\usepackage{subfigure}

\usepackage{newpxtext,newpxmath}

\let\coloneqq\relax

\usepackage[utf8]{inputenc}
\usepackage{amsthm}
\usepackage{amssymb}
\usepackage{amsmath}
\usepackage{bbold}
\usepackage{bbm}
\usepackage[pdftex, backref=page]{hyperref}
\usepackage{braket}
\usepackage{dsfont}
\usepackage{mathdots}
\usepackage{mathtools}
\usepackage{enumerate}
\usepackage[shortlabels]{enumitem}
\usepackage{csquotes}
\usepackage{stmaryrd}
\usepackage[cal=boondox]{mathalfa}
\usepackage{graphicx}
\usepackage{stackengine}
\usepackage{scalerel}
\usepackage{tensor}       
\usepackage{array}
\usepackage{makecell}
\newcolumntype{x}[1]{>{\centering\arraybackslash}p{#1}}
\usepackage{tikz}
\usepackage{pgfplots}
\usetikzlibrary{shapes.geometric, shapes.misc, positioning, arrows, arrows.meta, decorations.pathreplacing, decorations.pathmorphing, patterns, angles, quotes, calc}
\usepackage{booktabs}
\usepackage{xfrac}
\usepackage{siunitx}
\usepackage{centernot}
\usepackage{comment}
\usepackage{chngcntr}
\usepackage{caption}
\usepackage{subcaption}

\newtheorem{thm}{Theorem}
\newtheorem*{thm*}{Theorem}
\newtheorem{prop}[thm]{Proposition}
\newtheorem*{prop*}{Proposition}
\newtheorem{lemma}[thm]{Lemma}
\newtheorem*{lemma*}{Lemma}
\newtheorem{cor}[thm]{Corollary}
\newtheorem*{cor*}{Corollary}

\newtheorem*{cj*}{Conjecture}
\newtheorem{Def}[thm]{Definition}
\newtheorem*{Def*}{Definition}

\newtheorem*{question*}{Question}

\newtheorem*{problem*}{Problem}
\newtheorem{assumption}[thm]{Assumption}
\newtheorem*{assumption*}{Assumption}

\makeatletter
\def\thmhead@plain#1#2#3{%
  \thmname{#1}\thmnumber{\@ifnotempty{#1}{ }\@upn{#2}}%
  \thmnote{ {\the\thm@notefont#3}}}
\let\thmhead\thmhead@plain
\makeatother

\theoremstyle{definition}
\newtheorem{rem}[thm]{Remark}
\newtheorem*{note}{Note}

\newcommand{\bb}{\begin{equation}\begin{aligned}\hspace{0pt}}
\newcommand{\bbb}{\begin{equation*}\begin{aligned}}
\newcommand{\ee}{\end{aligned}\end{equation}}
\newcommand{\eee}{\end{aligned}\end{equation*}}
\newcommand*{\coloneqq}{\mathrel{\vcenter{\baselineskip0.5ex \lineskiplimit0pt \hbox{\scriptsize.}\hbox{\scriptsize.}}} =}

\newcommand{\eqt}[1]{\stackrel{\mathclap{\scriptsize \mbox{#1}}}{=}}
\newcommand{\leqt}[1]{\stackrel{\mathclap{\scriptsize \mbox{#1}}}{\leq}}

\newcommand{\geqt}[1]{\stackrel{\mathclap{\scriptsize \mbox{#1}}}{\geq}}
\newcommand{\ketbra}[1]{\ket{#1}\!\!\bra{#1}}

\newcommand{\sumno}{\sum\nolimits}

\newcommand{\e}{\varepsilon}
\renewcommand{\epsilon}{\varepsilon}

\newcommand{\ludo}[1]{{\color{Blues5seq5} #1}}

\newcommand{\id}{\mathds{1}}

\newcommand{\ve}{\varepsilon}

\DeclareMathOperator{\Tr}{Tr}
\DeclareMathOperator{\rk}{rk}

\DeclareMathOperator{\co}{conv}

\DeclareMathAlphabet{\pazocal}{OMS}{zplm}{m}{n}

\DeclareMathOperator{\supp}{supp}

\newcommand{\HH}{\mathcal{H}}

\newcommand{\EE}{\pazocal{E}}

\newcommand{\XX}{\mathcal{X}}

\newcommand{\PP}{\pazocal{P}}

\newcommand{\lsmatrix}{\left(\begin{smallmatrix}}
\newcommand{\rsmatrix}{\end{smallmatrix}\right)}

\newcommand{\rel}[3]{#1\big(#2\,\big\|\,#3\big)}
\newcommand{\Rel}[3]{#1\Big(#2\,\Big\|\,#3\Big)}

\stackMath

\stackMath

\makeatletter
\newcommand*\rel@kern[1]{\kern#1\dimexpr\macc@kerna}
\newcommand*\widebar[1]{%
  \begingroup
  \def\mathaccent##1##2{%
    \rel@kern{0.8}%
    \overline{\rel@kern{-0.8}\macc@nucleus\rel@kern{0.2}}%
    \rel@kern{-0.2}%
  }%
  \macc@depth\@ne
  \let\math@bgroup\@empty \let\math@egroup\macc@set@skewchar
  \mathsurround\z@ \frozen@everymath{\mathgroup\macc@group\relax}%
  \macc@set@skewchar\relax
  \let\mathaccentV\macc@nested@a
  \macc@nested@a\relax111{#1}%
  \endgroup
}

\counterwithin*{equation}{part}
\counterwithin*{thm}{part}
\counterwithin*{figure}{part}

\tikzset{meter/.append style={draw, inner sep=10, rectangle, font=\vphantom{A}, minimum width=30, line width=.8, path picture={\draw[black] ([shift={(.1,.3)}]path picture bounding box.south west) to[bend left=50] ([shift={(-.1,.3)}]path picture bounding box.south east);\draw[black,-latex] ([shift={(0,.1)}]path picture bounding box.south) -- ([shift={(.3,-.1)}]path picture bounding box.north);}}}
\tikzset{roundnode/.append style={circle, draw=black, fill=gray!20, thick, minimum size=10mm}}
\tikzset{squarenode/.style={rectangle, draw=black, fill=none, thick, minimum size=10mm}}

\definecolor{Blues5seq1}{RGB}{239,243,255}
\definecolor{Blues5seq2}{RGB}{189,215,231}
\definecolor{Blues5seq3}{RGB}{107,174,214}
\definecolor{Blues5seq4}{RGB}{49,130,189}
\definecolor{Blues5seq5}{RGB}{8,81,156}

\definecolor{Greens5seq1}{RGB}{237,248,233}
\definecolor{Greens5seq2}{RGB}{186,228,179}
\definecolor{Greens5seq3}{RGB}{116,196,118}
\definecolor{Greens5seq4}{RGB}{49,163,84}
\definecolor{Greens5seq5}{RGB}{0,109,44}

\definecolor{Reds5seq1}{RGB}{254,229,217}
\definecolor{Reds5seq2}{RGB}{252,174,145}
\definecolor{Reds5seq3}{RGB}{251,106,74}
\definecolor{Reds5seq4}{RGB}{222,45,38}
\definecolor{Reds5seq5}{RGB}{165,15,21}

\allowdisplaybreaks

\usepackage[most,breakable]{tcolorbox}
\newenvironment{boxedthm}[1]%
	{\expandafter\ifstrequal\expandafter{#1}{orange}{\begin{tcolorbox}[colback=red!15,colframe=orange!15,breakable,enhanced]}{\begin{tcolorbox}[colback=Blues5seq1,colframe=Blues5seq5,breakable,enhanced]}}%
	{\end{tcolorbox}}

\newenvironment{boxedprop}[1]%
	{\expandafter\ifstrequal\expandafter{#1}{orange}{\begin{tcolorbox}[colback=red!15,colframe=orange!15,breakable,enhanced]}{\begin{tcolorbox}[colback=white,colframe=Aquamarine,breakable,enhanced]}}%
	{\end{tcolorbox}}

\newenvironment{boxedstep}[1]%
	{\expandafter\ifstrequal\expandafter{#1}{orange}{\begin{tcolorbox}[colback=red!15,colframe=orange!15,breakable,enhanced]}{\begin{tcolorbox}[colback=white,colframe=Apricot,breakable,enhanced]}}%
	{\end{tcolorbox}}

\newenvironment{boxedassumpt}[1]%
	{\expandafter\ifstrequal\expandafter{#1}{orange}{\begin{tcolorbox}[colback=red!15,colframe=orange!15,breakable,enhanced]}{\begin{tcolorbox}[colback=white,colframe=SeaGreen,breakable,enhanced]}}%
	{\end{tcolorbox}}

\usepackage{authblk}
\usepackage[dvipsnames]{xcolor}

\newcommand{\RR}{\pazocal{R}}

\renewcommand{\EE}[1]{\underset{\scaleobj{.8}{#1}}{\mathds{E}\,}}
\renewcommand{\PP}[1]{\underset{\scaleobj{.8}{#1}}{\mathds{P}\,}}

\allowdisplaybreaks 
\usepackage{graphicx,import,xcolor,transparent} 
\usepackage{centernot}

\newcommand{\modifica}[1]{ #1} 

\usepackage{caption} 
\newcommand{\mycaption}[2]{\caption[#1]{%
{\centering #1 \par}\vspace{0.4em}
        \normalfont #2}
        }

\usepackage{subcaption}
\usepackage{tabularray}
\UseTblrLibrary{booktabs} 
\usepackage{pifont}

\numberwithin{equation}{section}

\let\ludo\relax

\begin{document}

\title{Generalised quantum Stein's lemma \\ more robust than ever}

\author[1]{Filippo Girardi\thanks{filippo.girardi@sns.it}}
\author[1]{Kuan-Yi Lee\thanks{kuanyi.lee@sns.it}}
\author[2,3,4]{\\ Masahito Hayashi}
\author[1]{Ludovico Lami}
\affil[1]{Scuola Normale Superiore, Piazza dei Cavalieri 7, 56126 Pisa, Italy}
\affil[2]{School of Data Science, The Chinese University of Hong Kong,
Shenzhen,\newline Longgang District, Shenzhen, 518172, China}
\affil[3]{International Quantum Academy (SIQA), Futian District, Shenzhen 518048, China}
\affil[4]{Graduate School of Mathematics, Nagoya University, Nagoya, 464-8602, Japan}

\date{}

\setcounter{Maxaffil}{0}
\renewcommand\Affilfont{\itshape\small}

\maketitle
\begin{abstract}
The generalised quantum Stein's lemma is a key result in quantum hypothesis testing, and connects this fundamental primitive of quantum information processing with quantum resource manipulation, a task that is central for technological applications. Prior works have proved this statement in the idealised setting of independent and identically distributed (i.i.d.)\ sequences of quantum states, and recent extensions consider also sources that are `close' to i.i.d., according to the strict notion put forth by Mazzola, Sutter, and Renner. For several applications, however, one would need to consider yet more general sources. We establish a version of the generalised quantum Stein's lemma that is conceptually much simpler and general, as it applies to any source that is asymptotically close to an i.i.d.\ state with respect to the normalised quantum Wasserstein distance of order 1. As an immediate consequence, we solve the Stein exponent of a scenario where the null hypothesis is arbitrarily varying, expressing it in terms of i.i.d.\ Stein exponents corresponding to arbitrary states in the convex hull of the null hypothesis base set.
\end{abstract}

\newpage
\tableofcontents
\newpage

\section{Introduction}

Quantum hypothesis testing is one of the most fundamental tasks in quantum information theory. In its standard independent and identically distributed (i.i.d.) setting, given $n$ copies of an unknown system, one must decide whether it was prepared as $\rho^{\otimes n}$ (null hypothesis) or as $\sigma^{\otimes n}$ (alternative hypothesis). Under any fixed constraint on the type I error $\epsilon \in (0,1)$, the quantum Stein lemma characterises the optimal exponent of the type II error with the Umegaki relative entropy $D(\rho\|\sigma)$~\cite{Hiai1991,Ogawa2000}.

A central generalisation replaces the single i.i.d.\ alternative by a sequence of composite alternative sets $(\mathcal S^{(n)})_{n\geq 1}$ (see Section~\ref{sec:intro_GQSL} for a technical introduction).
Under suitable structural assumptions, the resulting generalised quantum Stein lemma (GQSL) identifies the optimal type II error exponent against the set of composite alternative with the regularised relative entropy $D^{\infty}(\rho\|\mathcal S)$~\cite{Hayashi2025, Lami_2025}. A paradigmatic consequence of the GQSL is the reversibility of a wide class of resource theories -- including entanglement theory -- under asymptotically resource non-generating operations~\cite{BrandaoPlenio2, Hayashi2025, Lami_2025}.  

The GQSL was originally formulated and claimed by Brandão and Plenio~\cite{Brandao2010}. After a gap was identified in the original argument~\cite{gap}, following different approaches two rigorous proofs were obtained independently by Hayashi and Yamasaki~\cite{Hayashi2025} and by Lami~\cite{Lami_2025}. Recently, Mazzola, Sutter and Renner~\cite{Mazzola_2026_2} recovered the original proof by closing the gap.\smallskip

The original formulation of the GQSL assumes an exact i.i.d.\ structure for the null hypothesis, covering the idealised case of perfectly stable and memoryless source, while realistic procedures may in general suffer from correlations and defects among the copies of the systems~\cite{Mazzola_2026, almost_iid}. This raises a natural question.
\begin{center}
    \textit{How far can the i.i.d.\ assumption on the null hypothesis be relaxed without altering the Stein exponent, i.e.\ without compromising the performance on the detection of the underlying hypothesis?}
\end{center} 
A first restricted robustness result was obtained by Lami~\cite{Lami_2025}, who showed the exact GQSL exponent remains unchanged when the i.i.d.\ null hypothesis is replaced by an constant-size defect state~\cite{Brandao2010} with a fixed total number of defects. This result foreshadowed a systematic development to more general almost i.i.d.\ sources.

Three notions have since emerged: first, Mazzola--Sutter--Renner~\cite{Mazzola_2026} provided a generalisation of constant-size defect states, which we will refer to as MSR model; then Wasserstein and weakly almost i.i.d.\ sources were introduced~\cite{almost_iid} as broader notions of deviations from the i.i.d.\ regime. These notions provide progressively weaker forms of local and global control of the deviations of a sequence of states $(\rho_n)_{n\geq 1}$ from the idealised i.i.d.\ structure $(\rho^{\otimes n})_{n\geq 1}$ and lead to different operational guarantees correspondingly (see, e.g.,~\cite{a_tale}).\smallskip

This hierarchy of models of sources has only partially been explored in the setting of composite hypothesis testing problems yet. For the MSR model, a robust version of the GQSL has been established under the Brand\~ao-Plenio axioms on the alternative hypothesis~\cite{Mazzola_2026_2}. For the broader class of Wasserstein and weakly almost i.i.d.\ sources, the robustness has only been established in the setting where the alternative hypothesis is i.i.d.\ (or a MSR deviation from the i.i.d.\ setting)~\cite{a_tale,datta2026entropyconcentrationuniversaltypicality}; more generally, it was shown that it is possible to design tests that can attain the exact exponent of the corresponding i.i.d.\ problem without knowing the specific deviation from the i.i.d.\ regime (such tests are called \emph{universal} within a class of almost i.i.d.\ sources).

As such, the distinction between the robustness for a specific source and universality of unknown source is operationally important. An individual source statement allows the tests to dependent on the actual null sequence $(\rho_n)_{n\geq 1}$; by contrast, in the universal testing one, the null source is only known to belong to a prescribed family, and a single test must satisfy the type I error constraint for all states in that family.

In this work, we establish a universal GQSL for a `Wasserstein equiconvergent null source,' namely, a \ludo{source of null hypotheses that approximate uniformly $\rho^{\otimes n}$, as \mbox{$n\to\infty$}, with respect to the} normalised Wasserstein distance. For every such null family, we show the existence of universal tests achieving the same regularised error exponent of the standard GQSL. Therefore, the ignorance of the precise null source within the \ludo{equiconvergent} family
\ludo{causes no performance loss} in the asymptotic exponent. Robustness for individual Wasserstein almost i.i.d.\ source follows thereafter as a special case (Section~\ref{sec:GQSL}).
In addition, we show that the Wasserstein robust GQSL has immediate implications in other composite hypothesis testing problems, such as the compound i.i.d.\ and the arbitrarily-varying settings (Section~\ref{sec:applications}).

\subsection{Preliminary notation}
Let $\mathcal H\cong\mathbb C^d$ be a finite-dimensional
Hilbert space, and, for every $n\in\mathbb N$, let $[n]:=\{1,\ldots,n\},~~\mathcal H_n:=\mathcal H^{\otimes n}$ For a subset $S\subseteq[n]$, we write
\bb
    S^c\coloneqq[n]\setminus S,
    \qquad
    \mathcal H_S\coloneqq\bigotimes_{i\in S}\mathcal H_i.
\ee

We denote by $\mathcal L(\mathcal H_n)$ the space of linear operators on $\mathcal H_n$; whenever $A_S\in\mathcal L(\mathcal H_S)$, we note $A_S$ with the operator on $\mathcal H_n$ that acts as $A_S$ on the subsystems indexed by $S$ and as the identity on $S^c$. Thus, for $|\psi\rangle\in\mathcal H_n$, we write
\bb
    A_S|\psi\rangle \coloneqq \bigl(A_S\otimes \id_{S^c}\bigr)|\psi\rangle,
\ee
where the tensor factors are understood in their ambient order in $[n]$.

Furthermore, for a self-adjoint operator $X$ and $a\in\mathbb R$, we denote by
\bb
    \{X\geq a\} \coloneqq \id_{[a,\infty)}(X),
    \qquad
    \{X>a\} \coloneqq \id_{(a,\infty)}(X)
\ee
the corresponding spectral projectors. We use the analogous
notation for "$\le,~<$" parts.

For self-adjoint operators $A$ and $B$, we write
\bb
    \{A\geq B\} \coloneqq \{A-B\geq0\},
    \qquad
    \{A>B\} \coloneqq \{A-B>0\}.
\ee
For instance, $\{\rho_n\geq e^{nr}\sigma_n\}$ denotes the projector onto the non-negative spectral subspace of $\rho_n-e^{nr}\sigma_n$. Outside braces, $A\geq B$ denotes the usual Löwner order.

\subsection{Generalised quantum Stein's lemma: an axiomatic framework}\label{sec:intro_GQSL}

In order to formally introduce the rigorous statement of the generalised quantum Stein's lemma (GQSL) in the case of an i.i.d.\ null hypothesis, we need some preliminary notation and definitions, which we will also frequently use in the remainder of the paper.
For a composite alternative hypothesis \(\mathcal S^{(n)}\), define the type II hypothesis testing error as
\begin{align}
\rel{\beta_\epsilon}{\rho_n}{\mathcal S^{(n)}}
\coloneqq
\min_{0\le E\le \id}
\left\{ \max_{\sigma_n\in\mathcal S^{(n)}}\Tr[E\sigma_n]: \Tr[(\id-E)\rho_n]\le\epsilon
\right\}.
\label{eq:beta-composite}
\end{align}

Later on in the paper, we shall also use a composite version of the null hypothesis.  Let
\(\mathcal F^{(n)}\) be a nonempty set of states on \(\mathcal H^{\otimes n}\).
Define
\begin{align}
\rel{\beta_\epsilon}{\mathcal F^{(n)}}{\mathcal S^{(n)}}
\coloneqq
\min_{0\le E\le \id}
\left\{
    \max_{\sigma_n\in\mathcal S^{(n)}}\Tr[E\sigma_n]:
    \sup_{\rho_n\in\mathcal F^{(n)}}
    \Tr[(\id-E)\rho_n]\le\epsilon
\right\}.
\label{eq:beta-composite-null-and-alt}
\end{align}

Thus the type I error  
constraint is imposed uniformly over the composite null set $\mathcal F^{(n)}$, while the type II error is measured in the worst case over the alternative set \(\mathcal S^{(n)}\).

Let us now recall the definition and the fundamental properties of the hypothesis testing relative entropy. 
For a state $\rho_n$ and a composite alternative $\mathcal S^{(n)}$, set
\bb
\rel{D_H^\epsilon}{\rho_n}{\mathcal S^{(n)}}
= \inf_{\sigma_n \in \mathcal S^{(n)}} D_H^\epsilon(\rho_n\|\sigma_n)\, ;
\ee
As for null set $\mathcal F^{(n)}$ and alternative set $\mathcal S^{(n)}$, define it as
\bb
  \rel{D_H^\epsilon}{\mathcal F^{(n)}}{\mathcal S^{(n)}}
  \coloneqq \inf_{\substack{\rho_n \in \mathcal F^{(n)}\\ \sigma_n \in \mathcal S^{(n)}}} D_H^\epsilon (\rho_n \|\sigma_n)\, .
\ee

Intuitively speaking, the following minimax result guarantees that testing two unknown states taken from two convex sets is at most as difficult as testing the worst pair of \emph{known} states from the sets.

\begin{lemma}[{(Composite hypothesis-testing minimax identity~\cite[Lemma~31]{fang2025generalized})}]\label{lem:minmax_fang}
Let $\mathcal F^{(n)}$ and $\mathcal S^{(n)}$ be nonempty convex sets of states. Then,
\bb
    \rel{\beta_\varepsilon}{\mathcal F^{(n)}}{\mathcal S^{(n)}}
    =
    \sup_{\substack{\rho\in\mathcal F^{(n)}\\
    \sigma\in\mathcal S^{(n)}}}
    \beta_\varepsilon(\rho\|\sigma).
\ee
Consequently,
\bb
    \rel{D_H^\varepsilon}{\mathcal F^{(n)}}{\mathcal S^{(n)}}
    =
    - \log \rel{\beta_\varepsilon}{\mathcal F^{(n)}}{\mathcal S^{(n)}} .
\ee
\end{lemma}

In order to prove the generalised quantum Stein's lemma, we need to make some assumptions on the sequence of families of alternative hypotheses $(\mathcal{S}^{(n)})_{n\geq 1}$, otherwise the statement of interest may not hold (for instance, without convexity of $\mathcal S^{(n)}$ it is possible to identify a counterexample~\cite{Hiai_2009}).

Here we recall some of the main requirements that were considered in the literature -- although not all of them are necessary in order to prove the GQSL.
\begin{enumerate}[left=2.5em]
\item[(CC)]  \emph{Convexity and topological closedness.} For each $n$, \(\mathcal S^{(n)}\) is a compact, convex
    subset of the state space on \(\mathcal H^{\otimes n}\).
\item[(TR)]  \emph{Closedness under tracing out subsystems.} For all $m,n\geq 1$ and for all $\sigma_{n+m}\in \mathcal S^{(n+m)}$, the state $({\rm Id}_n\otimes\Tr_m) [\sigma_{n+m}]$ belongs to $\mathcal S^{(n)}$. 
\item[(TEN)] \emph{Closedness under tensor products.} The family is closed under tensor products: if
    \(\sigma_n\in\mathcal S^{(n)}\) and
    \(\omega_m\in\mathcal S^{(m)}\), then $\sigma_n\otimes\omega_m\in \mathcal S^{(n+m)}$.
\item[(PER)] \emph{Closedness under permutations of subsystems.} For all states \(\sigma_n\in\mathcal S^{(n)}\) and permutations $\pi\in S_n$, the permuted state $U_\pi \sigma_n U_\pi^\dagger$ belongs to \(\mathcal S^{(n)}\).
\item[(FR)] \emph{Existence of a full-rank
state.} There exists a full-rank state \(\sigma_{\rm full}\in\mathcal S^{(1)}\).
\item[(REP)] \emph{Replacer stability.} There exists a full-rank state \(\sigma_{\rm full}\in\mathcal S^{(1)}\)
such that, for every $n$, every \(\sigma_n\in\mathcal S^{(n)}\), and every site \(i\in[n]\), the state
$\pazocal R_i^{(\sigma_{\rm full})}(\sigma_n)$ belongs to $\mathcal S^{(n)}$, where \(\pazocal R_i^{(\sigma_{\rm full})}\) denotes the channel that replaces the
$i$-th tensor factor by $\sigma_{\rm full}$ and leaves all other tensor factors unchanged.
\end{enumerate}

The first five conditions are known as Brand\~ao-Plenio axioms, as they were originally introduced in the initial formulation of the GQSL~\cite{Brandao2010}. 

If one assumes all the five Brand\~ao-Plenio axioms, the replacer stability automatically follows.
Indeed, given $\sigma_n \in \mathcal S^{(n)}$ and fix a site $i\in [n]$, then (1) append a full-rank state $\sigma_n \otimes \sigma_{\rm full} \in \mathcal S^{(n+1)}$, (2) trace out the $i$-th subsystem, and (3) permute the full-rank state $\sigma_{\rm full}$ into the $i$-th position. The resulting state is $\pazocal R_{i}^{\sigma_{\rm full}}(\sigma_n)$.

\begin{Def}
For an $n$-partite state $\rho_n$ on \(\mathcal H^{\otimes n}\), define its relative entropy from the alternative set \(\mathcal S^{(n)}\) by
\bb
  \rel{D}{\rho_n}{\mathcal S^{(n)}}
  \coloneqq
  \min_{\sigma_n\in\mathcal S^{(n)}} D(\rho_n\|\sigma_n),
\ee
the minimum is attained because \(\mathcal S^{(n)}\) is compact and the relative entropy is lower semicontinuous.
For a single-site state \(\rho\) on \(\mathcal H\), define
\bb
D^\infty(\rho\|\mathcal S)
\coloneqq\lim_{n\to\infty} \frac1n \rel{D}{\rho^{\otimes n}}{\mathcal S^{(n)}} ,
\ee
as its regularised relative entropy.
\end{Def}
Note that the limit exists under Assumption~\ref{ass:standing-Sn}. Indeed, setting $a_n \coloneqq D(\rho^{\otimes n}\|\mathcal S^{(n)})$, 
the tensor-product closure of \(\mathcal S\) implies subadditivity, $a_{n+m}\le a_n+a_m$; together with the Fekete's lemma gives $\displaystyle{\lim_{n\to\infty}\tfrac{a_n}{n}=\inf_n\tfrac{a_n}{n}}$.

Now we have all the notation and definitions needed to formally state the generalised quantum Stein's lemma.

\begin{thm}[(Generalised quantum Stein's lemma~\cite{Hayashi2025, Lami_2025})]\label{thm:gqsl_iid}
Suppose that the sequence of alternative hypotheses $\mathcal S=(\mathcal{S}^{(n)})_{n\geq 1}$, with $\mathcal{S}^{(n)}\subseteq\mathcal{D}(\mathcal{H}^{\otimes n})$, satisfies the five Brand\~ao-Plenio axioms. Then~\cite{Lami_2025}
\begin{align}
\lim_{n\to\infty}\frac 1n
\rel{D_H^\epsilon}{\rho^{\otimes n}}{\mathcal S^{(n)}} =
D^{\infty}(\rho \|\mathcal S)\qquad \forall \,\epsilon\in (0,1).
\end{align}
The same result holds if we only assume (CC), (TEN) and (FR)~\cite{Hayashi2025}.
\end{thm}

The main goal of this paper is to prove Theorem~\ref{thm:gqsl_iid} when the i.i.d.\ source $(\rho^{\otimes n})_{n\geq 1}$ is replaced by any arbitrary sequence $(\rho_n)_{n\geq 1}$ such that 
\bb
    \lim_{n\to\infty}\frac 1n \|\rho_n-\rho^{\otimes n}\|_{W_1}=0,
\ee
where $\|\,\cdot\,\|_{W_1}$ is the quantum Wasserstein distance of order 1~\cite{De_Palma_2021}, which we are going to introduce in detail in the following section.

\subsection{The quantum Wasserstein distance of order 1}
The Hamming distance between two strings of fixed length $x^n,y^n\in\mathcal{X}^n$, which counts the number of coordinates in which they differ, is a natural metric for the quantification of the amount of defects that a source of noise produces on an uncorrupted encoding of classical information. Leveraging the theory of optimal transportation, the Hamming distance can be lifted to a metric on the space of probability distributions on fixed-length strings $\mathcal{P}(\mathcal{X}^n)$. More precisely, we are going to consider the extension given by the Wasserstein distance of order 1 -- known as Ornstein’s $\bar d$-distance~\cite{ornstein1973application} -- which has multiple applications in ergodic theory and information theory~\cite{gray2011entropy}. Its crucial feature is its sensitivity to local perturbations: probability distributions concentrated on strings differing only in a few coordinates remain close in $W_1$ distance, even when they may be perfectly distinguishable in total variation distance.

A seminal work by De Palma et al.~\cite{De_Palma_2021} introduced a generalisation to the quantum case -- called the \emph{quantum Wasserstein distance of order 1} -- which found applications in several areas of quantum information theory, including quantum differential privacy~\cite{hirche2023quantumdifferentialprivacyinformation}, equivalence of ensembles in quantum statistical mechanics~\cite{De_Palma_2022,De_Palma_2025}, learning of many-body quantum states~\cite{De_Palma_2024,Rouz__2024,Rouz__2024b}, rapid thermalization for geometrically local Hamiltonians~\cite{Bardet_2024,Kochanowski_2025,bakshi2025dobrushinconditionquantummarkov}, quantum machine learning through quantum generative adversarial networks~\cite{Kiani_2022}, and limitations of variational quantum algorithms~\cite{De_Palma_2023}. Finally, it turned out to be a foundational tool to describe the properties of almost i.i.d.\ states~\cite{almost_iid, Mazzola_2026_2, a_tale}.

As a generalization of the Hamming distance from classical strings to quantum states of multipartite quantum systems, the construction is based on the notion of neighbouring states. Two states $\rho$ and $\sigma$ of the $n$-partite system $A_1\ldots A_n$ are called \emph{neighbouring} whenever there exists a subsystem $A_i$ such that $\Tr_{A_i}\rho=\Tr_{A_i}\sigma$, namely the two states become indistinguishable upon discarding subsystem $A_i$. The convex hull of differences between neighbouring states, seen as a unit ball, uniquely identifies the quantum $W_1$ norm, and the associated quantum $W_1$ distance is the metric induced by this norm. More precisely,~\cite{De_Palma_2021} give the following definition.

\begin{Def}[(Quantum Wasserstein distance of order 1)]
Let $\mathcal{H}_A$ be a finite-dimensional Hilbert space and let $n\geq 1$. Then, for all $\rho,\,\sigma\in \mathcal{D}(\mathcal{H}_A^{\otimes n})$, the \emph{quantum Wasserstein distance of order 1} is defined as
\bb
    \|\rho-\sigma\|_{W_1} \coloneqq \min \Bigg\{ \;\sum_{i=1}^n c_i\quad  &\text{such that}\quad  && c_i\geq 0, \qquad  \rho-\sigma=\sum_{i=1}^n c_i\left(\tau^{(i)}-\eta^{(i)}\right) \\
    & \text{with}&& \tau^{(i)},\eta^{(i)}\in \mathcal{D}(\mathcal{H}^{\otimes n}_A), \quad \Tr_{A_i}\tau^{(i)}=\Tr_{A_i}\eta^{(i)}&\Bigg\}.
\ee
\end{Def}

We recall some simple properties that we will frequently use in this work.
\begin{itemize}
    \item \textbf{Data-processing under single system channels. } The Wasserstein distance of order 1 -- already in the fully classical case -- does not satisfy the data-processing inequality for arbitrary channels. However, if $\Lambda_k$ is a quantum channel acting on a single system $k\in[n]$, then~\cite[Proposition~3]{De_Palma_2021}
    \bb
        \|\Lambda_k(\rho_n)-\Lambda_k(\sigma_n)\|_{W_1}\leq \|\rho_n-\sigma_n\|_{W_1} \qquad \forall \rho_n,\sigma_n\in\mathcal{D}(\mathcal{H}^{\otimes n}).
    \ee
    \item \textbf{Data-processing under permutation twirl.} Since the quantum Wasserstein distance of order 1 is invariant under permutations~\cite[Proposition~3]{De_Palma_2021} and satisfies the triangle inequality, it is simple to show that
    \bb
        \|\pazocal P(\rho_n)-\pazocal P(\sigma_n)\|_{W_1}\leq \|\rho_n-\sigma_n\|_{W_1} \qquad \forall \rho_n,\sigma_n\in\mathcal{D}(\mathcal{H}^{\otimes n}),
    \ee
    where $\pazocal P(\,\cdot\,)\coloneqq \frac 1 {n!}\sum_{\pi\in S_n} U_{\pi}\,\cdot\, U_\pi^\dagger$ is the permutation twirl.
    \item \textbf{Relation with the trace distance and additivity.} In the case $n=1$, the quantum Wasserstein distance of order 1 coincides with the trace distance~\cite[Proposition~2]{De_Palma_2021}:
    \bb
        \|\rho-\sigma\|_{W_1}=\frac 12 \|\rho-\sigma\|_1\qquad \forall \rho,\sigma\in\mathcal{D}(\mathcal{H}).
    \ee
    Moreover, it is additive under tensor products~\cite[Proposition~4]{De_Palma_2021}:
    \bb
        \|\rho_n\otimes\rho_m'-\sigma_n\otimes\sigma_m'\|_{W_1}=\|\rho_n-\sigma_n\|_{W_1}+\|\rho_m'-\sigma_m'\|_{W_1},
    \ee
    for all $\rho_n,\sigma_n\in\mathcal{D}(\mathcal{H}^{\otimes n})$ and $\rho_m',\sigma_m'\in\mathcal{D}(\mathcal{H}^{\otimes m})$. As an immediate consequence, we derive the identity
    \bb\label{eq:W_tr}
        \frac 1n\|\rho^{\otimes n}-\sigma^{\otimes n}\|_{W_1}=\frac 12\|\rho-\sigma\|_1\qquad \forall \rho,\sigma\in \mathcal{D}(\mathcal{H}).
    \ee
\end{itemize}

The original work by De Palma et al. provides a continuity bound of the von Neumann entropy in terms of the quantum Wasserstein distance of order 1~\cite[Theorem 1]{De_Palma_2021}. The optimal continuity bound was then derived by De Palma and Trevisan in a following work as follows.
\begin{thm}[{(Continuity of the von Neumann entropy~\cite[Theorem 9.1]{De_Palma_2023b})}]\label{thm:SW}
For any $n\geq 1$ and any $\rho_n,\,\sigma_n\in\mathcal{D}\left(\mathcal{H}^{\otimes n}\right)$ we have
\begin{equation}
    \frac{1}{n}\left|S(\rho_n) - S(\sigma_n)\right| \le h\left(w_n\right) + w_n\ln\left(\left(\dim\mathcal{H}\right)^2-1\right)\,,\quad \text{where}\quad w_n\coloneqq \frac1n \left\|\rho_n-\sigma_n\right\|_{W_1}
\end{equation}
and $h(x) \coloneqq -x\ln x - \left(1-x\right)\ln\left(1-x\right)$ is the binary entropy function, with $x\in[0,1]$.
\end{thm}

The following crucial continuity bound was proved by Mazzola, Sutter, and Renner.
\begin{prop}[{(Continuity of the relative entropy of resource\cite[Proposition~3.1]{Mazzola_2026_2})}]\label{prop:3.1} Assume that a family of states $\mathcal{S}^{(n)}\subseteq \mathcal{D}(\mathcal{H}^{\otimes n})$ satisfies (CC) and (REP) among the conditions listed in Section~\ref{sec:intro_GQSL}. Then, for all $\rho_n,\rho'_n\in\mathcal{D}(\mathcal{H}^{\otimes n})$ such that $w_n\coloneqq\frac 1n \|\rho_n-\rho'_n\|_{W_1}\leq \frac 12$, we have
    \bb
        \left|\frac{1}{n}\rel{D}{\rho_n}{\mathcal{S}^{(n)}} - \frac{1}{n} \rel{D}{\rho'_n}{\mathcal{S}^{(n)}}\right| \leq 3h_2(w_n)+6w_n\log \frac d{\lambda_{\min}(\sigma_{\rm full})},
    \ee
    where $d\coloneqq \dim \mathcal{H}$, and $\lambda_{\min}(\sigma_{\rm full})>0$ denotes the smallest eigenvalue of $\sigma_{\rm full}$.
\end{prop}

As a simple corollary of Proposition~\ref{prop:3.1}, we immediately derive the continuity of the regularised relative entropy of resource under (CC) and (REP).

\begin{cor}\label{cor:3.1} Assume that a sequence of families of states $(\mathcal{S}^{(n)})_{n\geq 1}$ satisfies (CC) and (REP) among the conditions listed in Section~\ref{sec:intro_GQSL}. Then, if $(\rho_k)_{k\geq 1}\subseteq \mathcal{D}(\mathcal{H})$ converges to $\rho$, i.e.\ $\displaystyle{\lim_{k\to \infty} \|\rho_k-\rho\|_1=0}$, then
\bb
    \lim_{k\to\infty}D^{\infty}(\rho_k\|\mathcal{S})=D^\infty(\rho\|\mathcal{S}).
\ee
\end{cor}
\begin{proof}
    By~\eqref{eq:W_tr} we have
    \bb
        w_k \coloneqq \frac 1n\|\rho_k^{\otimes n}-\rho^{\otimes n}\|_{W_1}=\frac 12 \|\rho_k-\rho\|_1 = o(1)\quad \text{for} \quad k\to \infty,
    \ee
    whence, by Proposition~\ref{prop:3.1},
    \bb
        \lim_{k\to\infty}\big|D^\infty(\rho_k\|\mathcal{S})-D^\infty(\rho\|\mathcal S)\big|&=\lim_{k\to\infty}\lim_{n\to\infty}\left|\frac{1}{n}D(\rho_k^{\otimes n}\|\mathcal{S}^{(n)})- \frac{1}{n}D(\rho^{\otimes n}\|\mathcal{S}^{(n)})\right|\\
        &\leq \lim_{k\to\infty} \left(3h_2(w_k)+6w_k\log \frac d{\lambda_{\min}(\sigma_{\rm full})}\right)=0,
    \ee
    which concludes the proof.
\end{proof}

\subsection{Almost i.i.d.\ sources}

In this work, we are going to study the behaviour of Wasserstein almost i.i.d.\ sources~\cite{almost_iid} in composite asymmetric hypothesis testing. 

\begin{Def}
    Let $\rho\in\mathcal{D}(\mathcal{H})$ be a quantum state on a  Hilbert space $\mathcal{H}$. A sequence $(\rho_n)_n$ of states $\rho_n\in\mathcal{D}(\mathcal{H}^{\otimes n})$ is said a \emph{Wasserstein almost i.i.d.\ source along $\rho$} whenever  
    \bb
        \lim_{n\to\infty}\frac 1 n \|\rho_n-\rho^{\otimes n}\|_{W_1}=0.
    \ee
    We denote by $\pazocal{A}_{\rho}^{W_1}$ the set of all the Wasserstein almost i.i.d.\ sources along $\rho$.
\end{Def}

There are other two notions of almost i.i.d.\ sources that were recently introduced in the literature, but which we will not consider in this paper. We recall them in order to make a brief comparison with Wasserstein almost i.i.d.\ sources. Let $\rho \in \mathcal{D}(\mathcal{H}_A)$ be a reference state, and let $(\rho_n)_{n\geq 1}$ be a sequence of states with $\rho_n \in \mathcal{D}(\mathcal{H}_A^{\otimes n})$.
\begin{itemize}
    \item \textbf{Constant-size defect almost i.i.d.\ sources}~\cite{Brandao2010, RennerPhD}. Let \[\pazocal{V}^n_r(\mathcal{H}_{AE}, \ket{\psi})\coloneqq\{U_\pi(\ket{\psi}^{\otimes n-r} \otimes \ket{\omega^{(r)}}): \pi \in S_n, \ket{\omega^{(r)}} \in \mathcal{H}^{\otimes r}_{AE}\}\quad\text{for}\quad 0\le r\leq n.\] We say that $(\rho_n)_n$ is a \emph{constant-size defect almost i.i.d.\ source along} $\rho$ whenever there exist a fixed, independent of $n$ integer $r$, a purification $\ket{\psi_\rho}_{AE}$ of $\rho$ and, a permutation-invariant purification $\ket{\Gamma_n}_{A^n E^n}$ of $\rho_n^{A^n}$ such that $\ket{\Gamma_n}_{A^n E^n}\in \mathrm{span}\,\pazocal{V}^n_{r}(\mathcal{H}_{AE},\ket{\psi_\rho}_{AE})$ for every $n\geq r$.
    \item \textbf{Mazzola--Sutter--Renner (MSR) almost i.i.d.\ sources}~\cite{Mazzola_2026, almost_iid}. For $r\leq n$, let $\pazocal{V}^n_r(\mathcal{H}_{AE}, \ket{\psi})$ as above. We say that $(\rho_n)_n$ is an \emph{MSR almost i.i.d.\ source along} $\rho$ whenever there exist a purification $\ket{\psi_\rho}_{AE}$ of $\rho$ and an extension $\rho_n^{A^n E^n}$ of $\rho_n^{A^n}$ such that $\supp(\rho_n^{A^n E^n}) \subseteq  \mathrm{span}\,\pazocal{V}^n_{r_n}(\mathcal{H}_{AE},\ket{\psi_\rho}_{AE})$ for a sequence $(r_n)_n$ of integers $r_n\leq n$ satisfying $\displaystyle{\lim_{n\to\infty}\tfrac {r_n}n=0}$.
    \item \textbf{Weakly almost i.i.d.\ sources}~\cite{almost_iid}. We say that $(\rho_n)_n$ is a \emph{weakly almost i.i.d.\ source along $\rho$} if  
    \[
        \limsup_{n\to \infty} \EE{\substack{I\subseteq [n]\\ |I|=k}} \big\|(\rho_n)_I - \rho^{\otimes I}\big\|_1 = 0\qquad \forall k \in \mathbb{N}_+ ,
    \]
    where $I$ is a uniformly random subset of $[n]$ with size $k$.
\end{itemize}

The reason why we do not introduce a further notion of almost i.i.d.\ sources according to the trace distance, namely
\bb
    \lim_{n\to\infty} \left\|\rho_n-\rho^{\otimes n}\right\|_1 = 0\,,
\ee
is that such a notion of convergence, on the one hand, is way more restrictive and less physically motivated than the metrisation in terms of the $W_1$ distance~\cite{almost_iid, a_tale}; on the other hand, it does not necessarily provide better guarantees on the robustness of some relevant protocols in quantum Shannon theory~\cite{a_tale}.

It is possible to prove a hierarchical structure among the three notions of almost i.i.d.\ sources above defined. Clearly, every almost-power source is a MSR source. More generally, calling $\pazocal{A}_\rho^{\rm CSD}$, $\pazocal{A}_\rho^{\rm MSR}$, and $\pazocal{A}_\rho^{w}$ the set of all constant-size defect, MSR and weakly almost i.i.d.\ sources, respectively, we have\cite[Section~3]{almost_iid}
\bb
    \pazocal A_\rho^{\rm CSD}\subseteq\pazocal A_\rho^{\rm MSR} \subseteq \pazocal A_\rho^{W_1} \subseteq \pazocal A_\rho^{w}.
\ee

For our purposes, Wasserstein almost i.i.d.\ sources represent the most natural setting to consider.
\begin{itemize}
    \item With respect to the set of MSR almost i.i.d. sources (for which a generalised quantum Stein's lemma was proved in~\cite{Mazzola_2026_2}), $W_1$ sources provide a stronger guarantee on the robustness of composite asymmetric hypothesis testing. The $W_1$ goes beyond the sharp constraint on the tail of defects implied in the definition of MSR sources~\cite[Section~3.3.2]{almost_iid} and allows for a broader class of applications of the robust GQSL as a crucial technical tool in hypothesis testing problems (see Section~\ref{sec:applications}).
    \item Differently from the extension from MSR to $W_1$ sources, the relaxation to weakly almost i.i.d.\ sources would provide a less substantial generalisation of the robustness of composite hypothesis testing. Indeed, the only difference between weakly and Wasserstein almost i.i.d.\ sources is the asymptotic continuity of the entropy~\cite[Section~3.4]{almost_iid}, namely $\displaystyle{\lim_{n\to\infty}\tfrac{S(\rho_n)}{n}=S(\rho)}$. Furthermore, any hypothetical generalisation to the case of weak sources cannot be proved together with a converse inequality, as the second counterexample in~\cite[Remark~15]{a_tale} shows.
\end{itemize}

\subsection{Comparison with previous works}

\begin{table}[t]
    \centering 
    \begin{talltblr}[
     note{a} = {\small This property follows from the other assumptions. },
     note{b} = {\small Note that, by definition, (FR) immediately follows from (REP).},
     caption = {Comparison of the formulations of the generalised quantum Stein's lemma available in the literature and in our work. Note that the first five conditions, from (CC) to (FR), are known as Brand\~ao--Plenio axioms~\cite{Brandao2010}.},
     label = table,
    ]{
        width = \textwidth, 
        colspec = {X[2.8,m,c] X[1.2,m,c] X[0.55,m,c] X[0.55,m,c] X[0.55,m,c] X[0.55,m,c] X[0.55,m,c] X[0.55,m,c]},
    }
        \toprule
         & \vspace{-0.7em}\textbf{Null} & \SetCell[c=6]{c}{\textbf{Alternative hypothesis}} \\
        \textbf{Work} & \textbf{hypothesis} & (CC) & (TR) & (TEN) & (PER) & (FR) & (REP) \\
        \midrule
        Hayashi--Yamasaki 
        ~\cite{Hayashi2025} & i.i.d. & \ding{109} &  & \ding{109} & & \ding{109} & \\
        Lami ~\cite{Lami_2025} & constant-size defect & \ding{109} & \ding{109} & \ding{109} & \ding{109} & \ding{109} & (\ding{109}\TblrNote{a}\;\,) \\
        Mazzola--Sutter--Renner~\cite{Mazzola_2026_2} & MSR & \ding{109} & \ding{109} & \ding{109} & \ding{109} & \ding{109} & (\ding{109}\TblrNote{a}\;\,) \\
        This work (Theorem~\ref{thm:W1_GQSL}) & $W_1$ & \ding{109} &  & \ding{109} &  & (\ding{109}\TblrNote{b}\;\,) & \ding{109} \\
        \bottomrule
    \end{talltblr}
\end{table}

As summarised in Table~\ref{table}, previous literature covers the GQSL in two directions. On one hand, Hayashi--Yamasaki's work~\cite{Hayashi2025} establishes the GQSL under minimal assumptions on the alternative hypothesis (CC, TEN, FR), while the null hypothesis is always restricted to an i.i.d.\ source. On the other hand, Lami~\cite{Lami_2025} and Mazzola--Sutter--Renner~\cite{Mazzola_2026_2} extend the null hypothesis to constant-size defect states and MSR almost i.i.d.\ states, respectively, under more restrictive assumptions on the alternative hypothesis, namely, the five Brand\~ao--Plenio axioms~\cite{Brandao2010}.\smallskip

Theorem~\ref{thm:W1_GQSL} \ludo{combines the best of both worlds, so to speak: it} establishes the GQSL under a stronger request of robustness -- namely, the stability of the Stein's exponent under Wasserstein almost i.i.d.\ perturbations of the null hypothesis -- while only requiring the composite alternative set to satisfy Assumptions~\ref{ass:standing-Sn}, which are convexity and closedness (CC), closure under tensor product (TR), and stability under the action of the replacer channel on any site by a full-rank state $\sigma_{\rm full}\in \mathcal{S}^{(1)}$ (REP) -- which, by definition, automatically implies (FR). As such, our theorem simultaneously extends the admissible null source and weakens the structure assumptions on alternative family among all the robust settings explored so far.
\smallskip

More specifically, our main result in Theorem~\ref{thm:W1_GQSL} has two operationally distinct parts:
\begin{enumerate}
    \item[$\bullet$] \textbf{Individual sources statement.} Given a specific $W_1$ almost i.i.d.\ state $(\rho_n)_{n\geq 1}$ along $\rho$, and the test can be constructed dependent on $(\rho_n)_{n\geq 1}$. Our result shows that the optimal type II exponent of such a task coincides \ludo{with that of the GQSL}, namely, $D^{\infty}(\rho\|\mathcal S)$. Therefore, replacing the exact i.i.d.\ null source by an arbitrary $W_1$ almost i.i.d.\ source leaves the exponent unchanged, on both the achievability and converse.
    \item[$\bullet$] \textbf{Existence of universal tests.} Suppose that, under the null hypothesis, we only have the guarantee that the sequence of states $(\rho_n)_{n\geq 1}$ is any arbitrary $W_1$ almost i.i.d.\ source along a fixed state $\rho$,\footnote{With uniform rate of converge to $\rho^{\otimes n}$ in normalized $W_1$ distance  (see Definition~\ref{def:W1-equiconvergent}).} while the exact nature of the source, i.e.\ the specific sequence $(\rho_n)_{n\geq 1}$ is unknown, as one would expect in any noisy setting. We prove that such a task still attains the exponent $D^{\infty}(\rho\|\mathcal S)$. The resulting test might dependent on the null set $\mathcal F^{(n)}$, alternative set $\mathcal S^{(n)}$, and the type I threshold $\epsilon \in (0,1)$, but not on which particular element inside $\mathcal F^{(n)}$. This provides a sense of the test is universal within an equiconvergent family.
\end{enumerate}
We note that this simultaneous extension is not a direct consequence of the i.i.d.\ proof. The information spectrum of the null state $(\rho_n)_n$ is no longer govern by the standard i.i.d.\ asymptotic equipartition property, and the global pinching is not contractive under normalized $W_1$ distance. Therefore, a new information spectrum and block-converse statements that are stable under $W_1$ almost i.i.d.\ is required.\smallskip

Furthermore, the propositions used in the proof of our main result in order to show Wasserstein stability may also be independent interests (Proposition~\ref{prop:info-spectrum-W1},~\ref{prop:pinching-extension} and~\ref{prop:padded-block-product-strong-converse}). In particular, Proposition~\ref{prop:info-spectrum-W1} establishes a strong generalisation of~\cite[Proposition 21]{almost_iid}. The latter states that, for any arbitrary $W_1$ almost i.i.d.\ source $(\rho_n)_n$ along $\rho$, the the random variable $-\frac 1n\log \rho_n$ converges in expectation (with respect to the spectral measure of $\rho_n$, see Section~\ref{sec:info_spec}) to the von Neumann entropy of $\rho$, namely
\bb
    \lim_{n\to\infty}\EE{\rho_n}\Big[\!-\frac{1}{n}\log\rho_n\Big]\coloneqq \lim_{n\to\infty}-\frac{1}{n}\Tr[\rho_n\log\rho_n]=\lim_{n\to\infty}\frac{S(\rho_n)}n=S(\rho).
\ee
Proposition~\ref{prop:info-spectrum-W1} strengthens the previous result to the convergence in probability (see Section~\ref{sec:info_spec}):
\bb
    -\frac 1n \log \rho_n \xrightarrow{\mathsf P_{\rho_n}} S(\rho)\qquad \text{i.e.}\qquad \lim_{n\to\infty}\Tr\left[\rho_n \left\{
        \left\lvert-\tfrac1n\log \rho_n-S(\rho)\right\rvert>\delta
    \right\}\right]=0
\ee
for all $\delta>0$.\smallskip

Beyond the main robustness theorem, moreover, in Section~\ref{sec:applications} we give two relevant applications of Theorem~\ref{thm:W1_GQSL} to hypothesis testing problems with composite null and alternative hypotheses, which are summarized in Table~\ref{table_2}. 
\begin{enumerate}
    \item The first application is a revisited proof -- yielding a strict generalisation -- of the compound i.i.d.\ null hypothesis setting considered in~\cite[Theorem 1]{lami2026universalquantumresourcedistillation}. Indeed, our Theorem~\ref{thm:iid} provides the same asymptotic exponent under Assumption~\ref{ass:standing-Sn}, thereby removing the assumptions of closure under partial traces (TR) and permutations (PER) imposed in the previous formulation.  
    In addition, Theorem~\ref{thm:almost_iid} extends this result to a compound $W_1$ almost i.i.d.\ null hypothesis with the base set $\pazocal R_1$.
    \item The second application is a complementary result to Theorem~\ref{thm:iid}, where we consider an arbitrarily varying null hypothesis, namely, the single site may vary across the tensor factors. Under Assumption~\ref{ass:standing-Sn}, and under the additional closure of the alternative sets under permutation twirl, Theorem~\ref{thm:av} allows arbitrary tensor products of
    states from $\pazocal R_1$ and their convex mixtures (see Eq.~\eqref{eq:arbitrary_product}). In contrast to the compound setting, the resulting exponent is governed by $\inf_{\rho\in \co\!(\pazocal R_1)} D^\infty(\rho\|\mathcal S)$, where the optimisation over the convex hull reflects the freedom to vary the composition of the product across sites.
\end{enumerate}

\newpage
\section{A Wasserstein-robust generalised quantum Stein's lemma}\label{sec:GQSL}

\begin{boxedassumpt}{}
\begin{assumption}[(Conditions on the alternative sets)]
\label{ass:standing-Sn}
The family \((\mathcal S^{(n)})_{n\ge1}\)  satisfies the following
conditions.
\begin{enumerate}[label=(A\arabic*)]
    \item For each $n$, \(\mathcal S^{(n)}\) is a convex, closed
    subset of the state space on \(\mathcal H^{\otimes n}\).

    \item The family is closed under tensor products: if
    \(\sigma_n\in\mathcal S^{(n)}\) and
    \(\omega_m\in\mathcal S^{(m)}\), then
    \bb
        \sigma_n\otimes\omega_m\in \mathcal S^{(n+m)}.
    \ee
    \item There exists a full-rank state \(\sigma_{\rm full}\in\mathcal S^{(1)}\) such that, for every $n$, every \(\sigma_n\in\mathcal S^{(n)}\), and
every site \(i\in\{1,\ldots,n\}\),
\begin{align}
    \pazocal R_i^{(\sigma_{\rm full})}(\sigma_n)\in \mathcal S^{(n)}.
    \label{eq:replacer-stability}
\end{align}
Here \(\pazocal R_i^{(\sigma_{\rm full})}\) denotes the channel that replaces the
$i$-th tensor factor by $\sigma_{\rm full}$ and leaves all other tensor
factors unchanged.
\end{enumerate}
\end{assumption}
\end{boxedassumpt}

Note that we do not require permutational symmetry and partial trace closure.

 Whenever some source of noise acts on the ideal i.i.d.\ state $\rho^{\otimes n}$, we typically do not have a perfect characterisation of the perturbed state $\rho_n$, but we only assume that it is $\delta_n$-close to $\rho^{\otimes n}$, say, in normalised Wasserstein distance. For instance, we may expect our perturbed state $\rho_n$ to belong to the set
\bb
    \mathcal{F}^{(n)}\coloneqq\left\{\rho_n\in\mathcal{D}(\mathcal{H}^{\otimes n}): \frac{1}{n}\|\rho_n-\rho^{\otimes n}\|_{W_1}\leq \delta_n\right\}.
\ee
Therefore, for practical testing purposes, it is fundamental to understand whether we can design a sequence of tests $(E_n)_n$ that perform at least as well as in the idealised setting (i.e.\ as in Theorem~\ref{thm:gqsl_iid}) under the convergence constraint $\delta_n\to 0$, for all sequences $\rho_n\in \mathcal{F}^{(n)}$ and $\sigma_n\in\mathcal{S}^{(n)}$, namely
\bb\label{eq:univ}
    \liminf_{n\to\infty}-\frac{1}{n}\log \Tr[E_n\sigma_n]&\geq D^\infty(\rho\|\mathcal{S}),\\
    \Tr[E_n\rho_n]&\geq 1-\epsilon & \forall n\geq 1.
\ee
Formally speaking, given a general sequence of null hypotheses $(\mathcal F^{(n)})_{n\geq 1}$ that are \emph{convex}, if we are able to prove that
\bb
  \liminf_{n\to\infty}\frac 1n
D_H^\epsilon\left(\mathcal F^{(n)}\middle\|\mathcal S^{(n)}\right)\eqt{?} D^{\infty}(\rho\| \mathcal{S}),
\ee
then a sequence of test $(E_n)_n$ -- which we call universal -- satisfying~\eqref{eq:univ} exists due to Lemma~\ref{lem:minmax_fang}, because of the very definition of $\beta_\epsilon$ in~\eqref{eq:beta-composite-null-and-alt}. For our purposes, we need the following definition.

\begin{Def}\label{def:W1-equiconvergent} A sequence \((\mathcal F^{(n)})_{n\ge1}\) of nonempty
subsets of the state space on \(\mathcal H^{\otimes n}\) is said to be a \emph{Wasserstein equiconvergent} source along $\rho$ if 
\bb
\lim_{n\to\infty}
\sup_{\rho_n\in\mathcal F^{(n)}}\frac 1n\|\rho_n-\rho^{\otimes n}\|_{W_1}=0.
\ee
\end{Def}

When a family $(\mathcal{F}_n)_n$ constitutes an equiconvergent source along a state $\rho$, we call \emph{$W_1$-robust} tests the corresponding sequence of universal tests $E_n$, as they allow to achieve the optimal type II error exponent under Wasserstein almost i.i.d.\ perturbations that are uniformly small at every $n$.

\ludo{We are now ready to state the main result of this work:}

\begin{boxedthm}{}
\begin{thm}[(Generalised quantum Stein's lemma for $W_1$ almost i.i.d.\ sources)]
\label{thm:W1_GQSL}
Suppose that the sequence of alternative hypotheses $\mathcal S=(\mathcal{S}^{(n)})_{n\geq 1}$, with $\mathcal{S}^{(n)}\subseteq\mathcal{D}(\mathcal{H}^{\otimes n})$, satisfies Assumption~\ref{ass:standing-Sn}. 
\begin{enumerate}
\item \emph{(Existence of $W_1$-robust tests).} If \((\mathcal F^{(n)})_{n\ge1}\) is a sequence of compact, convex families that forms a Wasserstein equiconvergent source along $\rho$, then
\begin{align}
\lim_{n\to\infty}\frac 1n
\rel{D_H^\epsilon}{\mathcal F^{(n)}}{\mathcal S^{(n)}} =D^{\infty}(\rho \|\mathcal S)\qquad \forall \,\epsilon\in (0,1)\, .
\label{eq:W1-stable-main-composite-null}
\end{align}
\item \emph{(Robustness for individual sources).} In particular, if \((\rho_n)_n\) is an individual Wasserstein almost i.i.d.\ source along $\rho$, then
\begin{align}
\lim_{n\to\infty}\frac 1n
\rel{D_H^\epsilon}{\rho_n}{\mathcal S^{(n)}} = 
D^{\infty}(\rho \|\mathcal S)\qquad \forall \,\epsilon\in (0,1)\, .
\label{eq:W1-stable-main}
\end{align}
\end{enumerate}
\end{thm}
\end{boxedthm}

\begin{proof}
Section~\ref{sec:big_proof} is devoted to the proof of Theorem~\ref{thm:W1_GQSL}. In Section~\ref{sec:info_spec} we establish the core technical statements needed in the proof: the extension of the information-spectrum results used in~\cite{Hayashi2025} to Wasserstein almost i.i.d.\ sources. In Section~\ref{sec:individual} we prove the first part of Theorem~\ref{thm:W1_GQSL}, i.e.\ the robustness of the generalised quantum Stein's lemma for individual sources. Finally in Section~\ref{sec:universality} we complete our proof by tackling the universal extension of the previous result for equiconvergent sources.
\end{proof}

\section{Applications}\label{sec:applications}

\begin{table}[t]
    \centering 
    \begin{talltblr}[
     note{a} = {\small This property follows from the other assumptions.},
     note{b} = {\small Note that, by definition, (FR) immediately follows from (REP). },
     note{c} = {\small More precisely, we only assume closure under permutation twirl.},
     caption = {Extensions of the generalised quantum Stein's lemma to a composite null hypothesis.},
     label = table_2,
    ]{
        width = \textwidth, 
        colspec = {X[1.1,m,c] X[1.7,m, c] X[0.6,m,c] X[0.6,m,c] X[0.6,m,c] X[0.6,m,c] X[0.6,m,c] X[0.6,m,c]}, 
    }
        \toprule
         & \vspace{-0.7em}\textbf{Null} & \SetCell[c=6]{c}{\textbf{Alternative hypothesis}} \\
        \textbf{Work} & \textbf{hypothesis} & (CC) & (TR) & (TEN) & (PER) & (FR) & (REP) \\
        \midrule
        Lami~\cite{lami2026universalquantumresourcedistillation} & compound i.i.d.  & \ding{109} & \ding{109} & \ding{109} & \ding{109} & \ding{109} & (\ding{109}\TblrNote{a}\;\,) \\
        Theorem~\ref{thm:iid} & compound i.i.d. & \ding{109} &  & \ding{109} &  & (\ding{109}\TblrNote{b}\;\,) & \ding{109} \\
        Theorem~\ref{thm:almost_iid} & compound $W_1$~almost~i.i.d. & \ding{109} &  & \ding{109} & \ding{109}\TblrNote{c} & (\ding{109}\TblrNote{b}\;\,) & \ding{109} \\
        Theorem~\ref{thm:av} & arbitrarily varying & \ding{109} &  & \ding{109} & \ding{109}\TblrNote{c} & (\ding{109}\TblrNote{b}\;\,) & \ding{109} \\
        \bottomrule
    \end{talltblr}
\end{table}

Throughout this section we are going to frequently use the smoothed max-divergence and its connection with the hypothesis testing relative entropy. We therefore begin with a brief recap of the fundamental definitions and properties of $D_{\max}$ and its smoothings. 
\begin{Def}[(Max-relative entropy)] Let $\rho,\sigma \in \mathcal{D}(\mathcal{H})$ be two states. Then, we define their max-relative entropy as~\cite{Datta-alias}
\bb
    D_{\rm max}(\rho\|\sigma) \coloneqq \min \left\{ \lambda : \rho \le \exp[\lambda]\, \sigma \right\}.
\ee
\end{Def}
The max-relative entropy captures the worst-case comparison between $\rho$ and $\sigma$, and it serves as a (loose) upper bound of the relative entropy, which is taking an average.
To bring this upper bound closer to the relative entropy, the smoothed versions were introduced~\cite{RennerPhD, Datta08}:
\bb
    D_{\rm max}^{\epsilon}(\rho\|\sigma) 
    &\coloneqq
    \min_{\rho' :\, \frac 12 \| \rho-\rho' \|_1\le \epsilon} D_{\rm max}(\rho'\|\sigma) \\
     D_{\rm max}^{\epsilon,P}(\rho\|\sigma)
    &\coloneqq
    \min_{\rho' :\, P(\rho,\rho') \le \epsilon} D_{\rm max}(\rho'\|\sigma), \\
\ee
where the first line is defined by the trace distance, and the second one is by the purified distance, i.e.\ $ P(\rho,\sigma):= \sqrt{1- \| \sqrt{\rho}\sqrt{\sigma} \|_1^2}$.
By the Fuchs--van de Graaf inequalities,
\bb
    1-\sqrt{1-P^2(\rho,\sigma)}
    \le 
    \frac12 \| \rho-\sigma \|_1
    \le
    P(\rho,\sigma),
\ee
we have 
\bb\label{eq:Dmax_ineq}
    D_{\rm max}^{\sqrt{\epsilon (2-\epsilon)},P}(\rho\|\sigma)
    \le
    D^\epsilon_{\rm max}(\rho\|\sigma)
    \le
    D_{\rm max}^{\epsilon,P}(\rho\|\sigma).
\ee
The tight relation between $D_{\max}^\epsilon$ and $D_H^\epsilon$ is given by the following result.
\begin{lemma}[{(Weak/strong converse duality between $D^\ve_{\max}$ and $D^{1-\ve}_H$~\cite[Theorem 12]{tight-relations})}] \label{lemma:wsc}
Let $\rho,\sigma\in\mathcal{D}(\mathcal{H})$ be two arbitrary quantum states. Then, for $\epsilon \in (0,1)$, and all $\mu\in (0,\epsilon]$, it holds that
\bb
D_{\max}^{\sqrt{\epsilon}}(\rho\|\sigma)  + \log \frac1\epsilon \leq D^{1-\epsilon}_H(\rho\|\sigma) &\leq D_{\max}^{\epsilon-\mu}(\rho\|\sigma) + \log \frac{1}\mu.
\ee
For commuting $\rho$ and $\sigma$, the previous lower bound can be strengthened as
\bb
D_{\max}^{\epsilon}(\rho\|\sigma)  + \log \frac1\epsilon \leq D^{1-\epsilon}_H(\rho\|\sigma).
\ee
\end{lemma}

\subsection{Compound i.i.d.\ null hypothesis}
The work~\cite{lami2026universalquantumresourcedistillation} studied an extension of the generalised quantum Stein's lemma to the case of a \emph{compound i.i.d.\ null hypothesis}, with an application to universal resource distillation. For any arbitrary subset $\pazocal R_1\subseteq \mathcal{D}(\mathcal{H})$, they considered the sequence of hypotheses
\bb
\pazocal{R}^{\rm i.i.d.} \coloneqq \big( \pazocal{R}_n^{\rm i.i.d.} \big)_n\, ,\quad \pazocal{R}_n^{\rm i.i.d.} \coloneqq \left\{ \rho^{\otimes n}:\ \rho\in \pazocal{R}_1\right\}.
\label{composite_iid}
\ee
The main theorem of~\cite{lami2026universalquantumresourcedistillation} on the hypothesis testing side is the following.
\begin{thm}[{\cite[Theorem~1]{lami2026universalquantumresourcedistillation}}] \label{black_box_Stein_iid_thm}
Let $\RR_1\subseteq \mathcal{D}(\mathcal{H})$ be any set of states on the finite-dimensional Hilbert space $\HH$. Assume that the sequence $\mathcal{S} = (\mathcal{S}^{(n)})_n$ of alternative hypotheses $\mathcal{S}^{(n)}\subseteq \mathcal{D}(\mathcal{H}^{\otimes n})$ satisfies all five Brand\~{a}o--Plenio axioms. Then
\bb
\lim_{n\to\infty} \frac1n\, \rel{D_H^\e}{\co\!\big(\pazocal{R}_n^{\rm i.i.d.}\big)}{\mathcal{S}^{(n)}} = \inf_{\rho\in \pazocal R_1} D^\infty(\rho \|\mathcal{S})
\label{black_box_Stein_iid_DH}
\ee
for all $\epsilon\in (0,1)$.
\end{thm}

The key tool in the proof of Theorem~\ref{black_box_Stein_iid_thm} given in~\cite{lami2026universalquantumresourcedistillation} is the weak quasi-concavity property of $D_{\max}^{\epsilon,P}$, which we formally recall below, combined with the duality between the smoothed max-relative entropy and the hypothesis testing relative entropy (Lemma~\ref{lemma:wsc}).

\begin{lemma}[{(Weak quasi-concavity of the smooth max-relative entropy~\cite[Lemma 4]{lami2026universalquantumresourcedistillation})}] \label{lem:weak_q_conc}
Let $(p(x),\rho_x)_{x\in \mathcal{X}}$ be a finite ensemble of states $\rho_x\in\mathcal{D}(\mathcal{H})$, and let $\sigma\in\mathcal{D}(\mathcal{H})$. Then, for all $\epsilon,\mu\in (0,1)$ such that $\epsilon+\mu\leq 1$, we have
\begin{align}
\min_x D_{\max}^{\epsilon+\mu,\,P}\!(\rho_x\|\sigma) &\leq \Rel{D_{\max}^{\epsilon,P}}{\sumno_{x\in \XX} p(x)\, \rho_x}{\sigma} + \log \frac{|\XX|}{(\epsilon+\mu)^2 - \epsilon^2}\, . \label{quasi_concavity_D_max_eps}
\end{align}
\end{lemma}

The aim of this section is to give an alternative proof of Theorem~\ref{black_box_Stein_iid_thm} yielding a strict generalisation, which holds under less restrictive assumptions on the sequence of alternative hypotheses $\mathcal{S}$ (see Theorem~\ref{thm:iid}). Furthermore, our proof technique immediately allows for an extension to the compound \emph{almost i.i.d.\ setting} (see Theorem~\ref{thm:almost_iid}). In Table~\ref{table_2} we report a comparison between Theorem~\ref{black_box_Stein_iid_thm} (i.e.\ \cite[Theorem~1]{lami2026universalquantumresourcedistillation}) and our generalisations. 

\begin{boxedthm}{}
    \begin{thm}\label{thm:iid} Suppose that the family $(\mathcal S^{(n)})_{n\geq 1}$ satisfies Assumption~\ref{ass:standing-Sn}. Let $\pazocal R_1\subseteq \mathcal{D}(\mathcal{H})$ be an arbitrary subset of states. Then,
        \bb
            \lim_{n\to\infty} \frac1n\, \rel{D_H^\epsilon}{\co\!\big(\pazocal{R}_n^{\rm i.i.d.}\big)}{\mathcal{S}^{(n)}} = \inf_{\rho\in \pazocal R_1} D^\infty(\rho \|\mathcal{S})\qquad \forall \,\epsilon\in (0,1).
        \ee
    \end{thm}
\end{boxedthm}

Section~\ref{sec:proof_iid} is devoted to the proof of Theorem~\ref{thm:iid}. The novel tool that makes our proof more general and particularly elegant and simple is the $W_1$ robust generalised quantum Stein's lemma, i.e.\ Theorem~\ref{thm:W1_GQSL}.\smallskip

In Appendix~\ref{proof:alternative}, we are going to briefly mention a possible approach to the proof of Theorem~\ref{thm:iid}, which does not need Lemma~\ref{lem:weak_q_conc} in the fully quantum setting, but only a simple equivalent argument in the fully classical setting. The price for this simplified proof is a further assumption on the sequence of families $(\mathcal S^{(n)})_{n\geq 1}$, namely the \emph{closure under permutation twirl}:
\bb
    \frac{1}{n!}\sum_{\pi\in S_n} U_\pi^{\vphantom{\dagger}} \sigma_n U^\dagger_\pi \in \mathcal{S}^{(n)}\qquad \forall \sigma_n \in \mathcal{S}^{(n)}.
\ee
We refer the reder interested in this alternative approach to Appendix~\ref{proof:alternative}, while we now focus on the proof of Theorem~\ref{thm:iid} in full generality.

\subsubsection{Proof of Theorem~\ref{thm:iid}}\label{sec:proof_iid}
Our proof is essentially base on the follwing three steps.
\begin{enumerate}
    \item We first recall a simple converse bound, based on an elementary and well-known argument (Lemma~\ref{lem:conv_iid}).
    \item Following the strategy of the first part of the proof of~\cite[Theorem~1]{lami2026universalquantumresourcedistillation}, in Lemma~\ref{lem:ach2b_iid} we verify that our weaker assumptions on the alternative hypothesis are sufficient to lower bound
    \bb
    \liminf_{n\to\infty}\tfrac 1n \rel{D_H^\epsilon}{\co\!\big(\pazocal{R}_n^{\rm i.i.d.}\big)}{\mathcal S^{(n)}} \geq \lim_{\epsilon'\to 0}\liminf_{n\to\infty}\tfrac 1n  \inf_{\rho \in\pazocal R_1} \rel{D_H^{\epsilon'}}{\rho^{\otimes n}}{\mathcal S^{(n)}}\, .
    \ee
    \item In order to complete the proof we need the final identity
    \[\liminf_{n\to\infty}\tfrac 1n  \inf_{\rho \in\pazocal R_1} \rel{D_H^{\epsilon}}{\rho^{\otimes n}}{\mathcal S^{(n)}} = \inf_{\rho \in\pazocal R_1}D^\infty(\rho \| \mathcal S),\]
    which was initially proved in~\cite{lami2026universalquantumresourcedistillation} with a different strategy, and assuming the Brand\~ao~Plenio axioms in full generality. Our novel contribution consists in a completely different approach, holding under the weaker Assumption~\ref{ass:standing-Sn}: in Lemma~\ref{lem:ach2_iid} we reduce the commutation of $\displaystyle{\inf_{\rho\in\pazocal R_1}}$ and of $\displaystyle{\liminf_{n\to\infty}}$ to a simple application of Theorem~\ref{thm:W1_GQSL}.
\end{enumerate}

Let us start by recalling the following elementary converse statement. The proof is completely standard; we include it here for completeness.

\begin{lemma}\label{lem:conv_iid} Suppose that the family $(\mathcal S^{(n)})_{n\geq 1}$ satisfies Assumption~\ref{ass:standing-Sn}. Let $\pazocal R_1\subseteq \mathcal{D}(\mathcal{H})$ be an arbitrary subset of states. Then,
        \bb\label{eq:converse_iid}
            \liminf_{n\to\infty}\frac 1n\rel{D_H^\epsilon}{\co\!\big(\pazocal{R}_n^{\rm i.i.d.}\big)}{\mathcal{S}^{(n)}}\leq \inf_{\rho\in\pazocal R_1}D^{\infty}(\rho\|\mathcal{S})\qquad \forall \,\epsilon\in (0,1).
        \ee
    \end{lemma}
\begin{proof} For all $\rho\in \pazocal R_1$, we have $\rho^{\otimes n}\in \co\!\big(\pazocal{R}_n^{\rm i.i.d.}\big)$, whence
    \bb
        \liminf_{n\to\infty}\frac 1n\rel{D_H^\epsilon}{\co\!\big(\pazocal{R}_n^{\rm i.i.d.}\big)}{\mathcal{S}^{(n)}}\leq \liminf_{n\to\infty}\frac 1n\rel{D_H^\epsilon}{\rho^{\otimes n}}{\mathcal{S}^{(n)}}= D^{\infty}(\rho\|\mathcal{S}).
    \ee
    Then, by arbitrariness of $\rho\in\pazocal R_1$,
    we immediately have~\eqref{eq:converse_iid}.
\end{proof}

\begin{lemma}\label{lem:ach2b_iid} Suppose that the family $(\mathcal S^{(n)})_{n\geq 1}$ satisfies Assumption~\ref{ass:standing-Sn}. Let $\pazocal R_1\subseteq \mathcal{D}(\mathcal{H})$ be an arbitrary subset of states. Then,
      \bb 
      \liminf_{n\to\infty}\frac 1n \rel{D_H^\epsilon}{\co\!\big(\pazocal{R}_n^{\rm i.i.d.}\big)}{\mathcal S^{(n)}} \geq \lim_{\epsilon'\to 0}\liminf_{n\to\infty}\frac 1n  \inf_{\rho \in\pazocal R_1} \rel{D_H^{\epsilon'}}{\rho^{\otimes n}}{\mathcal S^{(n)}}\, .
        \ee
        for all $\epsilon\in (0,1)$.
    \end{lemma}
    \begin{proof}
        Similarly to the first part of the proof of Lemma~\ref{lem:ach1a_iid}, choose $(\rho^{(n)})_{n\geq 1}$ and $(\sigma_n)_{n\geq 1}$ to be two sequences of quasi-minimisers $\rho^{(n)}\in\co\!\big(\pazocal{R}_n^{\rm i.i.d.}\big)$ and $\sigma_n\in\mathcal{S}^{(n)}$, i.e.\ they satisfy~\eqref{eq:quasi_min}.
Since  $(\rho^{(n)})_{n\geq 1}$ is invariant under permutations, it belongs to real vector space $H_{d,n}^{\rm sym}$ of permutationally symmetric Hermitian operators on $\mathcal{H}^{\otimes n}\simeq \big(\mathbb{C}^d\big)^{\otimes n}$, which has dimension upper bounded as $\dim H_{d,n}^{\rm sym}\leq (n+1)^{d^2-1}$~\cite{Hayashi2016-bh}.
 By Carath\'eodory's theorem, we can write
 \bb
    \rho^{(n)}=\sum_{i=1}^Np_i\rho_i^{\otimes n}, \qquad 0\leq p_i\leq 1,\qquad \rho_i\in\pazocal R_1 \quad \text{ for all } 1\leq i\leq N,
 \ee
 with $N\leq (n+1)^{d^2-1}+1={\rm poly}(n)$.  Then,
    \bb\label{eq:lower_bounds2}
    \rel{D_H^\epsilon}{\rho^{(n)}}{\sigma_n} &= \Rel{D_H^\epsilon}{\sumno_{i=1}^N p_i \rho_i^{\otimes n}}{\sigma_n} \\
        &\geqt{(i)} \Rel{D_{\max}^{1-\epsilon^2/8,P}}{\sumno_{i=1}^N p_i\rho_i^{\otimes n}}{\sigma_n} + \log \frac1{1-\epsilon}\\
        &\geqt{(ii)} \min_i \rel{D_{\max}^{1-\epsilon^2/16}}{\rho_i^{\otimes n}}{\sigma_n} - \log {\rm poly}_{\epsilon,d}(n)\\
        &\geqt{(iii)} \min_i \rel{D_H^{\epsilon^2/32}}{\rho_i^{\otimes n}}{\sigma_n} - \log {\rm poly}'_{\epsilon,d}(n)\\
        &\geq \inf_{\rho \in\pazocal R_1} \rel{D_H^{\epsilon^2/32}}{\rho^{\otimes n}}{\mathcal{S}^{(n)}} - \log {\rm poly}'_{\epsilon,d}(n)\, ,
    \ee
where (i) follows from the weak/strong-converse duality (Lemma~\ref{lemma:wsc}), combined with~\eqref{eq:Dmax_ineq}, the inequality $\big((2-\sqrt{1-\epsilon})\sqrt{1-\epsilon}\big)^{1/2}\leq 1-\epsilon^2/8$ and the monotonicity of $\epsilon\to D_{\max}^\epsilon$; (ii) stems from the weak quasi-concavity of the smooth max-relative entropy (Lemma~\ref{lem:weak_q_conc}), followed once more by~\eqref{eq:Dmax_ineq}; (iii) is another application of Lemma~\ref{lemma:wsc}. Now recalling the definition of $\rho^{(n)}$ and $\sigma_n$, we use~\eqref{eq:lower_bounds2} to lower bound
    \bb\label{eq:previous2}
    \liminf_{n\to\infty}\frac 1n \rel{D_H^\epsilon}{\co\!\big(\pazocal{R}_n^{\rm i.i.d.}\big)}{\mathcal S^{(n)}} &\geq \liminf_{n\to\infty}\frac 1n \rel{D_H^\epsilon}{\rho^{(n)}}{\sigma_n} \\
    &\geq\liminf_{n\to\infty}\frac 1n  \inf_{\rho \in\pazocal R_1} \rel{D_H^{\epsilon^2/32}}{\rho^{\otimes n}}{\mathcal S^{(n)}} \\
    &\geq\lim_{\epsilon'\to 0}\liminf_{n\to\infty}\frac 1n  \inf_{\rho \in\pazocal R_1} \rel{D_H^{\epsilon'}}{\rho^{\otimes n}}{\mathcal S^{(n)}}\, ,
    \ee
    and this completes the proof of Lemma~\ref{lem:ach2b_iid}.
    \end{proof}

    \begin{lemma}\label{lem:ach2_iid} Suppose that the family $(\mathcal S^{(n)})_{n\geq 1}$ satisfies Assumption~\ref{ass:standing-Sn}. Let $\pazocal R_1\subseteq \mathcal{D}(\mathcal{H})$ be an arbitrary subset of states. Then,
        \bb\label{eq:claim2}
         \liminf_{n\to\infty}\frac 1n  \inf_{\rho \in\pazocal R_1} \rel{D_H^{\epsilon}}{\rho^{\otimes n}}{\mathcal S^{(n)}} &=\inf_{\rho \in\pazocal R_1} D^\infty(\rho \| \mathcal S)\qquad \forall \epsilon\in (0,1).
    \ee
    \end{lemma}
    \begin{proof}
    Take any arbitrary $\epsilon\in (0,1)$, and write
    \bb\label{eq:subsequence}
        \liminf_{n\to\infty}\frac 1n  \inf_{\rho \in\pazocal R_1} \rel{D_H^{\epsilon}}{\rho^{\otimes n}}{\mathcal S^{(n)}}
        &= \lim_{k\to\infty}\frac 1{\bar n(k)}\, \Rel{D_H^{\epsilon}}{\rho^{\otimes \bar n(k)}_{\bar n(k)}}{\mathcal S^{(\bar n(k))}}\, ,
    \ee
    where the subsequence $\bar n(k)$ is constructed as follows:
    \begin{itemize}
        \item first, we extract a subsequence $n(k)$ of $1,2,\dots$ such that
            \[\liminf_{n\to\infty}\frac 1n  \inf_{\rho \in\pazocal R_1} \rel{D_H^{\epsilon}}{\rho^{\otimes n}}{\mathcal S^{(n)}}
        = \lim_{k\to\infty}\frac 1{ n(k)} \inf_{\rho \in\pazocal R_1} \rel{D_H^{\epsilon}}{\rho^{\otimes  n(k)}}{\mathcal S^{(n(k))}}\,;\]
        \item then, we identify a subsequence of quasi-minimisers $\rho_{n(k)}\in\pazocal R_1$, e.g.\ 
        \[  \Big| \Rel{D_H^{\epsilon}}{\rho^{\otimes  n(k)}_{n(k)}}{\mathcal S^{(n(k))}} - \inf_{\rho \in\pazocal R_1} \rel{D_H^{\epsilon}}{\rho^{\otimes  n(k)}}{ \mathcal S^{(n(k))}}\Big|\leq 1\, ;\]
        \item finally, we extract a subsequence $\bar n(k)$ of $n(k)$ such that $\rho_{\bar n(k)}$ converges to an element $\bar \rho$ belonging to the closure $\bar{\pazocal R}_1$ of $\pazocal R_1$, namely $\displaystyle{\lim_{k\to\infty}\|\rho_{\bar n(k)}-\bar \rho\|_1=0}$.
    \end{itemize}
    This procedure yields~\eqref{eq:subsequence}. Define the sequence $(\bar \rho_n)_{n\geq 1}$ of states $\bar\rho_n\in\mathcal{D}(\mathcal{H}^{\otimes n})$ as
    \bb
        \bar \rho_n\coloneqq\begin{cases}
        \rho_{\bar n(k)}^{\otimes \bar n(k)} & \text{if } n=\bar n(k) \text{ for some } k\geq 1,\\
            \bar \rho^{\otimes n}&\text{otherwise.}
        \end{cases}
    \ee
    The source $(\bar\rho_n)_{n\geq 1}$ is $W_1$ almost i.i.d.\ along $\bar \rho$ since
    \bb
        \lim_{n\to\infty}\frac{1}{n}\|\bar\rho_n-\bar\rho^{\otimes n}\|_{W_1}\leq\lim_{k\to\infty}\frac 12\|\rho_{\bar n(k)} -\bar \rho\|_1=0.
    \ee
    Hence, by Theorem~\ref{thm:W1_GQSL}, we have
    \bb
        \lim_{k\to\infty}\frac 1{\bar n(k)} \Rel{D_H^{\epsilon}}{\rho^{\otimes \bar n(k)}_{\bar n(k)}}{\mathcal S^{(\bar n(k))}} \geq \liminf_{n\to\infty}\frac 1{ n} \rel{D_H^{\epsilon}}{\bar \rho_n}{\mathcal S^{( n)}} = D^{\infty}(\bar \rho\|\mathcal{ S})\geq \inf_{\rho\in\bar{\pazocal R}_1}D^\infty(\rho\|\mathcal{S}).
    \ee
    In order to restrict the optimisation from $\bar{\pazocal R}_1$ to $\pazocal R_1$ it is sufficient to leverage the continuity of the relative entropy of resource with respect to the normalised $W_1$ distance (Corollary~\ref{cor:3.1}). By~\eqref{eq:subsequence}, this concludes the proof of~\eqref{eq:claim2}. The converse inequality can be proved with the same approach of Lemma~\ref{lem:conv_iid}.
\end{proof}

\subsubsection{An extension to the case of a compound almost i.i.d.\ null hypothesis}
As a simple consequence of the previous proof strategy, we can extend Theorem~\ref{thm:iid} to the case of a compound almost i.i.d.\ null hypothesis. More precisely, let $\pazocal R_1\subseteq \mathcal{D}(\mathcal{H})$ be an arbitrary subset of states and let
\bb
\pazocal{R}^{\rm a-i.i.d.} \coloneqq \big( \pazocal{R}_n^{\rm a-i.i.d.} \big)_n\, ,\quad \text{with}\quad \pazocal{R}_n^{\rm a-i.i.d.} \subseteq \mathcal{D}(\mathcal{H}^{\otimes n}),
\ee
be a sequence of null hypotheses satisfying
\bb
    \lim_{n\to\infty}\sup_{\rho_n\in\pazocal{R}_n^{\rm a-i.i.d.}}\inf_{\rho\in\pazocal R_1}\frac 1n \|\rho_n-\rho^{\otimes n}\|_{W_1}=0,
\label{composite_almost_iid}
\ee
and such that, for all $\rho\in\pazocal R_1$, there exists at least one sequence of states $\rho_n\in \pazocal{R}_n^{\rm a-i.i.d.}$ ($n\geq 1$) which is a $W_1$ almost i.i.d.\ source along $\rho$.
In this case, we say that $\pazocal{R}^{\rm a-i.i.d.}$ is an \emph{compound almost i.i.d.\ null hypothesis with base set $\pazocal R_1$}.
Note that the particular case $\pazocal{R}_n^{\rm a-i.i.d.}=\{ \rho^{\otimes n}:\ \rho\in \pazocal{R}_1\}$, which satisfies~\eqref{composite_almost_iid}, is exactly the one discussed in Theorem~\ref{thm:iid}. 
\begin{boxedthm}{}
    \begin{thm}\label{thm:almost_iid} Suppose that the family $(\mathcal S^{(n)})_{n\geq 1}$ satisfies Assumption~\ref{ass:standing-Sn} and it is closed under permutation twirl. Let $\pazocal R_1\subseteq \mathcal{D}(\mathcal{H})$ be an arbitrary subset of states, and consider any arbitrary compound almost i.i.d.\ null hypothesis $\pazocal{R}^{\rm a-i.i.d.}$ with base set $\pazocal R_1$. Then,
        \bb
            \lim_{n\to\infty} \frac1n\, \rel{D_H^\epsilon}{\co\!\big(\pazocal{R}_n^{\rm a-i.i.d.}\big)}{\mathcal{S}^{(n)}} = \inf_{\rho\in \pazocal R_1} D^\infty(\rho \|\mathcal{S})\qquad \forall \,\epsilon\in (0,1).
        \ee
    \end{thm}
\end{boxedthm}

\begin{proof}
    The first steps of the achievability argument are almost identical to the strategy discussed in Section~\ref{sec:proof_iid}. Let $(\rho^{(n)})_{n\geq 1}$ and $(\sigma_n)_{n\geq 1}$ be two sequences of quasi-optimisers $\rho^{(n)}\in\co\!\big(\pazocal{R}_n^{\rm a-i.i.d.}\big)$ and $\sigma_n\in\mathcal{S}^{(n)}$ such that 
\bb\label{eq:q_min}
     \rel{D_H^\epsilon}{\co\!\big(\pazocal{R}_n^{\rm a-i.i.d.}\big)}{\mathcal{S}^{(n)}}\geq \rel{D_H^\epsilon}{\rho^{(n)}}{\sigma_n}-1.
\ee
By the data-processing inequality, calling $\pazocal P_n$ the permutation twirl over $S_n$, we 
\bb
    \rel{D_H^\epsilon}{\rho^{(n)}}{\sigma_n}\geq \rel{D_H^\epsilon}{\pazocal P_n(\rho^{(n)})}{\pazocal P_n(\sigma_n)}.
\ee
Now, since $\pazocal P(\rho^{(n)})$ is permutation-invariant by construction and belongs to $\pazocal P_n\big(\co\!\big(\pazocal{R}_n^{\rm a-i.i.d.}\big)\big)= \co\!\big(\pazocal P_n(\pazocal{R}_n^{\rm a-i.i.d.})\big)$, we proceed as in the proof of Lemma~\ref{lem:ach1a_iid} or Lemma~\ref{lem:ach2b_iid}. First, we write
 \bb
    \pazocal P_n(\rho^{(n)})=\sum_{i=1}^Np_i\pazocal P(\rho_{i,n}), \qquad 0\leq p_i\leq 1, \qquad \rho_{i,n}\in \pazocal{R}_n^{\rm a-i.i.d.}
 \ee
by Carathéodory's theorem, where $N\leq(n+1)^{d^2-1}+1={\rm poly}(n)$. Then, with either a quantum or a classical quasi-concavity argument for the smooth max-relative entropy (see Section~\ref{sec:proof_iid}), we can lower bound
\bb\label{eq:1}
    \rel{D_H^\epsilon}{\rho^{(n)}}{\sigma_n} &\geq \min_i \rel{D_H^{\epsilon^2/32}}{\pazocal P_n(\rho_{i,n})}{\pazocal P_n(\sigma_n)} - \log {\rm poly}_{\epsilon,d}(n)\\
    &\geq \inf_{\rho^{(n)}\in \pazocal{R}_n^{\rm a-i.i.d.}} \rel{D_H^{\epsilon^2/32}}{\pazocal P_n(\rho^{(n)})}{\mathcal{S}^{(n)}} - \log {\rm poly}_{\epsilon,d}(n)\, .
\ee
Now, take two quasi-optimisers $\bar \rho^{(n)}\in \pazocal{R}_n^{\rm a-i.i.d.} $ and $\bar \rho_{n}\in \pazocal R_1$ such that
\bb\label{eq:second}
    \inf_{\rho_n\in \pazocal{R}_n^{\rm a-i.i.d.}} \rel{D_H^{\epsilon^2/32}}{\pazocal P_n(\rho_{n})}{\mathcal{S}^{(n)}} &\geq \rel{D_H^{\epsilon^2/32}}{\pazocal P_n(\bar\rho^{(n)})}{\mathcal{S}^{(n)}} - 1,\\
    \frac{1}{n}\left\|\bar \rho^{(n)}-\bar \rho_n^{\otimes n}\right\|_{W_1}&\leq \sup_{\rho^{(n)}\in\pazocal{R}_n^{\rm a-i.i.d.}}\inf_{\rho\in\pazocal R_1}\frac 1n \left\|\rho^{(n)}-\rho^{\otimes n}\right\|_{W_1}+\frac 1n\, .
\ee
Extract a subsequence $n(k)$ such that
\bb\label{eq:2}
    \liminf_{n\to\infty}\frac 1n\,\rel{D_H^{\epsilon^2/32}}{\pazocal P_n(\bar \rho^{(n)})}{\mathcal{S}^{(n)}} = \lim_{k\to\infty}\frac 1{n(k)}\,\rel{D_H^{\epsilon^2/32}}{\pazocal P_{n(k)}(\bar \rho^{(n(k))})}{\mathcal{S}^{(n(k))}}\, ,
\ee
then extract from $n(k)$ a sub-subsequence $\bar n(k)$ such that 
\bb
    \lim_{k\to\infty} \left\|\bar\rho_{\bar n(k)}-\bar \rho\right\|_1 = 0
\ee
for some $\bar \rho\in\bar{\pazocal R}_1$, where $\bar{\pazocal R}_1$ denotes the closure of $\pazocal R_1$.  The source $(\tilde \rho_n)_{n\geq 1}$, defined as
    \bb
        \tilde \rho_n\coloneqq\begin{cases}
        \bar \rho^{(\bar n(k))} & \text{if } n=\bar n(k) \text{ for some } k\geq 1,\\
            \bar \rho^{\otimes n}&\text{otherwise,}
        \end{cases}\qquad \tilde\rho_n\in\mathcal{D}(\mathcal{H}^{\otimes n}),
    \ee
    is a $W_1$-almost i.i.d.\ along $\bar \rho$. Indeed,
    \bb\label{eq:itsW}
         \lim_{n\to\infty}\frac{1}{n}\left\|\tilde\rho_n-\bar\rho^{\otimes n}\right\|_{W_1} &\leq \lim_{k\to\infty}\frac{1}{\bar n(k)}\left\|\bar \rho^{(\bar n(k))}-\bar \rho^{\otimes \bar n(k)}\right\|_{W_1}\\
         &\leq \lim_{k\to\infty}\frac{1}{\bar n(k)} \left\|\bar \rho^{(\bar n(k))}-\bar \rho_{\bar n(k)}^{\otimes \bar n(k)}\right\|_{W_1}+\lim_{k\to\infty}\frac 12\left\|\bar \rho_{\bar n(k)}-\bar \rho\right\|_1 \\
         &\leq \lim_{k\to\infty}\sup_{\rho^{(\bar n(k))}\in\pazocal{R}_{\bar n(k)}^{\rm a-i.i.d.}}\inf_{\rho\in\pazocal R_1}\frac{1}{\bar n(k)} \left\|\rho^{(\bar n(k))}-\rho^{\otimes \bar n(k)}\right\|_{W_1} = 0\, ,
    \ee
    where the last inequality follows from the second bound in~\eqref{eq:second}.
    In particular, also $(\pazocal P_n(\tilde \rho_n))_{n\geq 1}$ is a $W_1$ almost i.i.d.\ source along $\bar\rho$. Then, combining the previous equations, we get
    \bb
        \liminf_{n\to\infty}\frac 1n\, \rel{D_H^\epsilon}{\rho^{(n)}}{\sigma_n} &\geqt{\eqref{eq:1}} \liminf_{n\to\infty}\frac 1n\inf_{\rho^{(n)}\in \pazocal{R}_n^{\rm a-i.i.d.}} \rel{D_H^{\epsilon^2/32}}{\pazocal P_n(\rho^{(n)})}{\mathcal{S}^{(n)}}\\
        &\eqt{\eqref{eq:2}}
        \liminf_{k\to\infty}\frac 1{n(k)}\, \rel{D_H^{\epsilon^2/32}}{\pazocal P_{n(k)}(\bar \rho^{(n(k))})}{\mathcal{S}^{(n(k))}}\\
        &= \liminf_{k\to\infty}\frac 1{\bar n(k)}\,\rel{D_H^{\epsilon^2/32}}{\pazocal P_{\bar n(k)}(\bar \rho^{(\bar n(k))})}{\mathcal{S}^{(\bar n(k))}}\\
        &= \lim_{n\to\infty}\frac 1n\, \rel{D_H^{\epsilon^2/32}}{\pazocal P_{ n}(\tilde \rho_n)}{\mathcal{S}^{( n)}} \\
        &= D^\infty(\bar \rho\|\mathcal{S})\\
        &\geq \inf_{\rho\in\bar {\pazocal R}_1}D^\infty(\rho\|\mathcal{S}),
    \ee
    where the last two equalities follow from Theorem~\ref{thm:W1_GQSL}. Then, we conclude the achievability inequality by noticing that $\displaystyle{\inf_{\rho\in\bar {\pazocal R}_1}D^\infty(\rho\|\mathcal{S})=\inf_{\rho\in\pazocal R_1}D^\infty(\rho\|\mathcal{S})}$, as already discussed in the final part of the proof of Lemma~\ref{lem:ach2_iid}. \bigskip
    
    For the converse, choose any arbitrary $\rho\in\pazocal R_1$ and identify a $W_1$ almost i.i.d.\ source $\rho_n\in \pazocal{R}_n^{\rm a-i.i.d.}$ ($n\geq 1$) along $\rho$. Then, for all $\epsilon\in(0,1)$
    \bb
        \limsup_{n\to\infty} \frac 1n\rel{D_H^\epsilon}{\co\!\big(\pazocal{R}_n^{\rm a-i.i.d.}\big)}{\mathcal{S}^{(n)}}&\leq \limsup_{n\to\infty} \frac 1n \rel{D_H^\epsilon}{\rho_n}{\mathcal{S}^{(n)}}
        = D^{\infty}(\rho\|\mathcal{S}),
    \ee
    where the last identity follows from Theorem~\ref{thm:W1_GQSL}. Then, by arbitrariness of $\rho\in\pazocal R_1$, we conclude that 
    \bb
        \limsup_{n\to\infty} \frac 1n\rel{D_H^\epsilon}{\co\!\big(\pazocal{R}_n^{\rm a-i.i.d.}\big)}{\mathcal{S}^{(n)}}&\leq \inf_{\rho\in\pazocal R_1} D^{\infty}(\rho\|\mathcal{S}),
    \ee
    which completes the proof.
\end{proof}

\subsection{Arbitrarily varying null hypothesis}
 
 In this section we prove a complementary result to Theorem~\ref{thm:iid}, based on a natural relaxation of the compound i.i.d.\ case: the \emph{arbitrarily varying} setting. Here, the null hypothesis allows for any arbitrary product of $n$ states from the base set $\pazocal R_1$. More formally, let us define
\bb
\pazocal{R}^{\rm a.v.} \coloneqq \big( \pazocal{R}_n^{\rm a.v.} \big)_n\, ,\quad \pazocal{R}_n^{\rm a.v.} \coloneqq \left\{ \rho^{(1)}\otimes\rho^{(2)}\otimes\cdots\otimes\rho^{(n)}:\ \rho^{(i)}\in \pazocal{R}_1\,\; \forall i\in[n]\right\}. \label{eq:arbitrary_product}
\ee
Then, the following result holds.
\begin{boxedthm}{}
    \begin{thm}\label{thm:av} Suppose that the family $(\mathcal S^{(n)})_{n\geq 1}$ satisfies Assumption~\ref{ass:standing-Sn} and is closed under permutation twirl. Let $\pazocal R_1\subseteq \mathcal{D}(\mathcal{H})$ be an arbitrary subset of states. Then,
        \bb
            \lim_{n\to\infty}\frac{1}{n}\,\rel{D_H^{\epsilon}}{{\rm conv}\big(\pazocal{R}_n^{\rm a.v.}\big)}{\mathcal{S}^{(n)}} = \inf_{\rho\in{\rm conv}(\pazocal{R}_1)}D^{\infty}(\rho\|\mathcal{S})\qquad \forall \,\epsilon\in (0,1).
        \ee
    \end{thm}
\end{boxedthm}

Before presenting the proof of Theorem~\ref{thm:av}, we need a result from the theory of types, which we briefly recall. Let $\mathcal{X}_1$ be a finite set. We call $\pazocal T_n$ the set of the types of the sequences in $\mathcal{X}_1^n$, namely
\bb
    \pazocal T_n\coloneqq \big\{t\in\mathcal{P}(\mathcal{X_1}): nt(x)\in\mathbb{N}, \forall x\in\mathcal{X}_1\big\},
\ee
where $\mathcal P(\mathcal{X}_1)$ is the set of probability distributions on $\mathcal{X}_1$.
We say that a sequence $x^n\in\mathcal{X}_1^n$ has type $t$ if $|\{i:x_i=x, i\in[n]\}|=nt(x)$ for all $x\in\mathcal{X}_1$. The uniform distribution $u_{t,n}$ on all the sequences of type $t\in\pazocal{T}_n$ is defined as
\bb
    u_{t,n}(x^n)\coloneqq\frac{\prod_{x\in\mathcal{X}}(nt(x))!}{n!}\times \begin{cases}
        1 & \text{if $x^n$ has type $t$,} \\
        0 & \text{otherwise.}
    \end{cases}
\ee

\begin{prop}[{\cite[Proposition~23]{almost_iid}}]\label{prop:quantitative_Wass}
        Let $n\geq 1$ be any integer, $t\in\pazocal{T}_{n}$ any type on an alphabet of cardinality $N\coloneqq |\XX|$. We denote by $u_{t,n}$ the uniform distribution on sequences $x^n$ of type $t$. Then,
    \bb
        \frac{1}{n}\|u_{t,n}-t^{\otimes n}\|_{W_1}\leq \sqrt{\frac{(N-1) \log (n+1)}{2n}}\, .
    \ee
\end{prop}

Heuristically, Proposition~\ref{prop:quantitative_Wass} can be interpreted as follows: since the mass of $t^{\otimes n}$ is asymptotically concentrated on the $(n,\epsilon)$-typical set $T_{n}^\epsilon$ of $t$, and since the items of $T_{n}^\epsilon$ are close in Hamming distance to the support of $u_{t,n}$, then $t^{\otimes n}$ and $u_{t,n}$ are close in Wasserstein distance.\smallskip

Now we have all the ingredients to prove Theorem~\ref{thm:av}.

\begin{proof}[Proof of Theorem~\ref{thm:av}.]
 Let us start with the converse inequality. Since for all $\rho \in {\rm conv}(\pazocal R_1)$ and $n\geq 1$ the state $\rho^{\otimes n}$ belongs to ${\rm conv}\big(\pazocal{R}_n^{\rm a.v.}\big)$, we have
    \bb
    \limsup_{n\to\infty}\frac{1}{n}\, \rel{D_H^{\epsilon}}{{\rm conv}\big(\pazocal{R}_n^{\rm a.v.}\big)}{ \mathcal{S}^{(n)}} &\leq \limsup_{n\to\infty}\frac{1}{n}\,\rel{D_H^{\epsilon}}{\rho^{\otimes n}}{ \mathcal{S}^{(n)}} = D^{\infty}(\rho\|\mathcal{S})\, ,
    \ee
    whence, by arbitrariness of $\tau \in {\rm conv}(\pazocal R_1)$ and by~\eqref{eq:lower_bound}, we conclude that
    \bb
        \limsup_{n\to\infty}\frac{1}{n}\,\rel{D_H^{\epsilon}}{{\rm conv}\big(\pazocal{R}_n^{\rm a.v.}\big)}{ \mathcal{S}^{(n)}} \leq\inf_{\rho \in {\rm conv}(\pazocal R_1)}D^{\infty}(\rho\|\mathcal{S})\, .
    \ee
The proof of the a achievability is divided into two parts: first, we consider the claim for a finite set $\pazocal R_1$, then we extend the result to a possibly infinite $\pazocal R_1$.\smallskip

\noindent\textbf{Step 1. Proof in the case $|\pazocal{R}_1|<+\infty$.}
 We consider a parameterisation of $\pazocal{R}_1$ as follows:
\bb
    \pazocal{R}_1=\{\rho_x\}_{x\in\mathcal{X}_1}\qquad \text{where}\qquad \mathcal{X}_1=\{x_1,\dots, x_N\}\, .
\ee
Given $\epsilon\in(0,1)$, we want to prove
\bb
    \inf_{\rho\in{\rm conv}(\pazocal{R}_1)}D^{\infty}(\rho\|\mathcal{S})\leqt{?}\liminf_{n\to\infty}\frac{1}{n}\,\rel{D_{\max}^{\epsilon,P}}{{\rm conv}\big(\pazocal{R}_n^{\rm a.v.}\big)}{ \mathcal{S}^{(n)}} \leq \inf_{\rho\in{\rm conv}(\pazocal{R}_1)}D^{\infty}(\rho\|\mathcal{S})\, ,
\ee
where $\displaystyle{D_{\max}^{\epsilon, P}(\pazocal A\|\pazocal B)\coloneqq \inf_{\rho\in\pazocal A}\inf_{\sigma\in\pazocal B}D_{\max}^{\epsilon, P}(\rho\|\sigma)}$.
By contradiction, suppose that the desired lower bound for the liminf fails. Then there exist
\(\lambda<\inf_{\rho\in{\rm conv}(\pazocal{R}_1)}D^{\infty}(\rho\|\mathcal{S})\)
and a strictly increasing sequence of integers $(n_j)_{j\geq 1}$ such that
\bb
    \rel{D_{\max}^{\epsilon,P}}{{\rm conv}\big(\pazocal{R}_{n_j}^{\rm a.v.}\big)}{ \mathcal{S}^{(n_j)}} < n_j\lambda\qquad\forall j\geq 1\, .
\ee
We restrict the remainder of this step to the indices $n\in\{n_j:j\geq1\}$; thus every subsequent occurrence of $n$ in this step ranges along this fixed subsequence.
Then, there exist
\begin{itemize}
    \item a state $R_n\in {\rm conv}\big(\pazocal{R}_n^{\rm a.v.}\big)$ and a state $\Omega_n\in\mathcal{D}(\mathcal{H}^{\otimes n})$ such that $P(R_n,\Omega_n)\leq \epsilon$, where $P$ is the purified distance;
    \item a state $\Sigma_n\in\mathcal{S}^{(n)}$ satisfying $D_{\max}(\Omega_n\|\Sigma_n)\leq n\lambda$.
\end{itemize}
In a single line,
\bb\label{eq:line}
    R_n \approx_{\epsilon} \Omega_n \leq \exp(n\lambda)\Sigma_n.
\ee
Since both $R_n\in {\rm conv}\pazocal{R}_1^{\otimes n, {\rm av}}$ and $\mathcal{S}_n$ are closed under the action of the permutation twirl, by the data-processing inequality we can assume without loss of generality that the states appearing in~\eqref{eq:line} are permutation invariant. As an immediate consequence, we can write $R_n$ as a sum indexed by types in $\pazocal{T}_n$:
\bb
    R_n=\sum_{t\in\pazocal{T}_n}p(t) \rho_{n,t} \qquad\text{with}\qquad \rho_{n,t}\coloneqq \pazocal P \left(\rho_{x_1}^{\otimes nt(x_1)}\otimes\cdots\otimes \rho_{x_N}^{\otimes nt(x_N)} \right),
\ee
where $p\in\mathcal{P}(\pazocal{T}_n)$ and $\pazocal P$ is the permutation twirl, namely
\bb
    \pazocal P(\,\cdot\,) = \frac{1}{n!}\sum_{\pi\in S_n}U_\pi\,\cdot\, U_\pi^\dagger
\ee
Then, by the weak quasi-concavity of the smoothed max-relative entropy (Lemma~\ref{lem:weak_q_conc}), for every $n\in\{n_j:j\geq1\}$ there exists
a type $t_n\in \pazocal{T}_n$ such that, for all $\epsilon'\in (\epsilon,1)$, 
\bb\label{eq:weak_qc}
    D_{\max}^{\epsilon',P}(\rho_{n,t_n}\|\Sigma_n)&\leq D_{\max}^{\epsilon,P}(R_n\|\Sigma_n)+\log |\pazocal T_n|-\log(\epsilon'^2-\epsilon^2)\\
    &\leq n\lambda+\log {\rm poly}(n)+f(\epsilon,\epsilon'),
\ee
where we have called $f(\epsilon,\epsilon')\coloneqq -\log(\epsilon'^2-\epsilon^2)$. 
Now, fix $\epsilon'\in (\epsilon,1)$ and $0< \delta< 1-\epsilon'$. By the weak/strong converse duality (Lemma~\ref{lemma:wsc}), we get
\bb\label{eq:inequalities}
    D_H^{\delta}(\rho_{n,t_n}\|\Sigma_n)+\log(1-\epsilon'-\delta)\leq D_{\max}^{\epsilon'}(\rho_{n,t_n}\|\Sigma_n)\leq n\lambda+\log {\rm poly}(n)+f(\epsilon,\epsilon'),
\ee
where in the last inequality we have recalled that $D_{\max}^{\epsilon'}(\rho_{n,t_n}\|\Sigma_n)\leq D_{\max}^{\epsilon',P}(\rho_{n,t_n}\|\Sigma_n)$.
We claim that
    \bb\label{eq:claim}
        \liminf_{n\to\infty} \frac{1}{n}D_H^{\delta}(\rho_{n,t_n}\|\Sigma_n)\geqt{?}  \inf_{\rho\in{\rm conv}(\pazocal{R}_1)}D^{\infty}(\rho\|\pazocal{S}).
    \ee
Take a subsequence $n(k)$ such that
\bb
     \lim_{k\to\infty}\frac{1}{n(k)}D_H^{\delta}(\rho_{n(k),t_{n(k)}}\|\Sigma_{n(k)})=\liminf_{n\to\infty} \frac{1}{n}D_H^{\delta}(\rho_{n,t_n}\|\Sigma_n).
\ee

    When $\pazocal{R}_1$ is finite, by the compactness of $\mathcal{P}(\pazocal{R}_1)$ we can take a subsequence $\bar n(k)$ of $n(k)$ such that $\displaystyle{\exists \bar t\coloneqq \lim_{k\to \infty} t_{\bar n(k)}}$.
    Then, by Proposition~\ref{prop:quantitative_Wass}, we have
    \bb
        \frac{1}{\bar n(k)}\left\|u_{t_{\bar n(k)},\bar n(k)}-\bar t^{\otimes \bar n(k)}\right\|_{W_1}&\leq \frac{1}{\bar n(k)}\left\|u_{t_{\bar n(k)},\bar n(k)}-t_{\bar n(k)}^{\otimes \bar n(k)}\right\|_{W_1}+\frac{1}{\bar n(k)}\left\|t_{\bar n(k)}^{\otimes \bar n(k)}-\bar t^{\otimes \bar n(k)}\right\|_{W_1}\\
        &\leq \sqrt{\frac{(N-1) \log (\bar n(k)+1)}{2\bar n(k)}}+\frac 12\|t_{\bar n(k)}-\bar t\|_1\, .
    \ee
        In particular,
    \bb\label{eq:limit_nk}
        \lim_{k\to\infty}\frac{1}{\bar n(k)}\left\|u_{t_{\bar n(k)},\bar n(k)}-\bar t^{\otimes \bar n(k)}\right\|_{W_1}=0\, .
    \ee
    As a corollary, $(\rho_{\bar n(k),t_{\bar n(k)}})_{k\geq 1}$ is a $W_1$ almost i.i.d.\ source along 
    \bb
    \rho_{\bar t}\coloneqq\sum_{x\in\mathcal{X}_1} \bar t(x)\rho_x.
    \ee
    Indeed, by the data-processing inequality with the single-system classical-quantum channel $x \mapsto\rho_x$, and by~\eqref{eq:limit_nk}, we get
    \bb
        \lim_{k\to\infty}\frac{1}{\bar n(k)}\left\|\rho_{\bar n(k),t_{\bar n(k)}}-\rho_{\bar t}^{\otimes \bar n(k)}\right\|_{W_1}\leq \lim_{k\to\infty}\frac{1}{\bar n(k)}\left\|u_{t_{\bar n(k)},\bar n(k)}-\bar t^{\otimes \bar n(k)}\right\|_{W_1}=0\, .
    \ee
    Then, the source $(\bar \rho_n)_{n\geq 1}$, \ludo{where $\rho_n\in \mathcal{D}(\mathcal{H}^{\otimes n})$ is} defined as
    \bb
        \bar  \rho_n\coloneqq\begin{cases}
         \rho_{\bar n(k),t_{\bar n(k)}} & \text{if } n=\bar n(k) \text{ for some } k\geq 1,\\
             \rho_{\bar t}^{\otimes n}&\text{otherwise,}
        \end{cases}
    \ee
    is a $W_1$-almost i.i.d.\ source along $\rho_{\bar t}$. As a consequence,
    \bb\label{eq:star}
        \lim_{k\to \infty} \frac{1}{n(k)} D_H^{\delta}(\rho_{n(k),t_{n(k)}}\|\Sigma_{n(k)})&=\lim_{k\to \infty} \frac{1}{\bar n(k)} D_H^{\delta}(\rho_{\bar n(k),t_{\bar n(k)}}\|\Sigma_{\bar n(k)})\\
        &\geq \liminf_{k\to \infty} \frac{1}{\bar n(k)} D_H^{\delta}(\rho_{\bar n(k),t_{\bar n(k)}}\|\mathcal{S}^{(\bar n(k))})\\
        &\eqt{($\star$)} \lim_{n\to \infty} \frac{1}{ n} D_H^{\delta}(\bar \rho_n\|\mathcal{S}^{(n)})\\
        &\eqt{($\star$)} D^{\infty}(\rho_{\bar t}\|\mathcal{S})\\
        &\geq \inf_{\rho\in{\rm conv}(\pazocal{R}_1)}D^{\infty}(\rho\|\mathcal{S}),
    \ee
    where the two steps marked with ($\star$) follow from Theorem~\ref{thm:W1_GQSL}.
    This concludes the proof of the claim~\eqref{eq:claim}.
Now, combining~\eqref{eq:inequalities} with~\eqref{eq:claim}, we get 
\bb
    \inf_{\rho\in{\rm conv}(\pazocal{R}_1)}D^{\infty}(\rho\|\mathcal{S}) \leq \liminf_{n\to\infty} \frac{1}{n} D_H^{\delta}(\rho_{n,t_n}\|\Sigma_n)\leq \lambda
\ee
which contradicts the initial assumption $\displaystyle{\inf_{\rho\in{\rm conv}(\pazocal{R}_1)}D^{\infty}(\rho\|\mathcal{S}) >\lambda }$. Hence, 
\bb
    \liminf_{n\to\infty}\frac{1}{n}\, \rel{D_{\max}^{\epsilon,P}}{{\rm conv}\big(\pazocal{R}_n^{\rm a.v.}\big)}{ \mathcal{S}^{(n)}} \geq \inf_{\rho\in{\rm conv}(\pazocal{R}_1)}D^{\infty}(\rho\|\mathcal{S}),
\ee
and, once again by the weak/strong converse duality (Lemma~\ref{lemma:wsc}) this concludes the proof of Theorem~\ref{thm:av} in the case of a finite set $\pazocal R_1$.\bigskip

\noindent\textbf{Step 2. Extension to the case $|\pazocal{R}_1|=+\infty$.} We want to make an outer approximation of the set $\pazocal R_1$ with a polytope contained in the state space, in order to reduce this second case to the first one. However, already in the simplest case $\pazocal R_1=\mathcal{D}(\mathcal{H})$, this is clearly not possible, due to the fact that $\mathcal{D}(\mathcal{H})$ is not itself a polytope. However, we can bypass this problem as follows. Let
\bb
\Delta_{p,\sigma_{\rm full}}(X)\coloneqq (1-p)X + p \Tr[X] \sigma_{\rm full}
\ee
be the depolarising channel with respect to the full rank state $\sigma_{\rm full}\in \pazocal S^{(1)}$. It is easy to prove that, for all $n\geq 1$,
\bb\label{eq:sopra0}
   \Delta_{p,\sigma_{\rm full}}^{\otimes n}\big(\mathcal{S}^{(n)}\big)\subseteq \mathcal{S}^{(n)} 
   \qquad\text{and}\qquad
\Delta_{p,\sigma_{\rm full}}^{\otimes n}\big({\rm conv}\big(\pazocal{R}_n^{\rm a.v.}\big)\big)={\rm conv}\big(\pazocal{R}_{n,p}^{\rm a.v.}\big),
\ee
where $\pazocal{R}_{n,p}^{\rm a.v.} \coloneqq \{ \rho^{(1)}\otimes\cdots\otimes\rho^{(n)}:\ \rho^{(i)}\in \pazocal{R}_{1,p}\}$ and $\pazocal R_{1,p}\coloneqq\Delta_{p,\sigma_{\rm full}}(\pazocal{R}_1)$. By the data-processing inequality and by~\eqref{eq:sopra0}, we have
\bb\label{eq:sopra3}
    \rel{D_H^{\epsilon}}{{\rm conv}\big(\pazocal{R}_n^{\rm a.v.}\big)}{\mathcal{S}^{(n)}} &\geq \rel{D_H^{\epsilon}}{\Delta_{p,\sigma_{\rm full}}^{\otimes n}\big({\rm conv}\big(\pazocal{R}_n^{\rm a.v.}\big)\big)}{\Delta_{p,\sigma_{\rm full}}^{\otimes n}\big(\mathcal{S}^{(n)}\big)} \\
    &\geq \rel{D_H^{\epsilon}}{{\rm conv}\big(\pazocal{R}_{n,p}^{\rm a.v.}\big)}{\mathcal{S}^{(n)}}\, .
\ee

Fix an arbitrary $0<p<1$. Then the closure of the (pre-compact) set $\pazocal R_{1,p}$ lies entirely in the interior of $\mathcal{D}(\mathcal{H})$. For every arbitrary $0<\delta<1$, it is possible to construct an outer approximation of ${\rm conv}\big(\pazocal R_{1,p}\big)$ with trace distance error at most $\delta$ given by the convex hull of a finite number of elements $\{\rho_1,\rho_2,\dots, \rho_{N_{p,\delta}}\}\eqcolon \pazocal R_{1,p}^{(\delta)}\subseteq \mathcal{D}(\mathcal{H})$ (see~\cite[Lemma~2.20]{codenotti2021generalisedflatnessconstantsframework} and~\cite[Theorem~1.8.16]{Schneider2013-bt}), namely
\bb\label{eq:sopra4}
{\rm conv}\big(\pazocal R_{1,p}\big)\subseteq {\rm conv}\big(\pazocal R_{1,p}^{(\delta)}\big)
    \qquad\text{and}\qquad  \displaystyle{\sup_{\rho\in{\rm conv}(\pazocal R_{1,p}^{(\delta)})}\inf_{\rho'\in{\rm conv}(\pazocal R_{1,p})}}\tfrac{1}{2}\|\rho-\rho'\|_1\leq \delta.
\ee
In particular, on one hand, we have
    ${\rm conv}\big(\pazocal{R}_{n,p}^{\rm a.v.}\big)\subseteq {\rm conv}\big(\pazocal{R}_{n,p,\delta}^{\rm a.v.}\big)$ for all $n\geq 1$,
with $\pazocal{R}_{n,p,\delta}^{\rm a.v.} \coloneqq \{ \rho^{(1)}\otimes\cdots\otimes\rho^{(n)}:\ \rho^{(i)}\in \pazocal{R}_{1,p}^{(\delta)}\}$, which implies
\bb\label{eq:pre_asymptotic}
\liminf_{n\to\infty}\frac{1}{n}\,\rel{D_H^{\epsilon}}{{\rm conv}\big(\pazocal{R}_n^{\rm a.v.}\big)}{ \mathcal{S}^{(n)}} &\geqt{\eqref{eq:sopra3}}
\liminf_{n\to\infty}\frac{1}{n}\,\rel{D_H^{\epsilon}}{{\rm conv}\big(\pazocal{R}_{n,p}^{\rm a.v.}\big)}{\mathcal{S}^{(n)}} \\
&\geqt{\eqref{eq:sopra4}}\liminf_{n\to\infty}\frac{1}{n}\,\rel{D_{H}^{\epsilon}}{{\rm conv}\big(\pazocal{R}_{n,p,\delta}^{\rm a.v.}\big)}{ \mathcal{S}^{(n)}} \\
    &=\inf_{\rho\in{\rm conv}\left(\pazocal{R}_{1,p}^{(\delta)}\right)}D^{\infty}(\rho\|\mathcal{S}),
    \ee
    where the last identity follows from Theorem~\ref{thm:av} for the finite set $\pazocal R_{1,\delta}$. On the other hand, by Assumption~\ref{ass:standing-Sn}, we can apply the continuity of the regularised relative entropy of resource (Corollary~\ref{cor:3.1}): letting $\delta \to 0$ and, later, $p\to 0$ in~\eqref{eq:pre_asymptotic}, we have
    \bb\label{eq:lower_bound}
    \liminf_{n\to\infty}\frac{1}{n}\,\rel{D_H^{\epsilon}}{{\rm conv}\big(\pazocal{R}_n^{\rm a.v.}\big)}{\mathcal{S}^{(n)}} &\geq \inf_{\rho \in {\rm conv}(\pazocal R_1)}D^{\infty}(\rho\|\mathcal{S}),
    \ee
    which completes the proof of Theorem~\ref{thm:av}. 
    \end{proof}

\begin{rem}
    Note that the robust generalised quantum Stein's lemma for MSR almost i.i.d.\ sources~\cite[Theorem~4.3]{Mazzola_2026_2} is not sufficient to complete the computation in~\eqref{eq:star}, as the source $(\rho_{\bar n(k),t_{\bar n(k)}})_{k\geq 1}$ is not MSR almost i.i.d.\ in general~\cite[Proposition~25]{almost_iid}.
\end{rem}

\section{Proof of Theorem~\ref{thm:W1_GQSL}}\label{sec:big_proof}

The purpose of this section is to isolate the technical tools that allow the proof of the GQSL in Ref.~\cite{Hayashi2025} to be applied with a non-i.i.d.\ null source.  
More precisely, in our setting the null state is an arbitrary Wasserstein almost i.i.d.\ source $(\rho_n)_n$ along a fixed state $\rho\in\mathcal{D}(\mathcal H)$.
Thus every information-spectrum argument, pinching technique, and likelihood-ratio estimate have to be proved directly for \(\rho_n\).

Here, we isolate the three ingredients that make the argument stable
under Wasserstein almost i.i.d. perturbations of the null source.
\begin{itemize}
    \item Proposition~\ref{prop:info-spectrum-W1}. We prove an intrinsic information-spectrum asymptotic equipartition property for $W_1$ almost i.i.d.\ source \((\rho_n\): with respect to the spectral measure of $\rho_n$, the random variable $-\frac 1n\log\rho_n$  converges in probability to the von Neumann entropy of $\rho$; this result simultaneously strengthens the convergence in expectation of~\cite[Proposition~21]{almost_iid} and the convergence in probability of the random variable $-\frac 1n\log\rho^{\otimes n}$, which only acts locally, given in~\cite[Lemma~3]{datta2026entropyconcentrationuniversaltypicality}.
    \item Proposition~\ref{prop:pinching-extension}. We show that the previous result is stable under any pinching map with a sub-exponential number of outcomes. The pinching technique will be a crucial tool in the proof of the generalised quantum Stein's lemma in order to reduce the quantum setting to the classical case.
    \item Proposition~\ref{prop:padded-block-product-strong-converse}. We combine the preceding information-spectrum input with a padded block-product log-likelihood moment estimate to prove a $W_1$-stable strong converse for block-product alternative hypotheses.
    Note that this is the replacement for the one used in the Hayashi--Yamasaki proof~\cite[Lemma S6]{Hayashi2025}.
\end{itemize}

\subsection[Ouverture. $W_1$-information-spectrum stability and consequences]{Ouverture. \texorpdfstring{$\boldsymbol{W_1}$}{W1}-information-spectrum stability and consequences}

\label{sec:info_spec}

We start with some preliminaries needed in the information spectrum framework.
For a state \(\rho_n\) on \(\mathcal{H}_n\), we regard
\bb
    Z_n(\rho_n)\coloneqq -\frac1n\log \rho_n
\ee
as a random variable $Z_n$ with respect to the spectral measure of \(\rho_n\). Thus, if
\bb
    \rho_n=\sum_i \lambda_i^{(n)}P_i^{(n)}
\ee
is the spectral decomposition of $\rho_n$, with $\lambda_i^{(n)}\neq \lambda_j^{(n)}$ for all $i\neq j$, then 
\bb
    \PP{\rho_n}\left(Z_n=-\tfrac 1n \log\lambda_i^{(n)}\right)=\lambda_i^{(n)}\Tr P_i^{(n)}.
\ee
Clearly, zero eigenvalues do not contribute. We say that $Z_n$ \emph{converges in probability} to the constant $z_\infty\in\mathbb{R}$ with respect to the spectral measure of \(\rho_n\) if
\bb\label{eq:conv_prob}
    \lim_{n\to\infty}\PP{\rho_n}\left(\big|Z_n-z_\infty\big|> \delta\right)=0 \qquad \forall \delta>0,
\ee
and we write
\bb
    -\frac 1n \log \rho_n \xrightarrow{\mathsf P_{\rho_n}} z_\infty.
\ee
Calling $\Lambda_{>0}\coloneqq\{i:\lambda_i^{(n)}>0\}$, we have
\bb
    \PP{\rho_n}\left(\big|Z_n-z_\infty\big|> \delta\right)&=
    \sum_{i\in \Lambda_{>0}}\lambda_i^{(n)}\chi_{\big|-\tfrac 1n \log\lambda_i^{(n)}-z_\infty\big|>\delta}\Tr\big[P_i^{(n)}\big]\\
    &=\sum_{i\in \Lambda_{>0}}\lambda_i^{(n)}\Tr\left[P_i^{(n)}\left\{\left|-\tfrac 1n \log\rho_n-z_\infty\right|>\delta\right\}\right]\\
    &=\Tr\left[\rho_n \left\{
        \left\lvert-\tfrac1n\log \rho_n-z_\infty\right\rvert>\delta
    \right\}\right]
\ee
where $\chi_{a>b}$ is the characteristic function of $a>b$, i.e.\ $\chi_{a>b}=1$ if $a>b$, and $\chi_{a>b}=0$ otherwise. Therefore,~\eqref{eq:conv_prob} is equivalent to 
\bb
    \lim_{n\to\infty}\Tr\left[\rho_n \left\{
        \left\lvert-\tfrac1n\log \rho_n-z_\infty\right\rvert>\delta
    \right\}\right]=0.
\ee
The first key statement that we are going to prove in the remainder of the section is the following claim of stability of the information spectrum for Wasserstein almost i.i.d.\ states. The proof heavily relies on the techniques used in~\cite[Section~9]{De_Palma_2023b} to prove the tight continuity of the von Neumann entropy with respect to the quantum Wasserstein distance of order 1.

\begin{boxedprop}{}
\begin{prop}[(Information-spectrum stability for Wasserstein almost i.i.d.\ sources)]
\label{prop:info-spectrum-W1}
Let \(\rho\) be a state on \(\mathcal H\) and let $(\rho_n)_n$ be a Wasserstein almost i.i.d.\ source along $\rho$. Then
\bb
    -\frac1n\log \rho_n
    \xrightarrow{\mathsf P_{\rho_n}}
    S(\rho),
\ee
namely, for every \(\delta>0\),
\bb
    \lim_{n\to\infty}\Tr \left[\rho_n\,
    \left\{
        \left\lvert-\tfrac1n\log \rho_n-S(\rho)\right\rvert>\delta
    \right\}\right]=0.
\ee
\end{prop}
\end{boxedprop}
\begin{proof}
    See Section~\ref{proof:info-spectrum-W1}.
\end{proof}

As a strengthening of Proposition~\ref{prop:info-spectrum-W1} we are going to prove in Proposition~\ref{prop:pinching-extension} a stability result when pinching the source by a subexponential number of projectors, as discussed in Section~\ref{sec:pinching-extension}, which turns out to be a crucial tool in the proof of our main result.

\begin{boxedprop}{}
\begin{prop}[(Information spectrum after subexponential pinching)]
\label{prop:pinching-extension}
Let \(\rho\) be a state on \(\mathcal H\) and let $(\rho_n)_n$ be a Wasserstein almost i.i.d.\ source along $\rho$.
For each $n$, let \(\{E_{n,j}\}_{j=1}^{m_n}\) be a family of mutually orthogonal projectors summing to the identity such that
\bb
    \lim_{n\to\infty}\frac{\log m_n}n=0.
\ee
Define the pinched state
\bb
    \widehat{\rho}_n\coloneqq\sum_{j=1}^{m_n} E_{n,j}\rho_nE_{n,j}.
\ee
Then
\bb
    -\frac1n\log \widehat{\rho}_n
    \xrightarrow{\mathsf P_{\widehat{\rho}_n}}
    S(\rho),
\ee
namely, for every \(\delta>0\),
\bb
    \lim_{n\to\infty}\Tr \widehat{\rho}_n\left\{
        \left\lvert-\tfrac1n\log \widehat{\rho}_n-S(\rho)\right\rvert>\delta
    \right\}
    =0.
\ee
\end{prop}
\end{boxedprop}

\begin{proof}
    See Section~\ref{sec:pinching-extension}.
\end{proof}

Finally, the following statement is a key result which will be needed both in the converse and in the achievability part of the proof of the $W_1$ robust generalised quantum Stein's lemma, as schematically represented in Figure~\ref{fig:converse} and~\ref{fig:achievability}, respectively.

\begin{boxedprop}{}
\begin{prop}[(\(W_1\)-stable strong converse for block alternative)]
\label{prop:padded-block-product-strong-converse}
Fix a block length \(m\ge 1\) and a full-rank state
\(\sigma_m\in\mathcal S^{(m)}\). For \(n=m\ell+r\), \(0\le r<m\), define
\bb\label{eq:block_pad}
  \sigma_n^{(m)}\coloneqq\sigma_m^{\otimes\ell}\otimes\sigma_{\rm full}^{\otimes r}.
\ee
Let \((\rho_n)_n\) be Wasserstein almost i.i.d.\ along \(\rho\).
Then, for every \(\eta\in(0,1)\),
\bb
  \limsup_{n\to\infty}
  \frac1n \rel{D_H^\eta}{\rho_n}{\sigma_n^{(m)}}
  \le
  \frac1mD(\rho^{\otimes m}\|\sigma_m).
  \label{eq:stable_strong_converse}
\ee
\end{prop}
\end{boxedprop}

\begin{rem}
The fact that the limit of the left-hand side of~\eqref{eq:stable_strong_converse} as $\eta \to 0^+$ is at most equal to the right-hand side can be shown elementarily by using the data-processing inequality on the hypothesis testing relative entropy. The main point of the above statement, however, is to show that the same holds for all $\eta\in (0,1)$.
\end{rem}

\begin{proof}
    See Section~\ref{subsec:padded-block-product} and Section~\ref{proof:padded-block}.
\end{proof}

\subsubsection{The dual formulation of the quantum Wasserstein distance of order $1$}

Let us recall the dual formulation of the quantum Wasserstein distance of order 1 (see~\cite{De_Palma_2021} for a detailed discussion). For a self-adjoint operator \(H\) on \(\mathcal{H}_n\), define its dependence on site \(x\) by
\begin{equation}
    \partial_x H
    :=
    2\inf_{K\in \pazocal{O}_{[n]\setminus\{x\}}} \lVert H-K\rVert_\infty,
    \label{eq:site-dependence}
\end{equation}
where \(\pazocal{O}_{[n]\setminus\{x\}}\) denotes the set of self-adjoint operators acting trivially on site \(x\in[n]\). The Lipschitz constant is
\bb
    \lVert H\rVert_L:=\max_{x=1,\ldots,n}\partial_xH.
\ee
The quantum \(W_1\) distance is equivalently characterised, with this normalization of \(\partial_x\), by the dual formula
\begin{equation}
    \|\rho-\sigma\|_{W_1}
    =
    \sup_{\lVert H\rVert_L\le1}\Tr\big[(\rho-\sigma)H\big].
    \label{eq:W1-duality}
\end{equation}
We shall use the following elementary consequence of the definition of \(\partial_x\).

\begin{lemma}[(Positive-operator dependence bound)]
\label{lem:positive-dependence-bound}
If \(C\ge0\), then for every site \(x\),
\bb
    \partial_x C\le \lVert C\rVert_\infty.
\ee
\end{lemma}

\begin{proof}
Since \(0\le C\le \lVert C\rVert_\infty \id\), we have
\bb
    -\frac{\lVert C\rVert_\infty}{2}\id
    \le
    C-\frac{\lVert C\rVert_\infty}{2}\id
    \le
    \frac{\lVert C\rVert_\infty}{2}\id.
\ee
Therefore, taking the site-independent approximant \(\lVert C\rVert_\infty \frac \id2\) in~\eqref{eq:site-dependence} gives \(\partial_xC\le \lVert C\rVert_\infty\).
\end{proof}

\subsubsection{Fattening, distance cutoffs and blowing-up}

\begin{figure}[t]
  \centering
  \def\svgwidth{0.92\linewidth}
  \small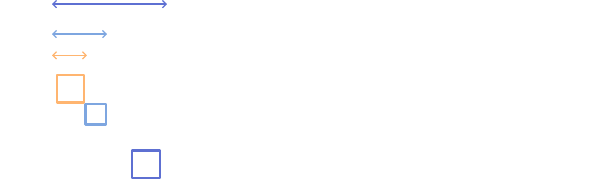
  \mycaption{A pictorial representation of the action of $D_{\pazocal V}$, $A_{\pazocal V,r}$ and $B_{\pazocal V,r}.$}{} 
  \label{fig:fattening}
\end{figure}

\begin{Def}[{(Fattening of a subspace~\cite[Section~2]{quantumTalagrand},~\cite[Definition~15]{Eldar_2017},~\cite[Definition~9.1]{De_Palma_2023b})}]
Let \(\pazocal V\subseteq \mathcal{H}_n\) be a subspace. For \(k=0,1,\ldots,n\), we define the $k$-\emph{fattening} $\pazocal V_k$ of the subspace $\pazocal V$ as
\bb
    \pazocal V_k
    \coloneqq
    \operatorname{span}
    \left\{
        A_S\ket{\psi}:
        \ket{\psi}\in \pazocal V,
        \ A\in \mathcal{L}(\mathcal{H}_S),
        \ S\subseteq\{1,\ldots,n\},
        |S|\le k
    \right\}.
\ee
Thus 
\bb
\pazocal V=\pazocal V_0\subseteq \pazocal V_1\subseteq\cdots\subseteq \pazocal V_n=\mathcal{H}_n.
\ee
Let \(\Pi_k\) be the orthogonal projector onto \(\pazocal V_k\), with \(\Pi_{-1}=0\). The \emph{distance operator} from \(\pazocal V\) is
\bb
    D_{\pazocal V}\coloneqq \sum_{k=0}^n k(\Pi_k-\Pi_{k-1}).
\ee
For an integer \(r\) with \(1\le r\le n\), define
\bb\label{eq:cutoff0}
    A_{{\pazocal V},r}\coloneqq\min(D_{\pazocal V},r),
    \qquad
    B_{{\pazocal V},r}\coloneqq(r-D_{\pazocal V})_+.
\ee
Equivalently,
\begin{equation}
    A_{{\pazocal V},r}=\sum_{k=0}^{r-1}(\id-\Pi_k),
    \qquad
    B_{{\pazocal V},r}=\sum_{k=0}^{r-1}\Pi_k.
    \label{eq:cutoffs}
\end{equation}
\end{Def}
The operators $D_{\pazocal V}$, $A_{\pazocal V,r}$ and $B_{\pazocal V,r}$ are pictorially represented in Figure~\ref{fig:fattening}. The reason why we are interested in these operators is that -- as stated in the following lemma -- their quantum Lipschitz constant is upper bounded by 1, hence they can be used in the dual formulation of the $W_1$ distance as guesses of $H$ in~\eqref{eq:W1-duality} in order to provide a lower bound. We follow a similar approach to~\cite[Proposition 9.1]{De_Palma_2023b}.

\begin{lemma}[(Lipschitz cutoffs of the distance operator)]
\label{lem:distance-cutoffs-Lipschitz}
For every subspace \(\pazocal V\subseteq \mathcal{H}_n\) and every integer \(1\le r\le n\), we have
\bb
    \lVert A_{V,r}\rVert_L\le1,
    \qquad
    \lVert B_{V,r}\rVert_L\le1.
\ee
\end{lemma}

\begin{proof}
Fix a site \(x\in[n]\). We first isolate the part of the cutoff observables which may depend on the site \(x\in[n]\). For each \(k=0,\dots, n\) and $x\in[n]$, define the intermediate subspace
\bb
    \pazocal V_{k,x}
    \coloneqq
    \operatorname{span}
    \left\{
        A_S\ket{\psi}:
        \ket{\psi}\in \pazocal V,
        \ A\in \mathcal{L}(\mathcal{H}_S),
        \ S\subseteq\{1,\ldots,n\},
        \ |S|\le k,
        \ S\ni x
    \right\},
\ee
and let \(\Pi_{k,x}\) be the orthogonal projector onto \(\pazocal V_{k,x}\).  We record the inclusions explicitly: for all \(k\ge0\),
\bb\label{eq:subspaces}
    \pazocal V_{k-1}\subseteq \pazocal V_{k,x}\subseteq \pazocal V_k.
\ee
Indeed,
\begin{itemize}
    \item For \(k=0\), we use the convention \(\pazocal V_{-1}=\{0\}\), and \(\pazocal V_{-1}\subseteq \pazocal V_{0,x}\subseteq\pazocal V_0\) is immediate.
    \item For \(k\ge1\), if a vector belongs to \(\pazocal V_{k-1}\), then it is a finite linear combination of vectors \(A|\psi\rangle\) with \(A\) acting on a set \(Y\) of at most \(k-1\) sites.  Replacing \(Y\) by \(Y\cup\{x\}\) and tensoring with the identity on site \(x\), the same vector is generated by an operator acting on a set of at most \(k\) sites containing \(x\).  Hence \(\pazocal V_{k-1}\subseteq \pazocal V_{k,x}\).  The inclusion \(\pazocal V_{k,x}\subseteq \pazocal V_k\) is immediate from the definition.
\end{itemize}  
Eq.~\eqref{eq:subspaces} implies
\bb
    \Pi_{k-1}\leq \Pi_{k,x}\leq \Pi_k,
\ee
whence
\begin{equation}
    0\le \Pi_k-\Pi_{k,x}\le \Pi_k-\Pi_{k-1}.
    \label{eq:projection-order}
\end{equation}

Moreover, \(\pazocal V_{k,x}\) is invariant under all unitaries acting only on site \(x\).  Indeed, if \(U_x\) acts on site \(x\) and \(A\) is supported on a set \(X\ni x\), then \(U_xA|\psi\rangle=(U_xA)|\psi\rangle\), and \(U_xA\) is again supported on \(X\). So, if $\{\ket{\psi_i}\}$ is an orthonormal basis for the subspace $\pazocal V_{k,x}$, then also $\{\ket{\phi_i}\}\coloneqq \{U_x\ket{\psi_i}\}$ is an orthonormal basis for $\pazocal V_{k,x}$. This proves that $\Pi_{k,x}$ commutes with all unitaries acting only on $x$:
\bb
    U_x \Pi_{k,x} U^\dagger_x = \sum_iU_x\ketbra{\psi_i}U^\dagger_x=\sum_i\ketbra{\phi_i}=\Pi_{k,x}.
\ee
Since the unitaries on $x$ span all the linear operators $\mathcal{L}(\mathcal{H}_x)$, \(\Pi_{k,x}\) belongs to \(\pazocal{O}_{\{1,\ldots,n\}\setminus\{x\}}\).  Thus each \(\Pi_{k,x}\) is site-independent with respect to \(x\).

We shall now use the following immediate consequence of the definition of \(\partial_x\): if \(K\in\pazocal{O}_{\{1,\ldots,n\}\setminus\{x\}}\), then
\begin{equation}
    \partial_x(H+K)=\partial_xH.
    \label{eq:dependence-translation-invariance}
\end{equation}
Indeed, in the infimum defining \(\partial_x(H+K)\), the change of variables \(K'\mapsto K'-K\) leaves the set \(\pazocal{O}_{\{1,\ldots,n\}\setminus\{x\}}\) invariant.
Define
\bb
    C_x:=\sum_{k=0}^{r-1}(\Pi_k-\Pi_{k,x}).
\ee
By~\eqref{eq:projection-order},
\bb
    0\le C_x
    \le
    \sum_{k=0}^{r-1}(\Pi_k-\Pi_{k-1})
    =
    \Pi_{r-1}
    \le \id.
\ee
Hence \(\lVert C_x\rVert_\infty\le1\). Since \(C_x\ge0\), Lemma~\ref{lem:positive-dependence-bound} gives
\begin{equation}
    \partial_x C_x\le1.
    \label{eq:Cx-dependence-bound}
\end{equation}
We now apply this estimate to the two cutoff observables. For \(B_{V,r}\), using~\eqref{eq:cutoffs},
\bb
    B_{\pazocal V,r}
    =
    \sum_{k=0}^{r-1}\Pi_k
    =
    \sum_{k=0}^{r-1}\Pi_{k,x}+C_x.
\ee
The first sum is site-independent with respect to \(x\). Therefore, by~\eqref{eq:dependence-translation-invariance} and~\eqref{eq:Cx-dependence-bound},
\bb
    \partial_x B_{\pazocal V,r}=
    \partial_x C_x
    \le1.
\ee
Similarly,
\bb
    A_{\pazocal V,r}
    =
    \sum_{k=0}^{r-1}(\id-\Pi_k)
    =
    \sum_{k=0}^{r-1}(\id-\Pi_{k,x})-C_x.
\ee
Again the first sum is site-independent with respect to \(x\). Since \(\partial_x(-C_x)=\partial_xC_x\), we get
\bb
    \partial_x A_{\pazocal V,r}
    =
    \partial_x(-C_x)
    =
    \partial_xC_x
    \le1.
\ee
Since \(x\) was arbitrary, \(\lVert A_{\pazocal V,r}\rVert_L\le1\) and \(\lVert B_{\pazocal V,r}\rVert_L\le1\).
\end{proof}

\begin{lemma}[(Blowing-up estimates)]
\label{lem:quantum-blowing-up}
Let \(\pazocal V\subseteq \mathcal{H}_n\) be a subspace and let \(\pazocal V_r\) be its \(r\)-fattening, where \(1\le r\le n\). For any states \(\rho,\sigma\) on \(\mathcal{H}_n\), we have
\begin{align}
    \Tr \big[\rho\,\Pi_{\pazocal V}\big]
    &\le
    \Tr \big[\sigma\,\Pi_{\pazocal V_r}\big]
    +
    \frac{1}{r}\|\rho-\sigma\|_{W_1},
    \label{eq:BU1}
    \\
    \Tr\big[ \rho\,(\id-\Pi_{\pazocal V_r})\big]
    &\le
    \Tr \big[ \sigma\,(\id-\Pi_{\pazocal V})\big]+
    \frac{1}{r}\|\rho-\sigma\|_{W_1}.
    \label{eq:BU2}
\end{align}
\end{lemma}

\begin{proof}
The identities in~\eqref{eq:cutoffs} show explicitly how the cutoffs  of~\eqref{eq:cutoff0} act on the layers \mbox{\({\rm Ran}(\Pi_k-\Pi_{k-1})\)}: the operator \(B_{\pazocal V,r}\) has eigenvalue \(r-k\) on this layer for \(k<r\) and vanishes for \(k\ge r\), while \(A_{\pazocal V,r}\) has eigenvalue \(\min\{k,r\}\) on the same layer (see also Figure~\ref{fig:fattening}).
By Lemma~\ref{lem:distance-cutoffs-Lipschitz} and the dual formula~\eqref{eq:W1-duality},
\bb\label{eq:bu0}
    \Tr\big[ B_{\pazocal V,r}( \rho- \sigma)\big]
    \le \|\rho-\sigma\|_{W_1}.
\ee
Using the inequality \(r\Pi_{\pazocal V}\le B_{\pazocal V,r}\le r\Pi_{\pazocal V_r}\) in~\eqref{eq:bu0}, we get
\bb
    r\Tr \rho\,\Pi_{\pazocal V}
    \le
    r\Tr \sigma\,\Pi_{\pazocal V_r}+\|\rho-\sigma\|_{W_1},
\ee
which gives~\eqref{eq:BU1}.
Similarly, using \(A_{\pazocal V,r}\) in~\eqref{eq:W1-duality}, we get
\bb\label{eq:bu0}
    \Tr\big[ A_{\pazocal V,r}(\rho-\sigma)\big]
    \le \|\rho-\sigma\|_{W_1},
\ee
and noticing that \(r(\id -\Pi_{\pazocal V_r})\le A_{\pazocal V,r}\le r(\id-\Pi_{\pazocal V})\), we conclude that
\bb
    r\Tr \rho(\id-\Pi_{\pazocal V_r})
    \le
    r\Tr \sigma(\id-\Pi_{\pazocal V})+\|\rho-\sigma\|_{W_1},
\ee
which gives~\eqref{eq:BU2}.
\end{proof}

\begin{lemma}[(Dimension growth under fattening)]
\label{lem:dimension-growth-fattening}
For every subspace \(\pazocal V\subseteq \mathcal{H}_n\) and every \(r=0,1,\ldots,n\),
\bb\label{eq:first_bound}
    \dim \pazocal V_r
    \le
    \dim \pazocal V \sum_{j=0}^r \binom nj(d^2-1)^j.
\ee
Consequently, for every \(0<\epsilon<1-\frac 1{d^2}\),
\bb
    \frac1n\log \dim \pazocal V_{r_n}
    \le
    \frac1n\log \dim \pazocal V
    +h_2(\epsilon)+\epsilon\log(d^2-1)+o(1),
\ee
where \(r_n\coloneqq \lfloor \epsilon n\rfloor\) and \(h_2\) is the binary entropy.
\end{lemma}

\begin{proof}
Choose a basis \(\{A_1,\ldots,A_{d^2}\}\) of \(\pazocal{B}(\mathbb C^d)\) with \(A_1=\id\). For \(a=(a_1,\ldots,a_n)\), set \(A_a=A_{a_1}\otimes\cdots\otimes A_{a_n}\) and \(\#a=|\{i:a_i\ne 1\}|\). Operators acting on at most \(r\) sites are spanned by \(A_a\) with \(\#a\le r\). Therefore,
\bb
    \pazocal V_r= \operatorname{span}\{A_a|\psi\rangle: |\psi\rangle\in \pazocal V,
    \#a\le r\}.
\ee
The first bound~\eqref{eq:first_bound} follows by counting words $a$ such that \(\#a\leq r\). Indeed, the words with \(\#a= j\leq r\) are exactly
\bb
    \underbrace{\binom{n}{j}}_{\substack{\text{positions $i$ of the}\\ \text{symbols $a_i>1$}}}\times \underbrace{(d^2-1)^j}_{\substack{\text{choice of the symbols}\\ \text{in the selected positions}}},
\ee
and summing over $j\in\{0,\dots, r\}$ yields~\eqref{eq:first_bound}.
The asymptotic bound follows from standard entropic estimates on binomial coefficients. 
\end{proof}

\begin{lemma}
\label{lem:typical-subspace}
Let \(\sigma_n=\sigma^{\otimes n}\). For every \(\delta>0\), there are projectors \(T_{n,\delta}\) such that
\bb\label{eq:typical}
    \lim_{n\to\infty}\Tr \sigma_nT_{n,\delta}=1,
    \qquad
    \dim T_{n,\delta}\le e^{n(S(\sigma)+\delta)}
\ee
for all sufficiently large $n$. Moreover, if \(M_n\subseteq \mathcal{H}_n\) satisfies \(\dim M_n\le e^{n(S(\sigma)-\gamma)}\) for some \(\gamma>0\), then
\bb
    \lim_{n\to\infty }\Tr \sigma_n\Pi_{M_n}=0.
\ee
\end{lemma}

\begin{proof}
See e.g.\ \cite[Chapter 15]{MARK}.
\end{proof}

\begin{lemma}[(Uniform small-subspace bound for \(\rho_n\))]
\label{lem:rho-small-subspaces}
Assume that $(\rho_n)_n$ is a Wasserstein almost i.i.d.\ source along $\rho$. Let \(\gamma>0\). If \(M_n\subseteq \mathcal{H}_n\) is any sequence of subspaces satisfying
\bb
    \dim M_n\le e^{n(S(\rho)-\gamma)},
\ee
then
\bb
    \lim_{n\to \infty}\Tr\rho_n\Pi_{M_n} = 0.
\ee
\end{lemma}

\begin{proof}
Choose a sufficiently small \(\epsilon>0\) such that \(h_2(\epsilon)+\epsilon\log(d^2-1)<\gamma/2\), and set \(r_n=\lfloor \epsilon n\rfloor\). By Lemma~\ref{lem:dimension-growth-fattening}, the dimension of the $r_n$-fattening of $M_n$ is upper bounded as
\modifica{
\bb
    \dim \big((M_n)_{r_n}\big)
    \le
    e^{n(S(\rho)-\gamma +h_2(\epsilon)+\epsilon\log(d^2-1)+o(1)\}}
\ee}
Thus, for some \(\gamma'>0\) and all sufficiently large $n$,
\bb
    \dim (M_n)_{r_n}\le e^{n(S(\rho)-\gamma')}.
\ee
The second claim of Lemma~\ref{lem:typical-subspace} 
gives
\bb
    \lim_{n\to \infty}\Tr\rho^{\otimes n}\Pi_{(M_n)_{r_n}} = 0.
\ee
Applying the blowing-up estimate~\eqref{eq:BU1} with \(\pazocal V=M_n\) and \(r=r_n\), we obtain
\bb
    \Tr\rho_n\Pi_{M_n}
    \le
    \Tr\rho^{\otimes n}\Pi_{(M_n)_{r_n}}
    +
    \frac{1}{r_n}\left\|\rho_n-\rho^{\otimes n}\right\|_{W_1}
    \xrightarrow{n\to\infty}0,
\ee
where the last term vanishes since $r_n$ grows linearly in $n$.
\end{proof}

\subsubsection{Proof of Proposition~\ref{prop:info-spectrum-W1}}\label{proof:info-spectrum-W1}

We prove Proposition~\ref{prop:info-spectrum-W1} by showing that its upper and lower spectral tails are bounded, where the upper and lower spectral tails are treated differently.
For convenience, we write
    $w_n\coloneqq \|\rho_n-\rho^{\otimes n}\|_{W_1}$,
and by assumption, \(\displaystyle{\lim_{n\to \infty }\tfrac{w_n}n = 0}\). We separate the two contributions
\bb
    &\Tr \rho_n\left\{
        \left\lvert-\tfrac1n\log\rho_n-S(\rho)\right\rvert>\alpha
    \right\}\\[0.3em]
    &\qquad\qquad\qquad=\underbrace{\Tr \rho_n \left\{-\tfrac1n\log\rho_n>S(\rho)+\alpha\right\}}_{\text{upper information tail}}+\underbrace{\Tr \rho_n\left\{-\tfrac1n\log\rho_n<S(\rho)-\alpha\right\}}_{\text{lower information tail}}
\ee
and we bound them separately.

\paragraph{Upper information tail.} 
Fix \(\alpha>0\), choose \(\delta>0\) and a sufficiently small \(\epsilon>0\) such that 
\bb
\delta+h_2(\epsilon)+\epsilon\log(d^2-1)<\alpha.
\ee
Let \(T_{n,\delta}\) be the typical projector of \(\rho^{\otimes n}\), set \(r_n=\lfloor \epsilon n\rfloor\), and call $Q_n$ the projector onto the $r_n$-fattening of the $(n,\delta)$-typical subspace, namely
\bb
    Q_n\coloneqq\Pi_{({\rm Ran} T_{n,\delta})_{r_n}}.
\ee
Combining the dimension bounds of Lemma~\ref{lem:dimension-growth-fattening} and Lemma~\ref{lem:typical-subspace}, 
\bb
    R_n\coloneqq {\rm rk} Q_n
    \le
    e^{n(S(\rho)+\delta+h_2(\epsilon)+\epsilon\log(d^2-1)+o(1))},
    \label{eq:Rn-bound}
\ee
and by~\eqref{eq:BU2},
\bb
    \eta_n\coloneqq\Tr \rho_n(\id-Q_n)
    \le
    \Tr \rho^{\otimes n}(\id-T_{n,\delta})+\frac{w_n}{r_n}
    \xrightarrow{n\to\infty} 0.
\ee
due to~\eqref{eq:typical} and to the assumption that $(\rho_n)_n$ is a Wasserstein almost i.i.d.\ source along $\rho$.
Now, let \(\lambda_1^{(n)}\ge\lambda_2^{(n)}\ge\cdots\) be the $d^n$ eigenvalues of \(\rho_n\), counted with multiplicity. Then,\footnote{This can also be seen as an application of Ky Fan's maximum principle~\cite[Exercise II.1.13]{BHATIA-MATRIX}.}
\bb\label{eq:trace}
    \sum_{i=1}^{R_n}\lambda_i^{(n)}=\max_{\substack{0\leq P\leq \id\\ \rk P\leq R_n}}\Tr \rho_n P\ge \Tr \rho_nQ_n=1-\eta_n\qquad \implies \qquad \sum_{i>R_n}\lambda_i^{(n)}\le\eta_n.
\ee
Whence,
\bb\label{eq;ineq_lambda}
    \Tr \rho_n \left\{\rho_n<e^{-n(S(\rho)+\alpha)}\right\}
    &=
    \sum_{i}\lambda_i^{(n)}\chi_{\lambda_i^{(n)}< e^{-n(S(\rho)+\alpha)}}\\
    &\le 
    \sum_{i>R_n}\lambda_i^{(n)} + \sum_{1\le i\le R_n} \modifica{\min}\{\lambda_i^{(n)},e^{-n(S(\rho)+\alpha)}\}\\
    &\le 
    \eta_n+R_ne^{-n(S(\rho)+\alpha)},
\ee
where $\chi_{a>b}$ is the characteristic function of $a>b$.
By~\eqref{eq:Rn-bound},
\bb
    R_ne^{-n(S(\rho)+\alpha)}
    \le
    e^{-n\big(\alpha-\delta-h_2(\epsilon)-\epsilon\log(d^2-1)-o(1)\big)}
    \xrightarrow{n\to\infty}0.
\ee
Therefore,
\bb
    \lim_{n\to\infty}\Tr \rho_n \left\{-\frac1n\log\rho_n>S(\rho)+\alpha\right\} = 0,
\ee
where $\alpha>0$ was arbitrary since the beginning.

\paragraph{Lower information tail.} The idea is that the spectral projector of \(\rho_n\) corresponding to
eigenvalues larger than \(e^{-n(S(\rho)-\alpha)}\) has dimension at most
\(e^{n(S(\rho)-\alpha)}\). The uniform small-subspace bound (Lemma~\ref{lem:rho-small-subspaces}) then implies that \(\rho_n\)-mass vanishes to this projector.
If \(S(\rho)<\alpha\), then \(e^{-n(S(\rho)-\alpha)}>1\): the lower tail is empty. Thus, we may assume \(S(\rho)\geq \alpha\) and get
\bb
    \text{rk} \left\{\rho_n>e^{-n(S(\rho)-\alpha)}\right\}\le e^{n(S(\rho)-\alpha)}.
\ee
By Lemma~\ref{lem:rho-small-subspaces}, we get
\bb
    \lim_{n\to\infty}\Tr \rho_n \left\{\rho_n>e^{-n(S(\rho)-\alpha)}\right\}
    =0,
\ee
which implies
\bb
    \lim_{n\to\infty}\Tr \rho_n\left\{-\frac1n\log\rho_n<S(\rho)-\alpha\right\}=0.
\ee
Combining the upper and lower tail, for every \(\alpha>0\),
\bb
    \lim_{n\to\infty}\Tr \rho_n\left\{
        \left\lvert-\frac1n\log\rho_n-S(\rho)\right\rvert>\alpha
    \right\}
    =0.
\ee
This completes the proof of Proposition~\ref{prop:info-spectrum-W1}.

\subsubsection{Proof of Proposition~\ref{prop:pinching-extension}}
\label{sec:pinching-extension}
The proof of the generalised quantum Stein's lemma in~\cite{Hayashi2025} repeatedly replaces a state by a pinched version with respect to a rounded auxiliary state. For $W_1$ sources, we cannot assume that this global pinching is $W_1$-contractive, as the quantum Wasserstein distance of order 1 is guaranteed to satisfy the data-processing inequality \emph{only for single-system channels}. Instead, we prove a different statement: if the pinching has only subexponentially many outcomes, then it cannot change the information-spectrum rate. 
This is exactly the regime produced by logarithmic spectral rounding, where the number of rounded
eigenvalue bins is polynomial.\smallskip

 The proof uses the upper- and lower-tail decomposition as above, together with the fact that pinching by \(m_n\) blocks can increase the relevant supporting dimension by at most the factor \(m_n\).

First, let us consider the upper information tail and fix \(\alpha>0\). Similarly to the proof of the bound on the upper tail in Proposition~\ref{prop:info-spectrum-W1}, we choose $\delta,\epsilon>0$ small enough so that
\bb
    \gamma\coloneqq \delta+h_2(\epsilon)+\epsilon\log(d^2-1)<\alpha;
\ee
then, there exists a sequence of projectors \(Q_n\) such that
\bb
    \lim_{n\to\infty}\Tr \rho_n Q_n=1,
    \qquad
    {\rm rk} Q_n\le e^{n\big(S(\rho)+\gamma+o(1)\big)}.
\ee
Let \(\widetilde{Q}_n\) be the projector onto the smallest subspace invariant under all the projectors \(E_{n,j}\) and containing \({\rm Ran} Q_n\). Equivalently,
\bb\label{eq:construction}
    {\rm Ran} \widetilde{Q}_n
    = \left\{
    \sum_{j=1}^{m_n} E_{n,j}\ket{\psi_j}:\ket{\psi_1},\dots, \ket{\psi_{m_n}} \in {\rm Ran} Q_n \right\}\supseteq {\rm Ran} Q_n.
\ee
The summands are mutually orthogonal because the projectors \(E_{n,j}\) are mutually orthogonal. This subspace is reducing for every \(E_{n,j}\); hence its orthogonal projector \(\widetilde{Q}_n\) commutes with every \(E_{n,j}\). It contains \(Q_n\), and its rank, by construction~\eqref{eq:construction}, satisfies
\bb
    {\rm rk} \widetilde{Q}_n\le m_n {\rm rk} Q_n
    \le
    e^{n\big(S(\rho)+\gamma+o(1)\big)},
\ee
because \(\log m_n=o(n)\). Since \(\widetilde{Q}_n\) commutes with the pinching projectors, the trace identity used below can be checked directly:
\bb
\begin{aligned}
    \Tr \widehat{\rho}_n \widetilde{Q}_n
    =
    \sum_{j=1}^{m_n}\Tr E_{n,j}\rho_nE_{n,j}\widetilde{Q}_n  
    =
    \sum_{j=1}^{m_n}\Tr \rho_nE_{n,j}\widetilde{Q}_nE_{n,j}  
    =
    \Tr \rho_n\widetilde{Q}_n.
\end{aligned}
\ee
Therefore
\bb
    \Tr \widehat{\rho}_n \widetilde{Q}_n
    =
    \Tr \rho_n \widetilde{Q}_n
    \ge
    \Tr \rho_n Q_n
    \xrightarrow{n\to\infty}1.
\ee
Let \(\widetilde{R}_n={\rm rk} \widetilde{Q}_n\), and let \(\widehat{\lambda}_1^{(n)}\ge \widehat{\lambda}_2^{(n)}\ge\cdots\) be the eigenvalues of \(\widehat{\rho}_n\). Similarly to~\eqref{eq:trace},
\bb
    \sum_{i=1}^{\widetilde{R}_n}\widehat{\lambda}_i^{(n)}
    \ge
    \Tr \widehat{\rho}_n \widetilde{Q}_n
    =1-o(1).
\ee
Proceeding as in~\eqref{eq;ineq_lambda}, we get
\bb
    \Tr \widehat{\rho}_n\left\{\widehat{\rho}_n<e^{-n(S(\rho)+\alpha)}\right\}
    \le
    o(1)+\widetilde{R}_ne^{-n(S(\rho)+\alpha)},
\ee
and recalling that \(\gamma<\alpha\) and \(R_n\le e^{n\big(S(\rho)+\gamma+o(1)\big)}\), gives
\bb
    \lim_{n\to\infty}\Tr \widehat{\rho}_n\left\{-\frac1n\log \widehat{\rho}_n>S+\alpha\right\}=0.
\ee
We next prove the bound on the lower information tail. Let
\bb
    \widehat{P}_n^-\coloneqq\left\{\widehat{\rho}_n>e^{-n(S(\rho)-\alpha)}\right\},
    \qquad
    \widehat{L}_n\coloneqq{\rm Ran} \widehat{P}_n^-.
\ee
If \(S(\rho)<\alpha\), then \(\widehat{P}_n^-=0\), and there is nothing to prove. Thus we may assume \(S(\rho)\geq\alpha\). Then \({\rm rk} \widehat{L}_n\le e^{n(S(\rho)-\alpha)}\). Since \(\widehat{\rho}_n\) is block diagonal with respect to the decomposition \(\sum_jE_{n,j}=\id\), every spectral projector of \(\widehat{\rho}_n\) is also block diagonal. In particular, \(\widehat{P}_n^-\) commutes with every \(E_{n,j}\). Therefore,
\bb
    \Tr \widehat{\rho}_n \widehat{P}_n^-
    =
    \Tr \rho_n \widehat{P}_n^-.
\ee
By Lemma~\ref{lem:rho-small-subspaces}, applied to the subspace \(\widehat{L}_n\),
\bb
    \Tr \rho_n \widehat{P}_n^-
    =
    \Tr \rho_n\Pi_{\widehat{L}_n}
    \xrightarrow{n\to\infty}0.
\ee
Thus,
\bb
    \lim_{n\to\infty}\Tr \widehat{\rho}_n\left\{-\frac1n\log \widehat{\rho}_n<S-\alpha\right\}=0.
\ee
Combining the two tail estimates proves the claim.

\subsubsection{Technical ingredients needed for the proof of Proposition~\ref{prop:padded-block-product-strong-converse}}\label{subsec:padded-block-product}

The last technical tool needed in the proof of Theorem~\ref{thm:W1_GQSL} is the strong converse for padded block-product alternatives given by Proposition~\ref{prop:padded-block-product-strong-converse}. This is the \(W_1\)-stable replacement of the padded i.i.d.\ block-product strong converse used in Ref.~\cite[Lemma S6]{Hayashi2025}.\smallskip

The proof has two main ingredients.
\begin{enumerate}
    \item Lemma~\ref{lem:binary-rounding} gives a binary-rounded pinching to reduce the problem to a commuting pair with only polynomially many spectral values. 
    \item Lemma~\ref{lemma:new_block} provides a noncommutative law of large numbers for padded block-product alternative hypothesis, and it is essentially based on a simple application of~\cite[Lemma~2]{datta2026entropyconcentrationuniversaltypicality}.
\end{enumerate}
Combining these ingredients, we first prove a block strong converse result for tests given by projectors (see~\eqref{eq:projector}); then, we extend the statement to general POVMs using Lemma~\ref{lem:projection-to-DH}.

\begin{figure}[t]
  \centering
  \includegraphics[width=0.92\textwidth]{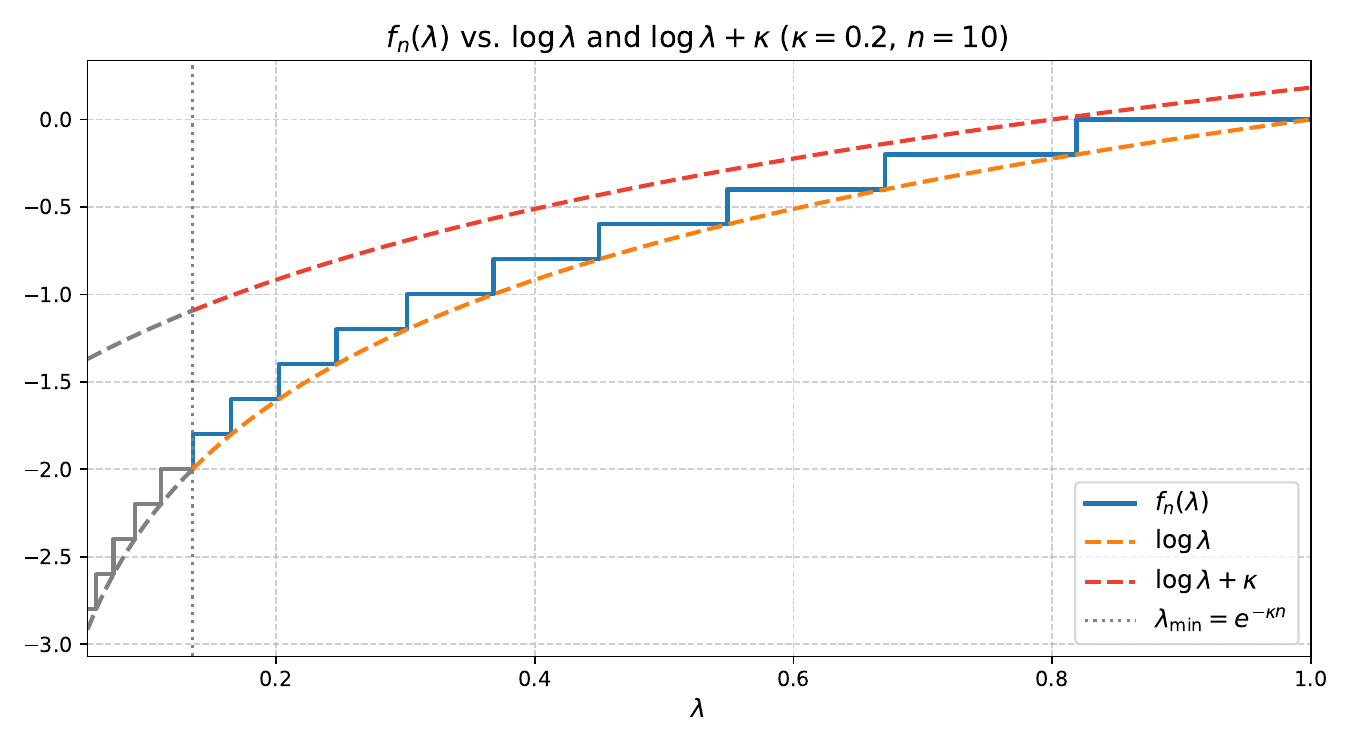}
  \mycaption{An example of logarithmic rounding.}{} 
  \label{fig:f_n}
\end{figure}

\begin{lemma}[(Binary pinching and logarithmic rounding)]
\label{lem:binary-rounding}
Let \(\sigma_n\) be a full-rank state on $\mathcal H^{\otimes n}$ satisfying $\sigma_n\ge e^{-\kappa n}\id$ for some constant \(\kappa<\infty\) independent of $n$. Let $A_n$ be a projector, and define the binary pinching map
\bb\label{eq:binary_pinching_map}
  \pazocal P(\sigma_n)
  \coloneqq
  A_n\sigma_n A_n +(\id -A_n)\sigma_n (\id-A_n).
\ee
Then, there exists a full-rank state $\sigma_n'$ such that $\pazocal P (\sigma_n')=\sigma_n'$ and $|{\rm spec}(\sigma_n')|\leq n+1$.
Moreover,
\bb
  \frac {e^{-\kappa}}2\sigma_n\leq e^{-\kappa}\pazocal P(\sigma_n)
  \le
  \sigma_n'
  \le
  e^\kappa\pazocal P(\sigma_n). \label{eq:binary-pinching-porder}
\ee
and
\bb\label{eq:bound_II}
  -\frac1n \log\Tr[A_n\sigma_n]
  \le
  -\frac1n \log\Tr[A_n\sigma_n']
  +\frac \kappa n.
\ee
\end{lemma}

\begin{proof}
Since $\sigma_n\ge e^{-\kappa n}\id$, we also have $\pazocal P(\sigma_n)\ge e^{-\kappa n}\id$.
Let
\bb
  \pazocal P(\sigma_n) = \sum_i\lambda_i P_i
\ee
be the spectral decomposition of \(\pazocal P(\sigma_n)\). For $\lambda\in[e^{-\kappa n},1]$, define
\bb\label{eq:r1}
  f_n(\lambda)
  \coloneqq
  \kappa\left\lceil
    \frac{\log\lambda +\kappa n}{\kappa}
  \right\rceil -\kappa n.
\ee
Then
\bb
  \log\lambda\le f_n(\lambda)\le\log\lambda +\kappa,\label{eq:log-runding}
\ee
and $f_n$ takes at most $n+1$ distinct values (see Figure~\ref{fig:f_n}). Define
\bb
  \widetilde\sigma_n \coloneqq \sum_i e^{f_n(\lambda_i)}P_i,
  \qquad
  \sigma_n' \coloneqq \frac{\widetilde\sigma_n}{\Tr\widetilde\sigma_n}.
\ee
By Eq.~\eqref{eq:log-runding}
\(
  \pazocal P(\sigma_n)
  \le
  \widetilde\sigma_n
  \le
  e^\kappa\pazocal P(\sigma_n)\), whence \(1\le\Tr\widetilde\sigma_n\le e^\kappa
\). Then
we get~\eqref{eq:binary-pinching-porder}.
Furthermore, by construction, 
\begin{itemize}
    \item since \(\pazocal P(\sigma_n)\) commutes with $A_n$, so does
$\sigma_n'$; hence $\pazocal P(\sigma_n')=\sigma_n'$;
    \item $\sigma_n'$ has at most $n+1$ distinct eigenvalues by construction;
    \item the binary pinching inequality\cite{Hayashi2002} gives
\bb
  \sigma_n\le2\pazocal P(\sigma_n)\le2e^\kappa\sigma_n'.
\ee
\end{itemize}
Given that $\sigma_n'\le e^\kappa\pazocal P(\sigma_n)$, we have $\Tr A_n\sigma_n = \Tr A_n\pazocal P(\sigma_n)  \ge e^{-\kappa}\Tr A_n\sigma_n' $, we finally get
\bb
  -\frac 1n \log\Tr A_n\sigma_n
  &\le
  -\frac 1n \log\Tr A_n\sigma_n' + \frac{\kappa}{n},
\ee
which concludes the proof.
\end{proof}

The following lemma is a swift application of~\cite[Lemma 2]{datta2026entropyconcentrationuniversaltypicality} to our case.

\begin{lemma}[(Noncommutative law of large numbers for padded block-product observable)]\label{lemma:new_block}
    Fix a block length $m$ and a full-rank state $\sigma_m$ on $\mathcal H^{\otimes m}$.
For $n=m\ell + r$, $0\le r<m$, define
\bb\label{eq:above1}
  \sigma_n^{(m)}\coloneqq\sigma_m^{\otimes\ell}\otimes\sigma_{\rm full}^{\otimes r},
  \qquad 
  h_m\coloneqq-\frac 1m \Tr\big[ \rho^{\otimes m}\log\sigma_m\big].
\ee
If $(\rho_n)_n$ is a $W_1$ almost i.i.d.\ source
along $\rho$, then
\bb
    \limsup_{n\to\infty}\Tr\Big[\rho_n \Big\{-\tfrac 1n \log \sigma_n^{(m)}>h_m+ \delta\Big\}\Big]=0\qquad \forall \delta>0.
\ee
\end{lemma}
\begin{proof}
    Let us decompose $\mathcal{H}^{\otimes n}$ in $\ell + 1$ blocks as in Figure~\ref{fig:padding}:
    \bb\label{eq:block_structure}
        \mathcal{H}^{\otimes n}=\underbrace{\mathcal{H}^{\otimes m}}_{B_1}\otimes \underbrace{\mathcal{H}^{\otimes m}}_{B_2}\otimes \cdots \otimes \underbrace{\mathcal{H}^{\otimes m}}_{B_\ell}\otimes\underbrace{\mathcal{H}^{\otimes r}}_{B_{\rm padd.}}.
    \ee
    \modifica{The idea is that, up to an asymptotically negligible correction due to padding, we want to apply the result of~\cite[Lemma 2]{datta2026entropyconcentrationuniversaltypicality} to the source $(\rho_n)_{n\geq 1}$ using the block structure of~\eqref{eq:block_structure}. Note that $(\rho_{m\ell}^{(r)})_{\ell\geq 1}$, defined as
    \bb
        \rho_{m\ell}^{(r)}\coloneqq \Tr_{m\ell+1,\dots, m\ell+r}\rho_n\in\mathcal{D}(\mathcal{H}^{\otimes m\ell})\simeq\mathcal{D}\big((\mathcal{H}^{\otimes m})^{\otimes \ell}\big),\quad \text{with} \quad n=m\ell + r,
    \ee
    is  a Wasserstein almost i.i.d.\ source along $\rho^{\otimes m}$ by~\cite[Lemma~31]{a_tale}. In particular, $(\rho_{m\ell}^{(r)})_{\ell\geq 1}$ is also a weakly almost i.i.d.\ source along $\rho^{\otimes m}$ by~\cite[Corollary~16]{almost_iid}.}
    Then, we have
    \bb
        \log \sigma_n^{(m)} = \sum_{i=1}^\ell \big(\log \sigma_m\big)_{B_i}\otimes\id_{B_i^c}+\big(\log\sigma_{\rm full}^{\otimes r}\big)_{B_{\rm padd.}}\otimes\id_{B_{\rm padd.}^c},
    \ee
    where $B^c$ denotes the complementary set of $B$. Let 
    \bb
    \lambda\coloneqq -(m-1)\log\lambda_{\min}(\sigma_{\rm full})\geq \max_{1\leq r\leq m-1}\|\log\sigma_{\rm full}^{\otimes r}\|_\infty, \qquad \lambda<\infty,
    \ee
    and $\mu\coloneqq\|\log\sigma_m\|_\infty<\infty$. Given any arbitrary $\delta>0$, for $n\geq n_0\coloneqq  \frac 2\delta(\mu+\lambda)$, we can upper bound
    \bb\label{eq:operator_ineq}
        -\frac 1n\log \sigma_n^{(m)} &\leq -\frac 1n\sum_{i=1}^\ell \big(\log \sigma_m\big)_{B_i}\otimes\id_{B_i^c}+\frac\lambda n\id\\
        &\leq -\frac 1{m\ell}\sum_{i=1}^\ell \big(\log \sigma_m\big)_{B_i}\otimes\id_{B_i^c}+\frac{\mu+\lambda} n\id\\
        &\leq -\frac 1{m\ell}\sum_{i=1}^\ell \big(\log \sigma_m\big)_{B_i}\otimes\id_{B_i^c}+\frac\delta 2\id,
    \ee
    whence, for all $h\in\mathbb{R}$, $\delta>0$ and $n\geq n_0$,
    \bb\label{eq:4.100}
         \Big\{-\tfrac 1n \log \sigma_n^{(m)}>h+ \delta\Big\}&\leq \Big\{-\tfrac{1}{m\ell} \textstyle\sum_{i=1}^\ell \big(\log \sigma_m\big)_{B_i}\otimes\id_{B_i^c}>h+ \frac\delta 2\Big\}\\
         &= \Big\{-\tfrac{1}{m\ell} \textstyle\sum_{i=1}^\ell \big(\log \sigma_m\big)_{B_i}\otimes\id_{B_i^c}>h+ \frac\delta 2\Big\}_{B_{\rm padd.}^c}\otimes \id_{B_{\rm padd.}}.
    \ee
    Note that the operator inequality in~\eqref{eq:operator_ineq} can be legitimately lifted to projectors as in~\eqref{eq:4.100} because the operators involved commute.
    Choosing $h=h_m$ as in~\eqref{eq:above1}, note that
    \bb 
        \lim_{\ell\to\infty}\Tr\Big[\rho_{m \ell}^{(r)} \Big\{-\tfrac{1}{m\ell} \textstyle\sum_{i=1}^\ell \big(\log \sigma_m\big)_{B_i}\otimes\id_{B_i^c}>h_m+ \frac\delta 2\Big\}\Big]=
        0\quad \text{for all} \quad 0\leq r <m, \delta>0,
    \ee
    due to~\cite[Lemma 2]{datta2026entropyconcentrationuniversaltypicality} applied to the weakly almost i.i.d.\ source $(\rho_{m\ell}^{(r)})_{\ell \geq 1}$ along $\rho^{\otimes m}\in\mathcal{D}(\mathcal{H}^{\otimes m})$.
    Then, by~\eqref{eq:4.100}, this implies
    \bb
        &\limsup_{n\to\infty}\Tr\Big[\rho_n \Big\{-\tfrac 1n \log \sigma_n^{(m)}>h_m+ \delta\Big\}\Big]=
        0\qquad \forall \delta>0,
    \ee
    which completes the proof.
\end{proof}

\begin{lemma}
\label{lem:slow-scale-exp-moment}
Let $\rho_n$ be states and let $X_n$ be self-adjoint operators satisfying $\displaystyle{\sup_n\|X_n\|_\infty<\infty}$.
Suppose that, for every \(\varepsilon>0\),
\bb\label{eq:limit}
  \lim_{n\to\infty}\Tr[\rho_n \{X_n>h+\varepsilon\}] = 0.
\ee
Then there exists a sequence \(0<s_n\leq 1\) such that
\bb\label{eq:three_claims}
  \lim_{n\to\infty}ns_n=\infty,\quad \lim_{n\to\infty}s_n=0 \quad\text{and}\quad \Tr\rho_n e^{ns_nX_n}
  \le
  \exp\big(ns_n\big(h +o(1)\big)\big).
\ee
\end{lemma}
In particular, setting $X_n=-\frac 1n \log\sigma_n^{(m)}$, the hypothesis~\ref{eq:limit} is satisfied by Lemma~\ref{lemma:new_block}; hence, Lemma~\eqref{lem:slow-scale-exp-moment} immediately gives
\begin{cor}[(Moment estimate)]
\label{cor:padded-moment}
Fix a block length $m$ and a full-rank state $\sigma_m$ on $\mathcal H^{\otimes m}$.
For $n=m\ell + r$, $0\le r<m$, define
\bb
  \sigma_n^{(m)}\coloneqq\sigma_m^{\otimes\ell}\otimes\sigma_{\rm full}^{\otimes r},
  \qquad 
  h_m\coloneqq-\frac 1m \Tr\big[ \rho^{\otimes m}\log\sigma_m\big].
\ee
If $(\rho_n)_n$ is a $W_1$ almost i.i.d.\ source along $\rho$, then there exists a sequence $s_n>0$ such that
\bb\label{eq:three_claims2}
  \lim_{n\to\infty}s_n=0, \quad \lim_{n\to\infty}ns_n=\infty\quad \text{and}\quad\Tr \big[\rho_n\big(\sigma_n^{(m)}\big)^{-s_n}\big]
  \le
  \exp\big(ns_n \big(h_m+o(1)\big)\big).
\ee
\end{cor}
\begin{proof}[Proof of Lemma~\ref{lem:slow-scale-exp-moment}]
Let \(\displaystyle{M\coloneqq \sup_n\|X_n\|_\infty}\). By~\eqref{eq:limit}, we can choose a sequence $\varepsilon_n\downarrow 0$ such that
\bb
\delta_n \coloneqq \Tr\rho_n\{X_n>h+\varepsilon_n\}
\ee
asymptotically vanishes, i.e.\ $\displaystyle{\lim_{n\to \infty} \delta_n = 0}$.
Using the convention $\log 0 = -\infty$, we choose
\bb
  s_n\coloneqq\min\left\{\frac{1}{\sqrt n},\frac 1n\sqrt{-\log \delta_n}\right\}\leq 1,
\ee
which clearly satisfies the first two claims in~\eqref{eq:three_claims}.
Now, we can upper bound
\bb
  \Tr\rho_n e^{ns_nX_n}
  &=
  \Tr\rho_n \{X_n\le h +\varepsilon_n\}e^{ns_nX_n} +\Tr\rho_n \{X_n> h +\varepsilon_n\} e^{ns_nX_n} \\
  &\le
  e^{ns_n(h+\varepsilon_n)} +\delta_n e^{ns_nM}\\
  &\leq
  e^{ns_nh}
  \left(
    e^{ns_n\varepsilon_n} + \delta_n e^{(M-h)\sqrt{-\log\delta_n}}
  \right).
\ee
Notice that the second term vanishes in $n$, as
\bb
    \lim_{x\to 0^+}xe^{c\sqrt{-\log x}}=0 \qquad \forall c\in \mathbb{R}.
\ee
Therefore,
\bb
  \log\Tr \rho_n e^{ns_nX_n}
  &\le
  ns_nh +
  \log\left(
    e^{ns_n\varepsilon_n} + o(1)
  \right) 
  =ns_n\big(h + o(1)\big),
\ee
since $\displaystyle{\lim_{n\to\infty} \epsilon_n=0}$.
\end{proof}

\begin{figure}[t]
  \centering
  \def\svgwidth{\linewidth}
  \small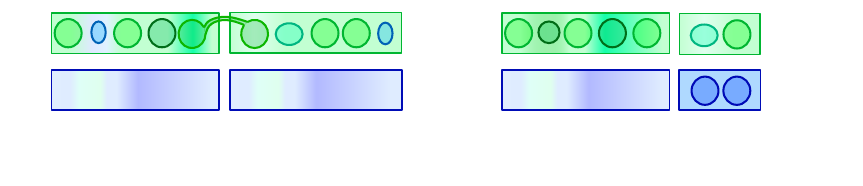
  \mycaption{Blocking and padding the systems}{A pictorial representation of $\ell =3$, $m=5$, $r=2$.} 
  \label{fig:padding}
\end{figure}

\begin{lemma}
\label{lem:projection-to-DH}
Let \(\rho_n,\sigma_n\) be states. Suppose that $R\geq 0$ is a number with the following property: for every sequence of projectors $A_n$ satisfying
\bb\label{eq:i}
\liminf_{n\to\infty}\Tr A_n\rho_n >0,
\ee
we are guaranteed to have
\bb\label{eq:ii}
 \limsup_{n\to\infty}
  -\frac1n\log\Tr A_n\sigma_n
  \le R
\ee
Then, for every \(\eta\in (0,1)\),
\bb
  \limsup_{n\to\infty}
  \frac1n D_H^\eta(\rho_n\|\sigma_n)
  \le R.
\ee
\end{lemma}

\begin{proof}
Recall that
\bb
  \beta_\eta(\rho_n\|\sigma_n)
  \coloneqq
  \inf\left\{
    \Tr T_n\sigma_n :
    0\le T_n\le \id,\ 
    \Tr T_n\rho_n \ge 1-\eta
  \right\}.
\ee
Since the feasible set is compact in finite dimension, the infimum is attained. Fix an optimal test \(T_n\), and a constant $c\in (0,1-\eta)$ so that
\bb
  \Tr T_n\rho_n \ge1-\eta > c,
  \qquad
  \Tr T_n\sigma_n
  =
  \beta_\eta(\rho_n\|\sigma_n).
\ee
Define the threshold projector 
$A_n\coloneqq \left\{T_n\ge c\right\}$ which satisfies $T_n \le A_n + c(\id -A_n)$.
Then we have
\bb
  1 - \eta
  &\le
  \Tr T_n\rho_n  \le
  \Tr A_n\rho_n + c\Tr (\id-A_n)\rho_n =
  c+(1-c)\Tr A_n\rho_n.
\ee
Hence
\bb
  \Tr A_n\rho_n
  \ge
  \frac{1 - \eta - c}{1-c}
  >0.
\ee
On the other hand,
  $T_n\ge c A_n$
  gives
  $\Tr T_n\sigma_n
  \ge
  c\Tr A_n\sigma_n$.
By assumption (ii) and since $c>0$ is independent of $n$, we have 
\bb
  \limsup_{n\to\infty}
  -\frac 1n \log\beta_\eta(\rho_n\|\sigma_n)
  \le R.
\ee
This proves the claim.
\end{proof}

We now combine the previous lemmas to prove the padded block-product strong
converse.

\subsubsection{Proof of Proposition~\ref{prop:padded-block-product-strong-converse}}\label{proof:padded-block}
Let
\bb
  h_m \coloneqq -\frac1m \Tr \rho^{\otimes m}\log\sigma_m ,
  \qquad
  h_m - S(\rho) =\frac1m D(\rho^{\otimes m}\|\sigma_m).
  \label{eq:m-exponent}
\ee
Let $A_n$ be an arbitrary sequence of projectors such that
\bb\label{eq:assumption}
  \liminf_{n\to\infty}\Tr  A_n\rho_n > 0.
\ee
Given two full-rank states $\sigma_m\in\mathcal{S}^{(m)}$ and $\sigma_{\rm full}\in \mathcal S^{(1)}$, we define $\sigma_n^{(m)}\coloneqq\sigma_m^{\otimes\ell}\otimes\sigma_{\rm full}^{\otimes r}$, so that there exists \(0<\kappa<\infty\) such that $\sigma_n^{(m)}\ge e^{-\kappa n}\id$.
Applying Lemma~\ref{lem:binary-rounding} to the full-rank state $\sigma_n=\sigma_n^{(m)}$ there exists a full-rank state $\sigma_n'$ such that
\bb\label{eq:sopra1}
    -\frac 1n \log \Tr A_n\sigma_n^{(m)}
    &\le
    -\frac 1n \log\Tr A_n\sigma_n' + \frac{\kappa}{n}
\ee
and $\pazocal P(\sigma_n')=\sigma_n'$, where $\pazocal P$ is the binary pinching map defined in~\eqref{eq:binary_pinching_map}. In particular, the spectral PVM $\{\Pi_i\}$ of $\sigma_n'$ commutes with $A_n$, and the set of projectors 
\bb
\mathcal G_n\coloneqq \bigcup_i \{\Pi_iA_n,\Pi_i(\id-A_n)\}
\ee
is a PVM -- called \emph{common refinement} of \(\{\Pi_i\}\) and \(\{A_n,\id -A_n\}\) -- which defines a \emph{refined pinching map}
\bb
  \pazocal P_{\mathcal G_n}(\rho_n)\coloneqq \sum_{E\in \mathcal{G}_n}E\rho_n E.
\ee
Clearly,
\bb
  \Tr \big[\pazocal P_{\mathcal G_n}(\rho_n) A_n\big]
  =
  \Tr\big[\rho_n \pazocal P_{\mathcal G_n}(A_n)\big]=\Tr\big[\rho_n  A_n\big],
\ee
whence, by~\eqref{eq:assumption},
\bb\label{eq:assumption2}
    \liminf_{n\to\infty}\Tr\big[ A_n \pazocal P_{\mathcal G_n}(\rho_n)\big] > 0.
\ee
Furthermore, since $\sigma_n'$ has at most $n+1$ distinct eigenvalues (cf.\ Lemma~\ref{lem:binary-rounding}), \(\mathcal G_n\) has at most \(2(n+1)\) projectors, hence $\displaystyle{\lim_{n\to\infty}\tfrac{1}{n}\log |\mathcal{G}_n}|=0$. Then, by Proposition~\ref{prop:pinching-extension}, we have
\bb\label{eq:prop13}
    \lim_{n\to\infty}\Tr\Big[\pazocal P_{\mathcal G_n}(\rho_n)\Big\{
        \left\lvert-\tfrac1n\log\pazocal P_{\mathcal G_n}(\rho_n)-S(\rho)\right\rvert>\delta
    \Big\}\Big]
    =0
\ee
for every \(\delta>0\).
By~\eqref{eq:binary-pinching-porder}, $\sigma_n^{(m)}$ and $\sigma_n'$ also satisfy the operator inequality
\bb\label{eq:op_ineq}
    \sigma_n^{(m)}\leq 2e^{\kappa}\sigma_n'.
\ee
Furthermore, by the very definition of $\pazocal P_{\mathcal G_n}$, we have
\bb\label{eq:repinching}
    \pazocal{P}_{\mathcal{G}_n}(\sigma_n')= \sigma_n'.
\ee
In particular, $\pazocal{P}_{\mathcal{G}_n}(\rho_n)$ and $\sigma'_n$ commute, so that we can classically apply Markov's inequality in order to upper bound
\bb\label{eq:4.130}
    \Tr\Big[\pazocal P_{\mathcal{G}_n}(\rho_n) \Big\{-\tfrac 1n \log \sigma_n'>h+ \delta\Big\}\Big]&=\PP{P_{\mathcal{G}_n}(\rho_n)}\Big(-\tfrac 1n \log \sigma_n'>h+ \delta\Big)\\
    &=\PP{\pazocal P_{\mathcal{G}_n}(\rho_n)}\big( \sigma_n'< e^{-n(h+\delta)}\big)\\
    &=\PP{\pazocal P_{\mathcal{G}_n}(\rho_n)}\big( \sigma_n'^{- s}> e^{ns(h+\delta)}\big)\\
    &< e^{-ns(h+\delta)}\EE{\pazocal P_{\mathcal{G}_n}(\rho_n)}\big[ \sigma_n'^{-s}\big]\\
    &=e^{-ns(h+\delta)}\Tr\big[\pazocal P_{\mathcal{G}_n}(\rho_n) \sigma_n'^{-s}\big]\\
    &=e^{-ns(h+\delta)}\Tr\big[\rho_n \sigma_n'^{-s}\big],
\ee
for $h\in \mathbb{R}$, $s\in(0,1]$ and $\delta>0$ arbitrary, where the last inequality follows from~\eqref{eq:repinching} combined with $\pazocal P_{\mathcal{G}_n}=\pazocal P_{\mathcal{G}_n}^\dagger$.
Now we want to apply Corollary~\ref{cor:padded-moment}. Let $0<s_n\leq 1$ be the sequence as in~\eqref{eq:three_claims2}. Then, by operator (anti)monotonicity of $x\mapsto x^{-s}$ (with $0<s\leq 1$), Eq.~\eqref{eq:op_ineq} yields
\bb
     e^{-s_n(\kappa+\log 2)}(\sigma_n')^{-s_n}\leq \big(\sigma_n^{(m)}\big)^{-s_n},
\ee
whence, by~\eqref{eq:three_claims2}, 
\bb
    \Tr\big[\rho_n\sigma_n'^{\,-s_n}\big]\leq e^{s_n(\kappa+\log 2)}\Tr\big[\rho_n\big(\sigma_n^{(m)}\big)^{-s_n}\big]\leq e^{ns_n (h_m+o(1))}.
\ee
Thus, setting $h=h_m$ and $s=s_n$ in~\eqref{eq:4.130}, by the previous bound we get
\bb\label{eq:paddle-upper-tails}
    &\limsup_{n\to\infty}\Tr\Big[\pazocal P_{\mathcal{G}_n}(\rho_n) \Big\{-\tfrac 1n \log \sigma_n'>h_m+ \delta\Big\}\Big]\leq \lim_{n\to\infty}e^{-ns_n(\delta+o(1))}=0\qquad \forall \delta>0.
\ee
Since \(\pazocal P_{\mathcal{G}_n}\rho_n\) and $\sigma_n'$ commute, we can combine~\eqref{eq:paddle-upper-tails} with~\eqref{eq:prop13} in order to obtain
\bb\label{eq:log_likelihood}
  &\limsup_{n\to\infty} \Tr\Big[ \pazocal P_{\mathcal{G}_n}(\rho_n)
    \Big\{
      \tfrac{1}{n}
      \big(\log\pazocal P_{\mathcal{G}_n}(\rho_n) - \log\sigma_n'\big)
      >
      h_m - S(\rho) + \delta
    \Big\}\Big]\\
    &\qquad\leqt{(a)} \limsup_{n\to\infty}\Tr\Big[\pazocal P_{\mathcal{G}_n}(\rho_n) \Big\{-\tfrac 1n \log \sigma_n'>h_m+ \tfrac\delta 2\Big\}\Big]\\
    &\qquad \qquad +\lim_{n\to\infty}\Tr\Big[\pazocal P_{\mathcal G_n}(\rho_n)\Big\{
        \left\lvert-\tfrac1n\log\pazocal P_{\mathcal G_n}(\rho_n)-S(\rho)\right\rvert>\tfrac \delta 2
    \Big\}\Big]=0,
\ee
where $h_m - S(\rho)=\frac1m D(\rho^{\otimes m}\|\sigma_m)$. In particular, in (a) we have leveraged the fact that $\pazocal P_{\mathcal G_n}(\rho_n)$ and $\sigma_n'$ commute to use the classical inequality
\bb
    \mathbb{P}\big(A+B> a+b+\delta\big)&\leq \mathbb{P}\Big(A> a+\tfrac{\delta}{2} \lor B> b+\tfrac{\delta}{2}\Big)\\
    &\leq \mathbb{P}\Big(A> a+\tfrac{\delta}{2}\Big)+\mathbb{P}\Big(B-b>\tfrac{\delta}{2}\Big)\\
    &\leq \mathbb{P}\Big(A> a+\tfrac{\delta}{2}\Big)+\mathbb{P}\Big(|B-b|> \tfrac{\delta}{2}\Big) .
\ee
Let us briefly recall a fundamental inequality in classical hypothesis testing given by~\cite[Lemma~4.1.2]{HAN}, which is essentially equivalent to the Neyman-Pearson lemma~\cite{Neyman_Pearson} in the asymptotic limit $n\to\infty$. Let $X^n$ and $Y^n$ be two random variables taking values on a set $\mathcal{X}^n$, where $\mathcal{X}$ is finite, and having laws $p_n$ and $q_n$, respectively. Then, for every $S_n\subseteq \mathcal{X}^n$, the following inequality holds:
\bb\label{eq:Han}
    \PP{X^n\sim p_n}\big(X^n\notin S_n\big)+e^{nt}\PP{Y^n\sim q_n}\big(Y^n\in S_n\big)\geq \PP{X^n\sim p_n}\Big(\tfrac 1n \log \tfrac{p_n(X^n)}{q_n(X^n)}\leq t\Big)\qquad \forall t\in \mathbb R.
\ee
Let us compare~\eqref{eq:Han} with our setting, where we identify $X^n= \pazocal P_{\mathcal{G}_n}(\rho_n)$ and  $Y^n= \sigma_n'$, both intended as random variables with respect to the spectral measure of $\pazocal P_{\mathcal{G}_n}(\rho_n)$. We set $t=\frac 1m D(\rho^{\otimes m}\|\sigma_m) + \delta$, where $\delta>0$ is arbitrary.
\begin{itemize}
    \item On one hand,~\eqref{eq:log_likelihood} can be written as a bound on the log likelihood-ratio:
\bb
    \lim_{n\to\infty} \PP{\pazocal P_{\mathcal{G}_n}(\rho_n)}
    \Big[
      \tfrac{1}{n}
      \log\tfrac{\pazocal P_{\mathcal{G}_n}(\rho_n)}{\sigma_n'}
      >
     \tfrac 1m D(\rho^{\otimes m}\|\sigma_m) + \delta
    \Big]=0 \qquad\forall \delta >0,
\ee
namely 
\bb
\lim_{n\to\infty}\PP{X^n\sim p_n}\Big(\tfrac 1n \log \tfrac{p_n(X^n)}{q_n(X^n)}\leq t\Big)=1.
\ee
\item On the other hand, under~\eqref{eq:assumption} we have showed that~\eqref{eq:assumption2} holds, which can be rewritten as
\bb
    0< \liminf_{n\to\infty}\Tr\big[ A_n \pazocal P_{\mathcal G_n}(\rho_n)\big]=\liminf_{n\to\infty}\PP{X^n\sim p_n}\big(X^n\in S_n\big)
\ee
if we interpret the sequence of projectors $(A_n)_n$ as a sequence of classical events $(S_n)_n$.
\end{itemize}
 Therefore,  by~\eqref{eq:Han}, we have
\bb\label{eq:Han2}
    \liminf_{n\to\infty}e^{nt}\PP{Y^n\sim q_n}\big(Y^n\in S_n\big)&\geq \liminf_{n\to\infty}\left(\PP{X^n\sim p_n}\Big(\tfrac 1n \log \tfrac{p_n(X^n)}{q_n(X^n)}\leq t\Big)-\PP{X^n\sim p_n}\big(X^n\notin S_n\big)\right)\\
    &=\liminf_{n\to\infty}\left(1-\PP{X^n\sim p_n}\big(X^n\notin S_n\big)\right)\\
    &=\liminf_{n\to\infty}\PP{X^n\sim p_n}\big(X^n\in S_n\big)>0.
\ee
Hence, there exist constants $c>0$ and $n_0$ such that $\displaystyle{e^{nt}\PP{Y^n\sim q_n}\big(Y^n\in S_n\big)\geq c}$ for every $n\geq n_0$.
Therefore, for every $n\geq n_0$,
\bb
    -\frac {1}n\log \PP{Y^n\sim q_n}\big(Y^n\in S_n\big)\leq t-\frac1n\log c,
\ee
and consequently
\bb
    \limsup_{n\to\infty}-\frac {1}n\log \PP{Y^n\sim q_n}\big(Y^n\in S_n\big)\leq t.
\ee
Rephrasing the previous inequality in terms of our quantities, we have proved the implication
\bb\label{eq:implication}
  \liminf_{n\to\infty}\Tr A_n\rho_n > 0
  \quad\implies\quad
  \limsup_{n\to\infty} -\frac 1n \log \Tr A_n\sigma_n'
  \le \frac 1m D(\rho^{\otimes m}\|\sigma_m),
\ee
where we have used that $\delta>0$ can be taken arbitrarily small.

Recalling that $\sigma_n'$ satisfies the inequality~\eqref{eq:sopra1}, we can carry on the implication in~\eqref{eq:implication} to
\bb\label{eq:projector}
  \liminf_{n\to\infty}\Tr A_n\rho_n > 0
  \quad\implies\quad
  \limsup_{n\to\infty} -\frac 1n \log \Tr A_n\sigma_n^{(m)}  \le \frac 1m D(\rho^{\otimes m}\|\sigma_m),
\ee
which is in a form that ensures the validity of the hypothesis of Lemma~\ref{lem:projection-to-DH}. Whence, for every \(\eta\in(0,1)\),
\bb
  \limsup_{n\to\infty}
  \frac1n D_H^\eta(\rho_n\|\sigma_n^{(m)})
  \le \frac 1m D(\rho^{\otimes m}\|\sigma_m).
\ee
This completes the proof.

\subsection{First act. Proof of Theorem~\ref{thm:W1_GQSL} for individual sources}\label{sec:individual}

We first prove the individual-source statement. Throughout this section, we fix a Wasserstein almost i.i.d.\ source $(\rho_n)_{n \ge 1}$ along a state $\rho\in\mathcal{D}(\mathcal H)$.
The universal statement for equiconvergent sources will be derived in Section~\ref{sec:universality} from the individual-source statement by a worst-case minimax reduction.
The individual-source proof has two parts:
\begin{align}
  \limsup_{n\to\infty} \frac 1n \,\rel{D_H^\epsilon}{\rho_n}{\mathcal S^{(n)}}
  \le
  D^\infty(\rho\|\mathcal S)\, , \tag{Converse, see Proposition~\ref{prop:W1-stable-strong-converse}}
  \\
  \liminf_{n\to\infty} \frac 1n\, \rel{D_H^\epsilon}{\rho_n}{\mathcal S^{(n)}}
  \ge
  D^\infty(\rho\|\mathcal S)\, .\tag{Achievability, see Proposition~\ref{prop:W1-achievability}}
\end{align}

The main ingredients for these parts are the following: the converse bound uses the $W_1$-stable padded block-product strong converse; the achievability bound is based on a combination of Lemma~\ref{lem:one-step}, which we call \emph{gap contraction lemma}, with the Wasserstein asymptotic continuity of the relative entropy of resource (Proposition~\ref{prop:3.1}), namely
\bb
  \frac1n\,\rel{D}{\rho_n}{\mathcal S^{(n)}}
  =
  \frac1n\,\rel{D}{\rho^{\otimes n}}{\mathcal S^{(n)}} + o(1)\,.
\ee

\begin{rem}[(Relation with the proof of the GQSL for an i.i.d.\ null hypothesis)]
The proof follows the formal skeleton of the generalised quantum Stein lemma proof of Hayashi--Yamasaki~\cite{Hayashi2025}, but with a different null model. In the i.i.d.\ proof, the null state is exactly $\rho^{\otimes n}$. Here, the null is the non-i.i.d.\ sequence $(\rho_n)_n$, assumed only to be normalised $W_1$-close to $\rho^{\otimes n}$. Therefore the padded block-product strong converse used in the i.i.d.\ proof cannot be imported directly; it is replaced by Proposition~\ref{prop:padded-block-product-strong-converse}. The key steps for the proof of the strong converse and the comparison with the strategy of~\cite{Hayashi2025} are represented in Figure~\ref{fig:converse}.
In the achievability (direct) part, all testing, pinching, and information-spectrum estimates are carried out by $\rho_n$. The replacement of $\rho_n$ by $\rho^{\otimes n}$ is made only after minimisation over $\mathcal S^{(n)}$, through the Wasserstein continuity of the minimised relative entropy. In Figure~\ref{fig:achievability}, we represent the key steps in the i.i.d.\ setting and in the $W_1$ relaxation.
\end{rem}

\subsubsection{Converse}
\label{sub:converse}
\begin{figure}[t]
  \centering
  \def\svgwidth{\linewidth}
  \scriptsize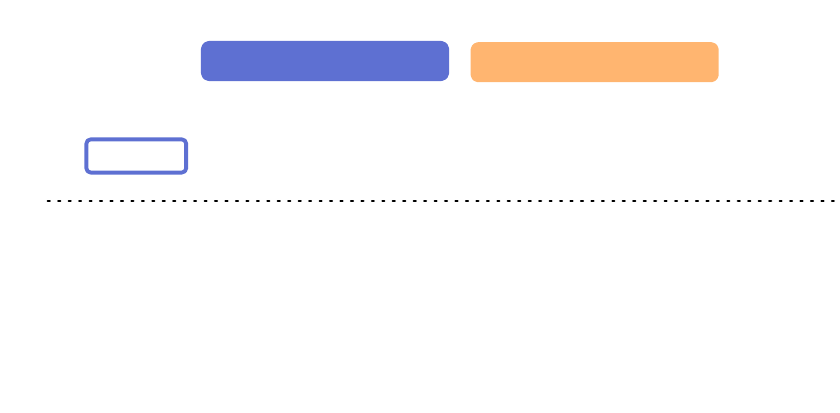
  \mycaption{Comparison of the converse proof for the i.i.d.\ GQSL and $W_1$ GQSL.}{For the i.i.d. GQSL, the key step is the block converse~\cite[Lemma S6]{Hayashi2025}, which directly provides a converse proof by extending the block length $m$ to infinity. Whereas, in the present setting, the non-i.i.d.\ null requires a block product that contains $W_1$ information-spectrum stability, pinching stability, and a block law of large number. These are the main ingredients to derive the $W_1$ block strong converse in Proposition~\ref{prop:padded-block-product-strong-converse}. Together with the full-rank approximation and extending the block length, Proposition~\ref{prop:padded-block-product-strong-converse} yields the desired $W_1$ converse proof.} 
  \label{fig:converse}
\end{figure}
The argument is the same as the strong-converse part of the generalised Stein lemma (cf.~\cite[Proposition S5]{Hayashi2025}), except that the i.i.d.\ block argument is replaced by the $W_1$-stable padded block-product strong converse (Proposition~\ref{prop:padded-block-product-strong-converse}). The idea here is to test only against a sequence of alternatives inside $\mathcal S^{(n)}$: we choose a nearly optimal block state in $\mathcal S^{(m)}$, make it full rank, and form a padded block-product alternative. Since this alternative belongs to $\mathcal S^{(n)}$, the composite alternative cannot have a larger exponent than this particular one. Finally, Proposition~\ref{prop:padded-block-product-strong-converse} then gives the desired upper bound. The crucial steps and ideas of the proof, compared with the i.i.d.\ case, are represented in Figure~\ref{fig:converse}.

Given that the padded block-product strong converse in Proposition~\ref{prop:padded-block-product-strong-converse} is stated for full-rank block alternatives, we now show that it is possible to mix any block alternative with the fixed full-rank free state $\sigma_{\rm full}^{\otimes m}$.

\begin{lemma}[(Full-rank approximation)]
\label{lem:full-rank-approx-S-m}
Let $m\ge 1$ and $\sigma_m\in\mathcal S^{(m)}$. For \(\delta\in(0,1)\), define
\bb
    \sigma_{m,\delta}
    \coloneqq
    (1-\delta)\sigma_m+
    \delta\sigma_{\rm full}^{\otimes m}.
\ee
Then $\sigma_{m,\delta}\in\mathcal S^{(m)}$ and $\sigma_{m,\delta}$ is full rank. Moreover, if $D\left(\rho^{\otimes m}\middle\|\sigma_m\right)<\infty$
then
\bb\label{eq:conv_above}
    \lim_{\delta\to 0^+}
    D\left(\rho^{\otimes m}\middle\|\sigma_{m,\delta}\right)
    =
    D\left(\rho^{\otimes m}\middle\|\sigma_m\right).
\ee

\end{lemma}

\begin{proof}
By definition, $\sigma_{m,\delta}\geq\delta\sigma_{\rm full}^{\otimes m}>0$; furthermore, by tensor-product closure and convexity of $\mathcal S^{(m)}$, we have $\sigma_{m,\delta} \in \mathcal S^{(m)}$.
Now, \(\sigma_{m,\delta}\) converges to \(\sigma_m\), and, assuming \(D(\rho^{\otimes m}\|\sigma_m)<\infty\), i.e.\
\(\supp\rho^{\otimes m}\subseteq\supp\sigma_m\), we conclude that
\bb
    \lim_{\delta\to 0^+}\Tr \rho^{\otimes m}\log\sigma_{m,\delta}
    =
    \Tr \rho^{\otimes m}\log\sigma_m.
\ee
The remaining term \(\Tr\rho^{\otimes m}\log\rho^{\otimes m}\) in~\eqref{eq:conv_above} is independent of \(\delta\), and this completes the proof.
\end{proof}

\begin{boxedstep}{}
\begin{prop}[(Strong converse part)]
\label{prop:W1-stable-strong-converse}
Let $(\rho_n)_n$ be a Wasserstein almost i.i.d.\ source along a state $\rho\in\mathcal{D}(\mathcal{H})$. Then, for every \(\epsilon\in(0,1)\), we have
\begin{align}
\limsup_{n\to\infty}
\frac1n\,\rel{D_H^\epsilon}{\rho_n\|\mathcal S^{(n)}}
  \le
  D^\infty(\rho\|\mathcal S).
\label{eq:W1-stable-strong-converse}
\end{align}
\end{prop}
\end{boxedstep}

\begin{rem}
    This result is new also for the case of a full-rank i.i.d.\ alternative hypothesis $\mathcal{S}^{(n)}=\{\sigma^{\otimes n}\}$: indeed, for individual $W_1$ sources $(\rho_n)_n$ along $\rho$, all the previous works~\cite{a_tale, datta2026entropyconcentrationuniversaltypicality} proved the inequality
\bb
    \liminf_{n\to\infty}
\frac1n\,\rel{D_H^\epsilon}{\rho_n}{\sigma^{\otimes n}}
  \geq
  D(\rho\|\sigma)\, .
\ee
Now, Proposition~\ref{prop:W1-stable-strong-converse} ensures that, if $\sigma\in\mathcal{D}(\mathcal{H})$ has full rank, then 
\bb
    \limsup_{n\to\infty}
\frac1n\,\rel{D_H^\epsilon}{\rho_n}{\sigma^{\otimes n}}
  =
  D(\rho\|\sigma).
\ee
Note that the assumption of full rank is essential: in~\cite[Remark~15]{a_tale} a counterexample is given when $\supp (\rho_n)\not\subseteq \supp (\sigma^{\otimes n})$. Another counterexample in~\cite[Remark~15]{a_tale} shows that Proposition~\ref{prop:W1-stable-strong-converse} cannot be extended to weak sources $(\rho_n)_n$, even assuming that $\sigma$ has full support.
\end{rem}

\begin{proof}[Proof of Proposition~\ref{prop:W1-stable-strong-converse}.]
Fix a block length $m\ge 1$ and let $\eta >0$. Choose, by compactness,
$\sigma_m \in \mathcal S^{(m)}$ such that
\bb
  \frac 1m D(\rho^{\otimes m}\|\sigma_m)
  =
  \frac 1m D(\rho^{\otimes m}\|\mathcal S^{(m)}).
\ee
For $\delta\in(0,1)$, define
\(
  \sigma_{m,\delta} \coloneqq(1-\delta)\sigma_m +\delta\,\sigma_{\rm full}^{\otimes m}
\), which has full rank.
Moreover, by Lemma~\ref{lem:full-rank-approx-S-m}, choosing $\delta>0$ sufficiently small such that $D(\rho^{\otimes m}\|\sigma_{m,\delta})\le D(\rho^{\otimes m}\|\sigma_{m})+\eta$, then
\bb
  \frac 1m D(\rho^{\otimes m}\|\sigma_{m,\delta})
  \le
  \frac 1m D(\rho^{\otimes m}\|\mathcal S^{(m)}) +\frac{\eta}{m} .
  \label{eq:non-full-rank-approx}
\ee
Now, let us take the padded block-product alternative
\(
  \sigma_n^{(m)}
  \coloneqq
  \sigma_{m,\delta}^{\otimes \ell}\otimes\sigma_{\rm full}^{\otimes r},
\)
which is full-rank and belongs to $\mathcal S^{(n)}$. Then,
\bb
  \rel{D_H^\epsilon}{\rho_n}{\mathcal S^{(n)}}
  \le
  \rel{D_H^\epsilon}{\rho_n}{\sigma_n^{(m)}}\, .
\ee
Applying Proposition~\ref{prop:padded-block-product-strong-converse}, followed by Eq.~\eqref{eq:non-full-rank-approx} we obtain
\bb
  \limsup_{n\to\infty}
  \frac 1n\, \rel{D_H^\epsilon}{\rho_n}{\mathcal S^{(n)}}
  \le
  \frac 1m\, D(\rho^{\otimes m} \|\sigma_{m,\delta})
  \le
  \frac 1m\,\rel{D}{\rho^{\otimes m}}{\mathcal S^{(m)}} +\frac \eta m\, .
\ee
Letting first $\eta\to0^+$ and then $m\to\infty$, and using the definition
of $D^\infty(\rho\|\mathcal S)$, we have the converse bound~\eqref{eq:W1-stable-strong-converse}.
\end{proof}

\subsubsection{Achievability}
\label{sub:achievabiltiy}
We now turn to the achievability part. The central ingredient -- which we call \emph{gap contraction lemma} (Lemma~\ref{lem:one-step}) -- is a one-step version of the key 'update' lemma~\cite[Lemma S8]{Hayashi2025} of the proof in the i.i.d.\ case. In contrast to the original iterative construction~\cite[Proposition S7]{Hayashi2025}, we use Lemma~\ref{lem:one-step} only once, in a proof by contradiction.

The skeleton of gap contraction lemma is similar to~\cite[Lemma~S8]{Hayashi2025}. The construction of the three-state mixture, logarithmic rounding, pinching, threshold projectors, and the three-region decomposition is independent of the type of null.
In the i.i.d.\ setting, the testing bound for the padded block-product anchor, cf.~\cite[Eq.\ S81]{Hayashi2025}, is bounded by~\cite[Lemma S6]{Hayashi2025}; here, it follows from Proposition~\ref{prop:padded-block-product-strong-converse}.
The main steps and ideas of the proof are represented in Figure~\ref{fig:achievability}, together with a comparison with the i.i.d.\ case.

\begin{figure}[t]
  \centering
  \def\svgwidth{\linewidth}
  \scriptsize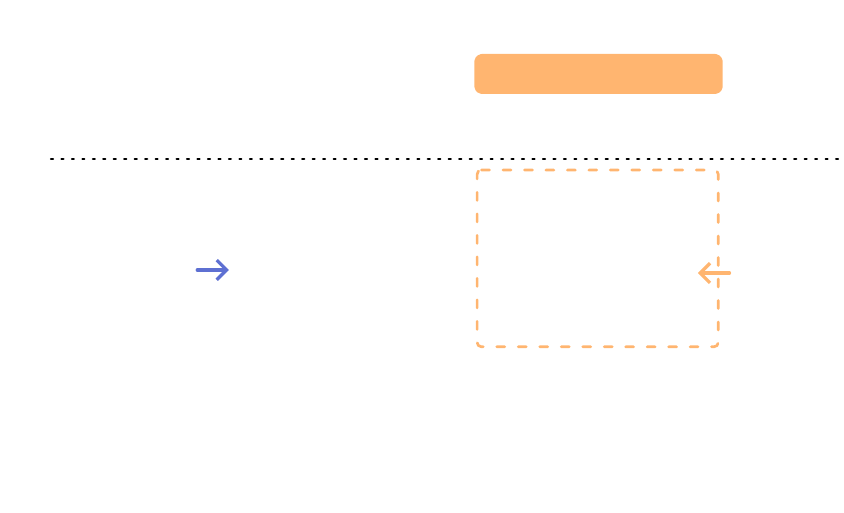
  \mycaption{Comparison of the achievability proof for the i.i.d.\ GQSL and $W_1$ GQSL.}{The proof of $W_1$ GQSL follows the central ideas of the i.i.d.\ one in~\cite{Hayashi2025}: the anchor state is formed by three distinct components, followed by logarithmic rounding, pinching, two testing bounds, and a three-region decomposition. The main new statement is the $W_1$ block strong converse in Proposition~\ref{prop:padded-block-product-strong-converse} which replaces the original block strong converse in~\cite[Lemma S6]{Hayashi2025}. The iterative gap-closing argument in the i.i.d.\ proof is replaced by a single gap-contraction lemma together with $W_1$ continuity, yielding a contradiction.
  } 
  \label{fig:achievability}
\end{figure}

\begin{lemma}[(Log-likelihood projector)] 
\label{lem:threshold-commuting}
Let $p_n$ and $q_n$ be two classical probability distributions on $\mathcal{X}^n$, where $\mathcal{X}$ is a finite set, and let $\epsilon\in(0,1)$. The following implications hold.
\begin{itemize}
\item Suppose that
\bb\label{eq:a}
\liminf_{n\to\infty}\frac1nD_H^\epsilon(p_n\|q_n)\le a.
\ee
Then, for every $\delta>0$,
  \bb\label{eq:typeI}
    \liminf_{n\to\infty}\Tr [p_n \{p_n\ge e^{n(a+\delta)}q_n\}]
    \le 1-\epsilon.
    \ee
    \item \modifica{Suppose that
    \bb\label{eq:a1}
\limsup_{n\to\infty}\frac1nD_H^{1-\epsilon}(p_n\|q_n)\le b.
\ee
Then, for every $\delta>0$, 
\bb\label{eq:typeI1}
    \limsup_{n\to\infty}\Tr [p_n \{p_n\ge e^{n(b+\delta)}q_n\}]
    \le \epsilon.
    \ee}
    \end{itemize}
\end{lemma}

\begin{proof}
This is the classical formulation of~\cite[Lemma S10]{Hayashi2025}, which was there proved in the quantum setting.
\end{proof}

\begin{figure}[t]
  \centering
  \def\svgwidth{0.55\linewidth}
  \small
\begingroup%
  \makeatletter%
  \providecommand\color[2][]{%
    \errmessage{(Inkscape) Color is used for the text in Inkscape, but the package 'color.sty' is not loaded}%
    \renewcommand\color[2][]{}%
  }%
  \providecommand\transparent[1]{%
    \errmessage{(Inkscape) Transparency is used (non-zero) for the text in Inkscape, but the package 'transparent.sty' is not loaded}%
    \renewcommand\transparent[1]{}%
  }%
  \providecommand\rotatebox[2]{#2}%
  \newcommand*\fsize{\dimexpr\f@size pt\relax}%
  \newcommand*\lineheight[1]{\fontsize{\fsize}{#1\fsize}\selectfont}%
  \ifx\svgwidth\undefined%
    \setlength{\unitlength}{183.82800005bp}%
    \ifx\svgscale\undefined%
      \relax%
    \else%
      \setlength{\unitlength}{\unitlength * \real{\svgscale}}%
    \fi%
  \else%
    \setlength{\unitlength}{\svgwidth}%
  \fi%
  \global\let\svgwidth\undefined%
  \global\let\svgscale\undefined%
  \makeatother%
  \begin{picture}(1,0.35161603)%
    \lineheight{1}%
    \setlength\tabcolsep{0pt}%
    \put(0,0){\includegraphics[width=\unitlength,page=1]{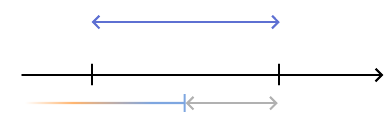}}%
    \put(0.55997563,0.32610062){\color[rgb]{0,0,0}\makebox(0,0)[rt]{\lineheight{2.81751251}\smash{\begin{tabular}[t]{r}$R'-R_\epsilon$\end{tabular}}}}%
    \put(0.25805214,0.21624213){\color[rgb]{0,0,0}\makebox(0,0)[rt]{\lineheight{2.81751251}\smash{\begin{tabular}[t]{r}$R_\epsilon$\end{tabular}}}}%
    \put(0.74989272,0.21688458){\color[rgb]{0,0,0}\makebox(0,0)[rt]{\lineheight{2.81751251}\smash{\begin{tabular}[t]{r}$R'$\end{tabular}}}}%
    \put(0.71239593,0.00686879){\color[rgb]{0,0,0}\makebox(0,0)[rt]{\lineheight{2.81751251}\smash{\begin{tabular}[t]{r}$\eta(R'-R_\epsilon)$\end{tabular}}}}%
    \put(0.33374783,0.00766166){\color[rgb]{0,0,0}\makebox(0,0)[rt]{\lineheight{2.81751251}\smash{\begin{tabular}[t]{r}$D^\infty(\rho\|\mathcal S)$\end{tabular}}}}%
  \end{picture}%
\endgroup%

  \mycaption{The gap contraction lemma.}{} 
  \label{fig:one_step}
\end{figure}

\begin{lemma}[(Gap contraction lemma)]
\label{lem:one-step}
Let $(\sigma_n)_n$ be any sequence of states with $\sigma_n\in\mathcal S^{(n)}$. Let us define
\begin{align}
  & R^{\epsilon} \coloneqq \liminf_{n\to\infty}\frac 1n\,\rel{D_H^\epsilon}{\rho_n}{\mathcal S^{(n)}}\, , \tag{Stein's exponent}
  \\
  &R'\coloneqq
  \liminf_{n\to\infty}
  \frac1n D(\rho^{\otimes n}\|\sigma_n), \tag{Auxiliary exponent}
\end{align}
and assume that
  $R^{\epsilon} < R' <\infty$.
Then, for every $\eta\in(0,\epsilon)$, there exists a sequence
$\sigma_n'\in\mathcal S^{(n)}$ such that
\bb
  \liminf_{n\to\infty}
  \frac1n D(\rho_n\|\sigma_n') - R^{\epsilon}
  \le
  (1-\eta)\left(R'-R^{\epsilon}\right).
\ee
\end{lemma}

\begin{rem}
The quantity $R'$ is an auxiliary rate associated with the sequence $(\sigma_n)_n$ and should not be confused with the target rate $D^\infty(\rho\|\mathcal S)$. In general, $R'\ge D^\infty(\rho\|\mathcal S)$, and the achievability statement aims to prove only $R^\varepsilon\ge D^\infty(\rho\|\mathcal S)$ rather than $R^\varepsilon\ge R'$.

\end{rem}

\begin{proof}[Proof of Lemma~\ref{lem:one-step}]
By definition, $D^\infty(\rho\|\mathcal S) \le R'$. For any fixed $\eta\in(0,\epsilon)$, define
\bb\label{eq:def_epsilon0}
  \epsilon_0
  \coloneqq
  \frac{\epsilon-\eta}{1-\epsilon}\left( R' - R^{\epsilon} \right) >0.
\ee
For all sufficiently large and fixed $m$, there is a state $\bar \sigma_m\in\mathcal S^{(m)}$ such that
\bb
  \frac1mD(\rho^{\otimes m}\|\bar \sigma_m)
  \le
  D^\infty(\rho\|\mathcal S)+\frac{\epsilon_0}{2}
  \le
  R'+\frac{\epsilon_0}{2}.
\ee
We can use Lemma~\ref{lem:full-rank-approx-S-m} to replace
$\bar \sigma_m$ by a full-rank state $\bar\sigma_{m,\delta}\in\mathcal S^{(m)}$, where $\delta$ is sufficiently small so that
\bb\label{eq:R'}
  \frac 1m D(\rho^{\otimes m}\|\bar \sigma_{m,\delta})
  \le
  R'+\epsilon_0.
\ee
Let us consider the padded block-product state $\bar\sigma_{n,\delta}^{(m)}\coloneqq\bar \sigma_{m,\delta}^{\otimes \ell}\otimes \sigma_{\rm full}^{\otimes r}$ and a worst alternative state
\(
  \displaystyle{\sigma_n^*
  \in
  \underset{\omega_n\in\mathcal S^{(n)}}{\operatorname{arg\,max}}\; 
  \beta_\epsilon(\rho_n\|\omega_n)}
\) -- which exists due to compactness of $\mathcal{S}^{(n)}$ --
in order to define
\bb
  \sigma_n'
  \coloneqq
  \frac 13
  \left(
    \sigma_n^* + \bar\sigma_{n,\delta}^{(m)} + \sigma_{\rm full}^{\otimes n}
  \right)\in\mathcal S^{(n)}.
\ee
 Let $\kappa\coloneqq-\log \big(\frac 13\lambda_{\min}(\sigma_{\rm full})\big)>0$, so that we can lower bound
\bb\label{eq:lb1}
  \sigma_n'
  \ge
  \frac 13 \sigma_{\rm full}^{\otimes n}
  \ge
  \frac 13 \lambda_{\min}(\sigma_{\rm full})^n \id\geq e^{-n\kappa}\id.
\ee

\paragraph{Step 1: Logarithmic rounding.} We now logarithmically round $\sigma_n'$ as in the proof of Lemma~\ref{lem:binary-rounding}, Eqs.~\eqref{eq:r1} and~\eqref{eq:log-runding} (see also~\cite[Lemma S8]{Hayashi2025}).
 Consider the spectral decomposition of $\sigma'_n$
 \bb
   \sigma_n'=\sum_i\lambda_{n,i}'P_{n,i}'
 \ee
 and, for $\lambda\in[e^{-\kappa n},1]$, define
\bb
  f_n(\lambda)
  \coloneqq
  \kappa\left\lceil
    \frac{\log\lambda +\kappa n}{\kappa}
  \right\rceil -\kappa n.
\ee
Recall that $f_n$ satisfies $\log\lambda\le f_n(\lambda)\le\log\lambda +\kappa$
and takes at most $n+1$ distinct values (see Figure~\ref{fig:f_n}). Setting
 \bb
   \widetilde \Sigma_n'
   \coloneqq
   \sum_i e^{f_n(\lambda_{n,i}')}P_{n,i}',
   \qquad \text{and}\qquad 
   \widetilde \sigma_n'
   \coloneqq
   \frac{\widetilde \Sigma_n'}{\Tr\widetilde \Sigma_n'},
 \ee
we have
 \bb\label{eq:ineq_kappa}
   e^{-\kappa} \sigma_n'
   \le
   \widetilde \sigma_n'
   \le
   e^{\kappa} \sigma_n'.
 \ee

\paragraph{Step 2: Reduction to a commuting pair via pinching.}
Eq.~\eqref{eq:ineq_kappa}, together with the operator monotonicity of the logarithm, gives
\bb
  \left|
  D(\rho_n\|\sigma_n') - D(\rho_n\|\widetilde\sigma_n')
  \right|
  \le \kappa.
\ee
Therefore,
\bb
  \liminf_{n\to\infty} \frac 1n D(\rho_n\|\sigma_n')
  =
  \liminf_{n\to\infty} \frac 1n D(\rho_n\|\widetilde\sigma_n').
  \label{eq:gap1}
\ee
Let $\pazocal P_n$ be the pinching map associated with the spectral projectors of $\widetilde\sigma_n'$. We have
\bb
  &\rel{D}{\rho_n}{\widetilde\sigma_n'} - \rel{D}{\pazocal P_n(\rho_n)}{\widetilde\sigma_n'}\\
  &\qquad=\Tr[\rho_n(\log\rho_n-\log\widetilde\sigma_n')]-\Tr\big[\pazocal P_n(\rho_n)\big(\log\pazocal P_n(\rho_n)-\log\widetilde\sigma_n'\big)\big]\\
  &\qquad=\Tr[\rho_n(\log\rho_n-\log\widetilde\sigma_n')]-\Tr\big[\rho_n\big(\pazocal P_n(\log\pazocal P_n(\rho_n))-\pazocal P_n(\log\widetilde\sigma_n')\big)\big]\\
  &\qquad=\Tr[\rho_n(\log\rho_n-\log\widetilde\sigma_n')]-\Tr\big[\rho_n\big(\log\pazocal P_n(\rho_n)- \log\widetilde\sigma_n'\big)\big]\\
  &\qquad=
  D(\rho_n\|\pazocal P_n(\rho_n))
\ee
where $0\leq D(\rho_n\|\pazocal P_n(\rho_n))\leq \log (n+1)$ due to the pinching inequality $\rho_n\le (n+1)\pazocal P_n(\rho_n)$~\cite{Hayashi2002}.
Together with~\eqref{eq:gap1}, we obtain
\bb\label{eq:step2}
  \liminf_{n\to\infty}
  \frac1n D(\rho_n\|\sigma_n')
  =
  \liminf_{n\to\infty}
  \frac1n\, \rel{D}{\pazocal P_n(\rho_n)}{\widetilde\sigma_n'}\, .
\ee

\paragraph{Step 3: Two testing exponent bounds and log-likelihood projectors.}
We have
\bb\label{eq:hp_bound1}
  \liminf_{n\to\infty}
  \frac 1n\,\rel{D_H^\epsilon}{\pazocal P_n(\rho_n)}{\widetilde\sigma_n'} &= \liminf_{n\to\infty}
  \frac 1n\,\rel{D_H^\epsilon}{\pazocal P_n(\rho_n)}{\pazocal P_n(\widetilde\sigma_n')} \\
  &\leqt{(i)}\liminf_{n\to\infty} \frac 1n D_H^\epsilon \left(\rho_n\|\widetilde\sigma_n'\right)\\
  &\leqt{(ii)}
  \liminf_{n\to\infty}\frac 1n D_H^\epsilon(\rho_n\|\sigma_n^*)\\
  &=\liminf_{n\to\infty}\frac 1n D_H^\epsilon(\rho_n\|\mathcal S^{(n)})
  = R^{\epsilon},
\ee
where the identities follow from the definitions of $\pazocal P_n$ and $\sigma_n^*$, while (i) is data-processing, and (ii) is a consequence of the first inequality in~\eqref{eq:ineq_kappa} combined with the definition of $\sigma_n'$:
\bb
  \widetilde\sigma_n'
  \geq e^{-\kappa}\sigma_n'
  \geq \frac{e^{-\kappa}}{3}\sigma_n^*.
\ee
Now, fix \modifica{$0<\epsilon_1<1$}. We have
\bb\label{eq:hp_bound2}
  \limsup_{n\to\infty}
  \frac 1n\,\rel{D_H^{1-\epsilon_1}}{\pazocal P_n(\rho_n)}{\widetilde\sigma_n'} 
  &\leqt{(iii)} \limsup_{n\to\infty}
  \frac 1n\,D_H^{1-\epsilon_1}
  \left(\rho_n\|\widetilde\sigma_n'\right)\\
  &\leqt{(iv)} \limsup_{n\to\infty}
  \frac 1n\,\rel{D_H^{1-\epsilon_1}}{\rho_n}{\bar\sigma_{n,\delta}^{(m)}} \\
  &\leqt{(v)}
  \frac 1m\, \rel{D}{\rho^{\otimes m}}{\modifica{\bar\sigma_{m,\delta}}}\\
  &\leqt{(vi)}
  R'+\epsilon_0\, ,
\ee
where (iii) is data-processing, (iv) follows from the inequality $\widetilde\sigma_n' \geq \frac 13 e^{-\kappa}\bar\sigma_{n,\delta}^{(m)}$, (v) is given by Proposition~\ref{prop:padded-block-product-strong-converse}, and \modifica{(vi)} is~\eqref{eq:R'}.
Let $\delta'>0$, and define
\bb
  \Pi_{n,1}
  &\coloneqq
  \left\{\pazocal P_n(\rho_n)\ge e^{n(R^{\epsilon}+\delta')}\widetilde\sigma_n'\right\},\\
  \Pi_{n,2}
  &\coloneqq
  \left\{\pazocal P_n(\rho_n)\ge e^{n(R'+\epsilon_0+\delta')}\widetilde\sigma_n'\right\}.
\ee
Using Lemma~\ref{lem:threshold-commuting}, we get the implications
\bb
 ~\eqref{eq:hp_bound1}\qquad &\implies \qquad\liminf_{n\to\infty}\Tr[ \pazocal P_n(\rho_n)\Pi_{n,1}]
  \le
  1-\epsilon\\
  ~\eqref{eq:hp_bound2}\qquad &\implies\qquad\limsup_{n\to\infty}\Tr[ \pazocal P_n(\rho_n)\Pi_{n,2}]
  \le
  \epsilon_1,
  \label{eq:step3-two-projs}
\ee
whence, by arbitrariness of $\epsilon_1>0$ \modifica{and by positivity of $\Tr[ \pazocal P_n(\rho_n)\Pi_{n,2}]$},
\bb\label{eq:epsilon1}
    \modifica{\lim_{n\to\infty}}\Tr[ \pazocal P_n(\rho_n)\Pi_{n,2}]=0
\ee
\begin{figure}[t]
  \centering
  \def\svgwidth{0.95\linewidth}
  \small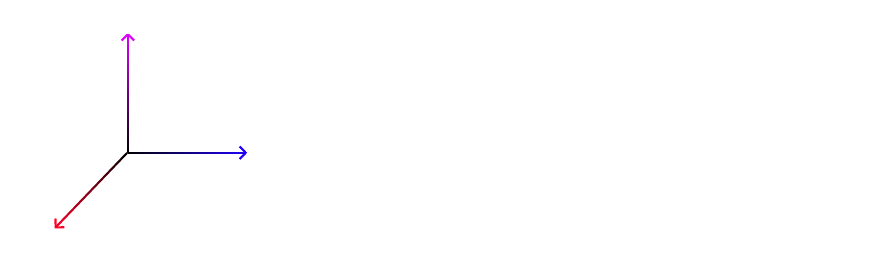
  \mycaption{The three-region decomposition of the Hilbert space $\mathcal{H}^{\otimes n}$.}{} 
  \label{fig:three}
\end{figure}
\paragraph{Step 4: Three-region decomposition.} By hypothesis, we have $R'> R^\epsilon$, whence $\Pi_{n,2}\le \Pi_{n,1}$. Define the three commuting projectors (see Figure~\ref{fig:three})
\bb
  E_{n,1}\coloneqq\id-\Pi_{n,1},
  \qquad
  E_{n,2}\coloneqq\Pi_{n,1}-\Pi_{n,2},
  \qquad
  E_{n,3}\coloneqq\Pi_{n,2}.
\ee
By the definition of $\Pi_{n,1}$ and $\Pi_{n,2}$,
\bb
  \frac1n E_{n,1}\left(\log \pazocal P_n(\rho_n)-\log \widetilde\sigma_n'\right)
  &\le
  (R^{\epsilon}+\delta')\,E_{n,1},\\
  \frac1n E_{n,2}\left(\log \pazocal P_n(\rho_n)-\log \widetilde\sigma_n'\right)
  &\le
  (R'+\epsilon_0+\delta')\,E_{n,2}.
\ee
For the third region, the full-rank component in $\sigma_n'$ gives a uniform lower bound on $\widetilde\sigma_n'$. Indeed,~\eqref{eq:ineq_kappa} combined with~\eqref{eq:lb1} gives
\bb
  \widetilde\sigma_n'\geq e^{-\kappa}\sigma_n' \geq e^{-(n+1)\kappa}\id\geq e^{-2n\kappa}\id,
\ee 
whence
\bb
  &\frac 1n E_{n,3}\left(\log \pazocal P_n(\rho_n)-\log \widetilde\sigma_n'\right)\\
  &\qquad =\frac 1n \underbrace{E_{n,3}\log \pazocal P_n(\rho_n)E_{n,3}}_{\leq 0}+\frac 1n E_{n,3}\,(\underbrace{-\log \widetilde\sigma_n'}_{\leq 2n\kappa \id})\,E_{n,3}\le 2\kappa E_{n,3}.
\ee
Using the decomposition of the space in these three regions, namely $\id=E_{n,1}+E_{n,2}+E_{n,3}$, we obtain
\bb
  &\frac1n D(\pazocal P_n(\rho_n)\|\widetilde\sigma_n')\\
  &\qquad = \frac 1n \Tr\big[\pazocal P_n(\rho_n)(E_{n,1}+ E_{n,2}+ E_{n,3})\big(\log \pazocal P_n(\rho_n)-\log \widetilde\sigma_n'\big)\big]\\
  &\qquad\le
  (R^{\epsilon}+\delta')\Tr[ \pazocal P_n(\rho_n)E_{n,1}] +
  (R'+\epsilon_0+\delta')\Tr[ \pazocal P_n(\rho_n)E_{n,2}] +
  2\kappa \Tr [\pazocal P_n(\rho_n)E_{n,3}] \\
  &\qquad = R^{\epsilon}+\delta'+(R'+\epsilon_0-R^{\epsilon})\Tr [\pazocal P_n(\rho_n)\Pi_{n,1}] +\bigl(2\kappa-(R'+\epsilon_0+\delta')\bigr)\Tr [\pazocal P_n(\rho_n)\Pi_{n,2}].
\ee
By Eq.~\eqref{eq:step3-two-projs}, there exists a subsequence $(n_j)_{j\geq1}$ such that 
\bb
\limsup_{j\to\infty}\Tr[\pazocal P_{n_j}(\rho_{n_j})\Pi_{n_j,1}]\leq1-\epsilon.
\ee
Along the same subsequence, Eq.~\eqref{eq:epsilon1} gives
$\Tr[\pazocal P_{n_j}(\rho_{n_j})\Pi_{n_j,2}]\to0$. Applying the previous three-region estimate along this subsequence, and using that the liminf of the full sequence is not larger than the limsup along any subsequence, yields
\bb
  \liminf_{n\to\infty}
  \frac1n D(\pazocal P_n(\rho_n)\|\widetilde\sigma_n')
  \leq
  R^{\epsilon}
  +\delta'
  +(1-\epsilon)(R'+\epsilon_0-R^{\epsilon}).
\ee
Letting $\delta'\to 0$ and using the definition~\eqref{eq:def_epsilon0} of $\epsilon_0$, we get
\bb
  \liminf_{n\to\infty}
  \frac 1n D(\pazocal P_n(\rho_n)\|\widetilde\sigma_n')
  -
  R^{\epsilon}
  &\leq  (1-\epsilon)(R'-R^\epsilon)+(1-\epsilon)\frac{\epsilon-\eta}{1-\epsilon}\left( R' - R^{\epsilon} \right)\\
  &=
  (1-\eta)(R'-R^{\epsilon}).
\ee
Finally, by~\eqref{eq:step2}, the previous upper bound gives
\bb
  \liminf_{n\to\infty} \frac 1n D(\rho_n\|\sigma_n') - R^{\epsilon}
  \le
  (1-\eta)(R'-R^{\epsilon}),
\ee
completing the proof of Lemma~\ref{lem:one-step}.
\end{proof}

\begin{boxedstep}{}
\begin{prop}[(Achievability bound)]\label{prop:W1-achievability}
For every $\epsilon\in(0,1)$,
\bb
  \liminf_{n\to\infty}
  \frac1n\,\rel{D_H^\epsilon}{\rho_n}{\mathcal S^{(n)}}
  \ge
  D^\infty(\rho\|\mathcal S).
\ee
\end{prop}
\end{boxedstep}

\begin{proof}
Suppose, by contradiction, that
\bb
  R^{\epsilon}= \liminf_{n\to\infty}\frac 1n\,\rel{D_H^\epsilon}{\rho_n}{\mathcal S^{(n)}} < D^\infty(\rho\|\mathcal S)\, .
\ee
For each $n$, choose $\displaystyle{\sigma_n^*\in
  \underset{\sigma_n\in\mathcal S^{(n)}}{\operatorname{arg\,min}}\; 
  D(\rho^{\otimes n}\|\sigma_n)}$, so that
\bb
  R'= \liminf_{n\to\infty} \frac 1n D(\rho^{\otimes n}\|\sigma_n^*)
  =D^\infty(\rho\|\mathcal S)>R^{\epsilon}.
\ee
Fixing
$\eta \in (0,\epsilon)$ and applying Lemma~\ref{lem:one-step}, we have
\bb
  \liminf_{n\to\infty}
  \frac1n D(\rho_n\|\sigma_n')
  &\le
  R^{\epsilon}
  +(1-\eta)
  \left(D^\infty(\rho\|\mathcal S)-R^{\epsilon}\right)\\
  &= 
  D^\infty(\rho\|\mathcal S) -\eta\left(D^\infty(\rho\|\mathcal S)-R^\epsilon\right) \\
  &<
  D^\infty(\rho\|\mathcal S).
\ee
Since $ D(\rho_n\|\mathcal S^{(n)}) \le D(\rho_n\|\sigma_n')$, we conclude that
\bb\label{eq:contr}
  \liminf_{n\to\infty} \frac1n D(\rho_n\|\mathcal S^{(n)})
  <
  D^\infty(\rho\|\mathcal S).
\ee
On the other hand, by Assumption~\ref{ass:standing-Sn}, we can apply the continuity of the relative entropy of resource in Wasserstein distance (Proposition~\ref{prop:3.1}), obtaining
\bb
  \lim_{n\to \infty}\frac1n\,\rel{D}{\rho_n}{\mathcal S^{(n)}}
  =
  \lim_{n\to \infty} \frac1n\, \rel{D}{\rho^{\otimes n}}{\mathcal S^{(n)}} = \modifica{D^\infty(\rho\|\mathcal S)},
\ee
which contradicts~\eqref{eq:contr}.
\end{proof}

\subsection{Second act. Universality for equiconvergent sources}\label{sec:universality}

We finally prove the existence of a universal sequence of \modifica{tests} for equiconvergent sources. The key idea is that the composite null problem can be reduced to the worst individual null state by a minimax identity.

Let $(\mathcal F^{(n)})_{n\ge1}$ be a Wasserstein equiconvergent source along $\rho$, namely 
\bb
\lim_{n\to\infty}
\sup_{\rho_n\in\mathcal F^{(n)}}\frac 1n\left\|\rho_n-\rho^{\otimes n}\right\|_{W_1}=0,
\ee
where each $\mathcal F^{(n)}$ is assumed to be compact and convex by hypothesis. For each $n$, choose $\rho_n^*
  \in
  \underset{\rho_n\in \mathcal F^{(n)}}{\operatorname{arg\,max}}\; 
  \beta_\epsilon(\rho_n\|\mathcal S^{(n)})$, so that $ D_H^\epsilon(\rho_n^*\|\mathcal S^{(n)})= \rel{D_H^\epsilon}{\mathcal F^{(n)}}{\mathcal S^{(n)}}$. 
Then,
\bb
 \lim_{n\to\infty}\frac 1n \left\|\rho_n^*-\rho^{\otimes n}\right\|_{W_1}
  \le
  \lim_{n\to\infty}\sup_{\rho_n\in \mathcal F^{(n)}}
  \frac 1n \left\|\rho_n-\rho^{\otimes n}\right\|_{W_1}
  =0.
\ee
Thus $(\rho_n^*)_n$ is an individual Wasserstein almost i.i.d.\ source along $\rho$. Applying the individual-source statement proved above, we immediately obtain
\bb
  \lim_{n\to\infty}
  \frac1n\, \rel{D_H^\epsilon}{\mathcal F^{(n)}}{\mathcal S^{(n)}} = \lim_{n\to\infty}
  \frac1n\,\rel{D_H^\epsilon}{\rho_n^*}{\mathcal S^{(n)}}
  =
  D^\infty(\rho\|\mathcal S)\, .
\ee
This proves the existence of a sequence of POVMs $(E_n)_{n\geq 1}$ that is universal within the sequence of families $\mathcal{F}=(\mathcal{F}^{(n)})_{n\geq 1}$, completing the proof of Theorem~\ref{thm:W1_GQSL}.

\subsection*{Acknowledgements} FG, K.-Y.L and LL acknowledge financial support from the European Union (ERC StG ETQO, Grant Agreement no.\ 101165230).

\begin{note}
The strategy and the key tools to prove Theorem~\ref{thm:W1_GQSL} were identified by the authors and techinically developed with the support of AI (Microsoft Copilot). The authors improved, simplified and carefully proofread the results. The applications in Section~\ref{sec:applications} were conceived and proved by the authors without the assistance of LLMs. ChatGPT was used for identifying imprecisions and typos.
\end{note}

\bibliography{biblio}

\appendix

\section{Alternative proof of Theorem~\ref{thm:iid}}\label{proof:alternative}

     Under the additional assumption that the sequence $(\mathcal S^{(n)})_{n\geq 1}$ is closed under the permutation twirl, the a classical formulation of the weak quasi-concavity of the hypothesis testing relative entropy -- which replaces the quantum formulation in Lemma~\ref{lem:weak_q_conc} -- can be used to give a simpler proof of Theorem~\ref{thm:iid}. For the convenience of the reader, here we state the result we are going to prove in this section.

    \begin{thm}\label{thm:iid2} Suppose that the sequence of families $(\mathcal S^{(n)})_{n\geq 1}$ satisfy Assumption~\ref{ass:standing-Sn}, and, in addition, suppose every family is closed under the permutation twirl, namely
    \bb
        \frac{1}{n!}\sum_{\pi\in S_n} U_\pi^{\vphantom{\dagger}} \sigma_n U^\dagger_\pi \in \mathcal{S}^{(n)}\qquad \forall \sigma_n \in \mathcal{S}^{(n)}.
    \ee
    Then, for any arbitrary subset of states $\pazocal R_1\subseteq \mathcal{D}(\mathcal{H})$, we have 
        \bb
            \lim_{n\to\infty} \frac1n\, \rel{D_H^\epsilon}{\co\!\big(\pazocal{R}_n^{\rm i.i.d.}\big)}{\mathcal{S}^{(n)}} = \inf_{\rho\in \pazocal R_1} D^\infty(\rho \|\mathcal{S})\qquad \forall \,\epsilon\in (0,1).
        \ee
    \end{thm}

\begin{proof}
    The proof follows exactly the same structure we presented  for Theorem~\ref{thm:iid} above (Section~\ref{sec:proof_iid}), where Lemma~\ref{lem:ach2b_iid} is now replaced by Lemma~\ref{lem:ach1a_iid}, which we state and prove below.
\end{proof}

\begin{lemma}\label{lem:ach1a_iid} Suppose that the sequence of families $(\mathcal S^{(n)})_{n\geq 1}$ satisfies Assumption~\ref{ass:standing-Sn} and it is closed under the permutation twirl. Let $\pazocal R_1\subseteq \mathcal{D}(\mathcal{H})$ be an arbitrary subset of states. Then,
      \bb 
      \liminf_{n\to\infty}\frac 1n \,\rel{D_H^\epsilon}{\co\!\big(\pazocal{R}_n^{\rm i.i.d.}\big)}{\mathcal S^{(n)}} \geq \lim_{\epsilon'\to 0}\liminf_{n\to\infty}\frac 1n \inf_{\rho \in\pazocal R_1} \rel{D_H^{\epsilon'}}{\rho^{\otimes n}}{\mathcal S^{(n)}}\, .
        \ee
        for all $\epsilon\in (0,1)$.
    \end{lemma}

Differently from Lemma~\ref{lem:ach2_iid}, now we only need a simple and fully classical weak quasi-concavity result, due to~\cite{quasi-concavity}.

\begin{lemma}[{(Classical weak quasi-concavity of $D_H^\epsilon$~\cite{quasi-concavity})}]\label{lem:wqc_class} Let $\{t_i\}_{i\in[N]}$ be a probability distribution on $[N]$ and let $\mathcal X$ be a discrete space. Then, for all families $\{p_i\}_{i\in[N]}\subseteq \mathcal{P}(\mathcal{X})$ and $q\in \mathcal{P}(\mathcal{X})$, we have
\bb
    D_H^\epsilon (p\|q)\geq \min_{1\leq i\leq N}D_H^\epsilon-\log N,
\ee
where $p\coloneqq\sum_{i=1}^Nt_ip_i$ and $\epsilon\in (0,1)$ is arbitrary.
    
\end{lemma}

\begin{proof}
Let $i\in[N]$ and let $A_i:\mathcal{X}\to [0,1]$ an optimal acceptance function in the definition of $D_H^\epsilon(p_i\|q)$. Define the acceptance function $A:\mathcal{X}\to [0,1]$ as the  pointwise maximum of the individual acceptance functions, namely,
\bb
    A \coloneqq \max_i A_i \geq A_i \qquad \forall i\in [N].
\ee
Then, the type I error probability under $A$ for the hypothesis testing problem $p$ vs $q$ satisfies
\bb
    \sum_{x\in\mathcal X}p(x)(1-A(x))=\sum_{x\in\mathcal X}\sum_{i=1}^Nt_ip_i(x)(1-A(x))\leq \sum_{i=1}^Nt_i\sum_{x\in\mathcal X}p_i(x)(1-A_i(x))\leq \epsilon,
\ee
while the type II error is bounded as
\bb
    \sum_{x\in\mathcal X}q(x)A(x)\leq \sum_{x\in\mathcal X}q(x)\sum_{i=1}^NA_i(x)\leq N \max_{1\leq i\leq N}\sum_{x\in\mathcal X}q(x)A_i(x),
\ee
which immediately implies the claim.
\end{proof}

Now we have all the ingredients to prove Lemma~\ref{lem:ach1a_iid}.

    \begin{proof}[Proof of Lemma~\ref{lem:ach1a_iid}.]
        
Let $(\rho^{(n)})_{n\geq 1}$ and $(\sigma_n)_{n\geq 1}$ be two sequences of states $\rho^{(n)}\in\co\!\big(\pazocal{R}_n^{\rm i.i.d.}\big)$ and $\sigma_n\in\mathcal{S}^{(n)}$ such that 
\bb\label{eq:quasi_min}
     \rel{D_H^\epsilon}{\co\!\big(\pazocal{R}_n^{\rm i.i.d.}\big)}{\mathcal{S}^{(n)}}\geq \rel{D_H^\epsilon}{\rho^{(n)}}{\sigma_n}-1.
\ee
Since the elements of $\pazocal{R}_n^{\rm i.i.d.}$ are permutation invariant and $\pazocal{R}_n^{\rm i.i.d.}$ is closed under the permutation twirl, by the data-processing inequality we can assume without loss of generality that $\sigma_n$ is permutation invariant.
The real vector space $H_{d,n}^{\rm sym}$ of permutationally symmetric Hermitian operators on $\mathcal{H}^{\otimes n}\simeq \big(\mathbb{C}^d\big)^{\otimes n}$, has dimension upper bounded as $\dim H_{d,n}^{\rm sym}\leq (n+1)^{d^2-1}$~\cite{Hayashi2016-bh}.
 By Carath\'eodory's theorem, since $\rho^{(n)}\in H_{d,n}^{\rm sym}$, we can write it as a convex combination of at most $N=(n+1)^{d^2-1}+1={\rm poly}(n)$ terms of the form $\rho_i^{\otimes n}$, where $\rho_i\in\pazocal R_1$, namely
 \bb
    \rho^{(n)}=\sum_{i=1}^Np_i\rho_i^{\otimes n}, \qquad 0\leq p_i\leq 1.
 \ee
Let us write $\sigma_n$ in the Schur--Weyl basis (see, e.g.,~\cite{Hayashi2016-bh}), and then let us consider the diagonalisation of each block:
\bb
\sigma_n=\bigoplus_{\lambda\in\pazocal{Y}_n^d} \sigma_n^\lambda \otimes\id_{\pazocal{V}_\lambda},\qquad \sigma_n^\lambda=\sum_{k\in K_n^  \lambda}\mu_n^{\lambda,k}\Pi_n^{\lambda,k} \in \mathcal{L}(\pazocal{U}_\lambda), \qquad |K_n^\lambda|\leq \dim \pazocal{U}_\lambda,
\ee
where the index $\lambda$ ranges on the set $\pazocal{Y}_d^n$ of Young diagrams with size $n$ and depth at most $d$; $\pazocal{U}_\lambda$ and $\pazocal{V}_\lambda$ are irreducible representations of the special unitary group ${\rm SU}(d)$ and of the symmetric group $S_n$, respectively. Let $\pazocal P$ be the pinching channel with projectors
\bb
    P=\bigcup_{\lambda\in\pazocal{Y}_n^d}\left\{\left(\Pi_n^{\lambda,k}\otimes \id_{\pazocal V_\lambda}\right)\oplus 0_{(\pazocal U^\lambda\pazocal \ludo{\otimes} V^\lambda)^\perp}\right\}_{k\in K^\lambda_n},
\ee
i.e.
\bb
    \pazocal P(\,\cdot\,) \coloneqq\sum_{\Pi\in P}\Pi\, \cdot\,\Pi.
\ee
Note that each projector $P\in \Pi$ acts non-trivially on one particular $\lambda$-sector $\pazocal{U}^\lambda \otimes \pazocal{V}^\lambda$, and vanishes on the orthogonal space $\big(\pazocal U^\lambda\pazocal \otimes V^\lambda\big)^\perp$. The action of the pinching channel $\pazocal{P}$ is however stronger than that of the permutation twirl, because, inside each $\lambda$-sector, a further pinching is applied, in order to force the output to commute with $\sigma_n$. By definition $\pazocal P(\sigma_n)=\sigma_n$, and,
since $\dim \pazocal{U}_\lambda\leq (n+1)^{d(d-1)/2}$ and $|\pazocal{Y}_n^d|\leq (n+1)^{d-1}$, we can upper bound $|P|\leq (n+1)^{d-1}\times (n+1)^{d(d-1)/2} \leq (n+1)^{d^2-1}={\rm poly}(n)$. Then,

\begin{align}\label{eq:lower_bounds}
\rel{D_H^\epsilon}{\rho^{(n)}}{\sigma_n} &= \Rel{D_H^\epsilon}{\sumno_{i=1}^N p_i \rho_i^{\otimes n}}{\sigma_n} \notag \\
&\geqt{(i)} \Rel{D_H^\epsilon}{\sumno_{i=1}^N p_i\pazocal P(\rho_i^{\otimes n})}{\sigma_n} \notag \\
&\geqt{(ii)} \min_i \rel{D_H^\epsilon}{\pazocal P(\rho_i^{\otimes n})}{\sigma_n} - \log N \notag \\
&\geqt{(iii)} \min_i \rel{D_{\max}^{\sqrt{1-\epsilon^2}}}{\pazocal P(\rho_i^{\otimes n})}{\sigma_n} - \log \frac{1-\epsilon}N\\
&\geqt{(iv)} \min_i \rel{D_{\max}^{\sqrt{1-\epsilon^2/2},P}}{\rho_i^{\otimes n}}{\sigma_n} - \log {\rm poly}_{\epsilon,d}(n) - \log|P| \notag \\
&\geqt{(v)} \min_i \rel{D_H^{\epsilon^2/4}}{\rho_i^{\otimes n}}{\sigma_n} - \log {\rm poly}'_{\epsilon,d}(n) \notag \\
&\geq \inf_{\rho \in\pazocal R_1} \rel{D_H^{\epsilon^2/4}}{\rho^{\otimes n}}{ \mathcal{S}^{(n)}} - \log {\rm poly}'_{\epsilon,d}(n), \notag
\end{align}
where (i) is the data-processing inequality under the action of $\pazocal P$, (ii) is a direct application of Lemma~\ref{lem:wqc_class}, and (iii) follows from the weak/strong-converse duality (Lemma~\ref{lemma:wsc}); in (iv) we change the smoothing from trace distance to purified distance as in~\eqref{eq:Dmax_ineq}, and then we apply the Li--Yao pinching inequality~\cite[Proposition~13]{Li2024}; finally, (v) is given by~\eqref{eq:Dmax_ineq}, followed again by the weak/strong-converse duality (where we have used the elementary inequality $\sqrt{1-x}< 1-x/2$ for $x\in(0,1)$ and the monotonicity of $\epsilon\mapsto D_H^\epsilon$). Then,  by definition of $\rho^{(n)}$ and $\sigma_n$,~\eqref{eq:lower_bounds} yields
\bb\label{eq:previous}
\liminf_{n\to\infty}\frac 1n \,\rel{D_H^\epsilon}{\co\!\big(\pazocal{R}_n^{\rm i.i.d.}\big)}{\mathcal S^{(n)}} &\geq \liminf_{n\to\infty}\frac 1n\,\rel{D_H^\epsilon}{\rho^{(n)}}{\sigma_n}\\
&\geq\liminf_{n\to\infty}\frac 1n  \inf_{\rho \in\pazocal R_1} \rel{D_H^{\epsilon^2/16}}{\rho^{\otimes n}}{\mathcal S^{(n)}}\\
&\geq \lim_{\epsilon'\to 0}\liminf_{n\to\infty}\frac 1n \inf_{\rho \in\pazocal R_1} \rel{D_H^{\epsilon'}}{\rho^{\otimes n}}{\mathcal S^{(n)}}\, ,
\ee
which completes the proof of Lemma~\ref{lem:ach1a_iid}.
\end{proof}
    
\end{document}